%% file: main.tex
\documentclass[%
 reprint,
 floatfix,
 amsmath,amssymb,
 aps,
]{revtex4-2}

\usepackage{graphicx}
\usepackage{xcolor}
\graphicspath{{figures/}}
\usepackage{bm}
\usepackage{braket}
\usepackage{amsthm}
\usepackage{needspace}
\usepackage{placeins}

\newtheorem{theorem}{Theorem}
\newtheorem{proposition}{Proposition}
\newtheorem{lemma}{Lemma}
\newtheorem{corollary}{Corollary}
\newtheoremstyle{assumptionstyle}
  {4pt plus 1pt minus 1pt}
  {4pt plus 1pt minus 1pt}
  {\itshape}
  {}
  {\bfseries}
  {.}
  {\newline}
  {\thmname{#1}\thmnumber{ #2}\thmnote{ \textnormal{(#3)}}}
\theoremstyle{assumptionstyle}
\newtheorem{assumption}{Assumption}
\theoremstyle{plain}
\newtheorem*{theoremdatadepformal}{Theorem~\ref{thm:datadep-bp-main}$'$}
\newtheorem*{theoremgradformal}{Theorem~\ref{thm:grad-suppress-main}$'$}

\begin{document}

\preprint{APS/123-QED}

\title{Trainability and Mode Separation of Mixed IQP-QCBMs}

\author{Youngseok Lee}
\email{ys\_lee@norma.co.kr}
\author{Hyunwoo Kim}
\email{hw\_kim@norma.co.kr}
\affiliation{%
Quantum AI Team, NORMA Inc., Republic of Korea
}%

\date{\today}

\begin{abstract}
Quantum circuit Born machines (QCBMs) based on instantaneous quantum
polynomial-time (IQP) circuits are promising quantum generative models for
their classical trainability. It is known that their ancilla-free form
avoids barren plateaus under certain initializations, but remains
non-universal. Although adding ancilla qubits raises the expressivity,
whether the ancilla-extended model retains local trainability remains
unknown. We propose the mixed IQP-QCBM, which generalizes the
ancilla-extended circuit as a weighted mixture of ancilla-free IQP circuits,
called branches. For a polynomial number of branches, we prove local
barren-plateau avoidance from data-agnostic and, under certain assumptions,
data-dependent initializations. We further show that the mixed IQP-QCBM can
surpass the best ancilla-free IQP circuit only if its branches generate a
number of distinct distributions. In particular, we focus on a behavior we
call \emph{mode separation}, in which each branch captures a particular
feature of the target. Mode separation is hard to attain from an
initialization whose branches generate the same distribution: the gradients
that would separate them are suppressed while the distributions they
generate remain close. This motivates \emph{cluster initialization}, which
assigns a different unsupervised data cluster to each branch and provides an
initial degree of mode separation. Exact calculations on two 16-bit datasets
support the barren-plateau and gradient-suppression claims. On four
benchmarks, binary clusters, a two-dimensional Ising model, binarized MNIST,
and a 484-spin glass, cluster initialization converges fastest and reaches
the lowest mean test $\mathrm{MMD}^2$. We observe that, when achieving the
lowest test $\mathrm{MMD}^2$, the mixed IQP-QCBM contains branches
specialized to distinguishable data features such as blob patterns,
magnetization sectors, or digit shapes.
\end{abstract}

\maketitle


\section{\label{sec:intro}Introduction}

Quantum generative models, including the quantum circuit Born machine
(QCBM)~\cite{benedetti2019,liu2018,coyle2020}, are studied as a route to
quantum advantage because quantum devices natively sample distributions that
can be classically intractable~\cite{coyle2020,bremner2011}; their
theoretical foundations beyond sampling hardness, such as explicit
generalization bounds, are also under study~\cite{du2022power}. Their
training is costly when losses and gradients must be evaluated on quantum
hardware at every optimization step. Random initialization can also produce
barren plateaus, where the variance of loss gradients decays exponentially with system size~\cite{mcclean2018,cerezo2021}.

The IQP-QCBM, built from an \emph{instantaneous quantum polynomial-time} (IQP)
circuit, supports classical training while retaining potentially hard
sampling. Its loss and gradients are written in Pauli-$Z$ correlators that
admit efficient classical estimates~\cite{vandennest2011,recio2025}, so quantum
hardware is not required during optimization. For sampling, in contrast, no
efficient classical algorithm is expected: exact classical sampling of a
general IQP circuit would collapse the polynomial hierarchy to its third
level~\cite{bremner2011}, and the hardness extends to approximate sampling
under additional conjectures~\cite{bremner2016,bremner2025stabilizer}.
Quantum hardware enters only to draw samples from the trained model.

The existing constructions of IQP-QCBMs, however, have limitations. The ancilla-free IQP-QCBM admits initializations that avoid
barren
plateaus~\cite{rudolph2024barriers,lerch2026,characterizing_iqp,deluca2026}
but is not universal as a probability
model~\cite{kurkin2025universality,kurkin2025note}. Adding ancilla qubits
enlarges the set of probability distributions that the QCBM can
represent~\cite{zhong2024mbl,wiebe2019hidden}, but no trainability guarantee is
known for the ancilla-extended circuit.

We introduce the \emph{mixed IQP-QCBM} to address this gap. The model combines the output
distributions of ancilla-free IQP circuits, called \emph{branches}, in a
weighted mixture. At uniform weights, this mixture is exactly the system
distribution obtained from an IQP circuit with unmeasured
ancillas~\cite{slim2026}. The branch representation turns the problem into
two concrete questions: whether the mixture does not suffer a local barren plateau, and whether its branches can learn different parts of
the target distribution.

\begin{figure*}[!t]
\includegraphics[width=\textwidth]{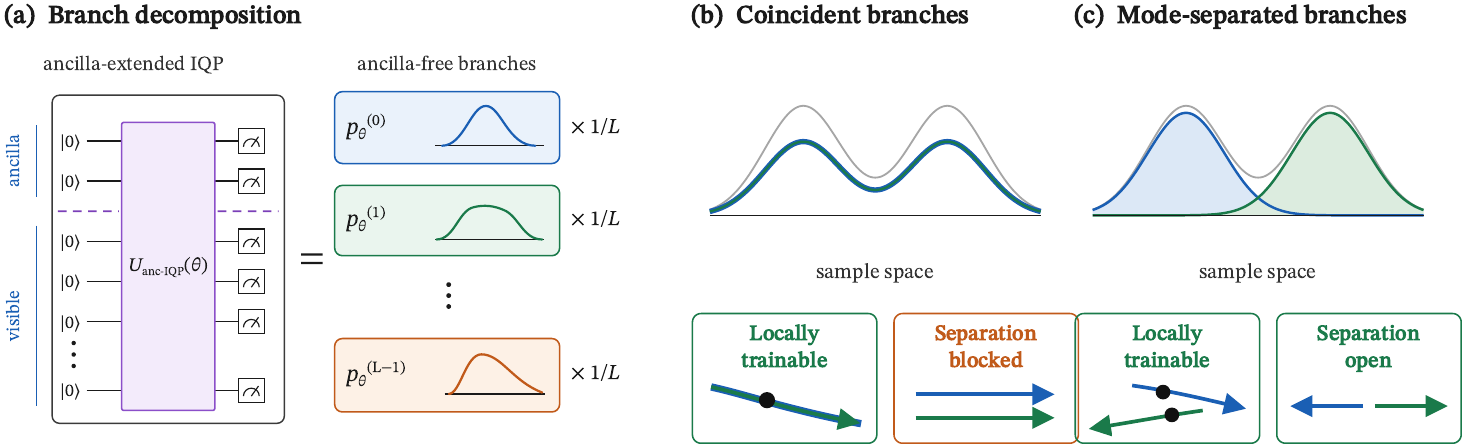}
\caption{\label{fig:overview} Central mechanism of the mixed IQP-QCBM.
Gray denotes the target, blue and green denote two branches, and arrows show
their responses to training. Under the corresponding conditions of Sec.~III,
each initialization shown is locally trainable. (a) At uniform weights,
tracing out the ancillas of $U_{\text{anc-IQP}}(\theta)$ yields $L$
ancilla-free branches, each weighted by $1/L$. (b) Data-agnostic and global
initialization are branch-coincident, so their separating gradient vanishes.
(c) Cluster initialization assigns different data clusters to the branches,
giving distinct starting distributions and noncanceling separating gradients.}
\end{figure*}

For the first question we identify locally trainable starting points for the
mixed IQP-QCBM under each of the three initialization schemes, extending guarantees for
ancilla-free IQP
circuits~\cite{rudolph2024barriers,lerch2026,characterizing_iqp,deluca2026}.
At the data-agnostic initialization, the loss curvature follows an exact $1/L^2$ behavior
and remains inverse-polynomial for a polynomial number of branches. Under
explicit conditions on the assigned data groups and branch weights, the global
and cluster initializations are locally trainable as well.

For the second question we prove a necessary condition: the mixture improves
on the MMD of the best ancilla-free IQP circuit only if its branches represent different
distributions. We call the target-aligned form of this diversity \emph{mode
separation}: different branches model
different modes or regions of the target. Creating mode separation from
identical branches is difficult because the branch-separating gradients vanish
at exact coincidence and remain suppressed nearby. Cluster initialization
instead assigns a separate unsupervised data cluster to each branch. Under the
data-dependent conditions, the resulting start is locally trainable and already
branch-separated.
Figure~\ref{fig:overview} summarizes this mechanism.

Exact calculations on
the two $n=16$ datasets reproduce the predicted curvatures, the $1/L^2$
behavior at the data-agnostic initialization and the leading positive term at
the global and cluster initializations, and show that the branch-separating
gradient vanishes at coincidence and grows linearly with the branch
spread. Across binary clusters,
a two-dimensional Ising model, binarized MNIST, and a $484$-spin glass,
cluster initialization converges fastest and reaches the lowest mean test
maximum mean discrepancy (MMD$^2$), or matches it within the observed seed
variation, on every benchmark. Its
advantage is clearest when reaching the target requires a large change in the
angles that couple the ancillas to the system, the parameters that make the
branches differ. In the mixtures with the lowest
test MMD$^2$, different branches characterize different blob patterns,
magnetization sectors, or digit shapes.

The remainder of the paper is organized as follows. Section~\ref{sec:model}
defines the mixed IQP-QCBM and its initializations.
Sections~\ref{sec:trainability} and \ref{sec:usable} analyze local trainability
and the branch diversity required for added expressivity.
Section~\ref{sec:result} tests these results numerically and examines mode
separation after training. Sections~\ref{sec:discussion} and
\ref{sec:conclusion} give the discussion and conclusion.

\section{\label{sec:prelim}\label{sec:model}\label{sec:framework}Mixed IQP-QCBM}

This section assembles the framework the rest of the paper relies on: the
model, its classically estimable correlators, the training loss, the branch
decomposition, and the initializations whose trainability and performance
the later sections analyze. We adopt the setup and notation of
Ref.~\cite{lerch2026} for the ancilla-free IQP-QCBM, its Pauli-$Z$
correlators, and the MMD loss expressed through them. 

\subsection{\label{sec:iqp-qcbm}The IQP-QCBM and its classical correlators}

The generative-modeling task is to find parameters $\bm{\theta}^{*}$ at
which a parametrized model distribution reproduces a given target
distribution over bit strings,
$\mathcal{P}_{\bm{\theta}^{*}}\approx\mathcal{P}_{\text{data}}$. Here the
target $\mathcal{P}_{\text{data}}=\{p_{\text{data}}(\bm{z})\}_{\bm{z}}$
assigns the probability $p_{\text{data}}(\bm{z})$ to each bit string
$\bm{z}\in\{0,1\}^n$, and the model
$\mathcal{P}_{\bm{\theta}}=\{p_{\bm{\theta}}(\bm{z})\}_{\bm{z}}$ is
parameterized by a set of trainable parameters $\bm{\theta}$.

A quantum state is a natural for this task: a measurement in the
computational basis returns a bit string distributed according to the Born rule.
A \emph{quantum circuit Born machine} (QCBM) builds its model distribution on
exactly this mechanism, encoding the distribution in the measurement statistics
of a parametrized state~\cite{benedetti2019,liu2018,coyle2020}:
\begin{equation}
\ket{\psi(\bm{\theta})}:= U(\bm{\theta})\ket{0},
\label{eq:state}
\end{equation}
where $U(\bm{\theta})$ is a parametrized quantum circuit acting on the all-zero
initial state $\ket{0}$. Measuring this state in the computational basis yields,
through Born's rule, the model distribution
\begin{equation}
p_{\bm{\theta}}(\bm{z}):=|\braket{\bm{z}|\psi(\bm{\theta})}|^2,
\label{eq:born}
\end{equation}
from which the model takes its name~\cite{cheng2018,han2018}.

In this work, the circuit $U(\bm{\theta})$ is an \emph{instantaneous quantum
polynomial-time} (IQP) circuit~\cite{shepherd2009}. A general IQP circuit
applies commuting rotations generated by arbitrary products of Pauli-$X$
operators. Throughout, we call an object \emph{$k$-body} when it involves
$k$ qubits. We use the low-body member of the IQP family whose generators act on at most two
qubits~\cite{lerch2026}. Acting on the all-zero state
$\ket{0}^{\otimes n}$, this $2$-body IQP-QCBM applies
\begin{equation}
U(\bm{\theta}):=\exp\!\left(i\sum_{j=1}^{n}\theta_j X_j
 + i\!\!\sum_{(j,k)\in E}\!\!\theta_{jk}X_jX_k\right),
\label{eq:iqp}
\end{equation}
where $E\subseteq\{(j,k):j<k\}$ is the edge set of the interaction graph.
Below, $G$ denotes the one- or two-qubit support of an IQP generator, and
$X_G$ and $\theta_G$ denote its Pauli word and angle.

The quantities estimated during training are the Pauli-$Z$
\emph{correlators}. For a subset
$A\subseteq[n]:=\{1,\dots,n\}$, define
\begin{equation}
z_A(\bm{\theta}):=\bra{\psi(\bm{\theta})}Z_A\ket{\psi(\bm{\theta})},
\qquad Z_A:=\bigotimes_{j\in A}Z_j.
\label{eq:zA}
\end{equation}
These expectation values can be estimated without executing the circuit.
Conjugating $Z_A$ by a rotation $e^{i\theta_GX_G}$ depends only on how many
qubits $G$ and $A$ share. If $|G\cap A|$ is even, $X_G$ commutes with
$Z_A$ and the rotation cancels; if it is odd, they anticommute and the
rotation contributes the doubled angle $2\theta_G$. We denote this parity by
$G\cdot A:=|G\cap A|\bmod 2$.

For an $X$-basis bit string $\bm{s}$, the generator $X_G$ has eigenvalue
$(-1)^{G\cdot\bm{s}}$, where
$G\cdot\bm{s}:=\sum_{j\in G}s_j\bmod 2$. The initial state
$\ket{0}^{\otimes n}$ assigns equal weight to all such strings. The
correlator is therefore the uniform average~\cite{vandennest2011}
\begin{equation}
z_A(\bm{\theta})
=\mathbb{E}_{\bm{s}}
\cos\!\Big(\sum_{G}\theta_G\big(1-(-1)^{G\cdot A}\big)(-1)^{G\cdot\bm{s}}\Big),
\label{eq:zA-classical}
\end{equation}
with $\bm{s}$ sampled uniformly from $\{0,1\}^n$. Because the cosine is
bounded, averaging it over $M$ random strings gives Monte Carlo error
$O(M^{-1/2})$, independent of $n$, and evaluating one string requires one
pass over the generators (Appendix~\ref{app:correlators}). 

\subsection{\label{sec:mmd}Low-body MMD loss}

The training loss is the maximum mean discrepancy
(MMD)~\cite{gretton2012,liu2018,muandet2017,rudolph2024barriers}: a
distribution is condensed into the mean of a feature map $\phi$ over its
samples, and the MMD is the feature-space distance between two means.
The feature map enters only through the kernel
$k(\bm{x},\bm{y})=\phi(\bm{x})^{\!\top}\phi(\bm{y})$, so the squared MMD
reduces to kernel expectations over samples,
\begin{align}
\mathrm{MMD}^2 :={}&
 \mathbb{E}_{\bm{x},\bm{x}'\sim p_{\bm{\theta}}}[k(\bm{x},\bm{x}')]
 +\mathbb{E}_{\bm{y},\bm{y}'\sim p_{\text{data}}}[k(\bm{y},\bm{y}')]
 \nonumber\\
 &-2\,\mathbb{E}_{\bm{x}\sim p_{\bm{\theta}},\,\bm{y}\sim p_{\text{data}}}
 [k(\bm{x},\bm{y})].
\label{eq:mmd-kernel}
\end{align}
For binary data we use the Hamming Gaussian kernel
$k(\bm{x},\bm{y})=\exp[-d_H(\bm{x},\bm{y})/(2\sigma^2)]$, with
$d_H(\bm{x},\bm{y})=\sum_{j=1}^{n}|x_j-y_j|$ the Hamming distance and
$\sigma$ the bandwidth~\cite{liu2018,coyle2020}. This kernel is
characteristic, so $\mathrm{MMD}^2=0$ if
and only if $p_{\bm{\theta}}=p_{\text{data}}$~\cite{gretton2012,sriperumbudur2010}.

Training uses a second, exact form of the same loss. For the Hamming
Gaussian kernel, Eq.~\eqref{eq:mmd-kernel} can be rewritten as a weighted sum
of squared Pauli-$Z$ correlator mismatches~\cite{lerch2026,rudolph2024barriers},
\begin{equation}
\mathcal{L}(\bm{\theta}):=\mathrm{MMD}^2(p_{\text{data}},p_{\bm{\theta}})
=\sum_{A\subseteq[n]}w_A\big(z_A(\bm{\theta})-t_A\big)^2,
\label{eq:mmd}
\end{equation}

The analysis of the following sections works directly with this correlator
form, so we now record its parts in detail.
Equation~\eqref{eq:mmd} has three ingredients: the model
correlator $z_A$ of Eq.~\eqref{eq:zA}, the matching data \emph{moment}
\begin{equation}
t_A:=\mathbb{E}_{\bm{y}\sim p_{\text{data}}}\big[(-1)^{\sum_{j\in A}y_j}\big],
\label{eq:tA}
\end{equation}
the same expectation value evaluated over the data, and a weight $w_A$
fixed by the kernel bandwidth. Assembled from the estimable correlators of
Eq.~\eqref{eq:zA-classical}, the loss and its gradients are classical, and
minimizing $\mathcal{L}$ drives each correlator toward its data moment~\cite{li2017mmdgan,li2015gmmn}. Under the
bit-to-spin map $s_j=(-1)^{x_j}=1-2x_j$, correlators and moments up to
the same body order determine one another by linear transformations. The weight
\begin{equation}
w_A:=p_\sigma^{|A|}\,(1-p_\sigma)^{\,n-|A|},\qquad
p_\sigma:=\tfrac12\big(1-e^{-1/(2\sigma^2)}\big),
\label{eq:weights}
\end{equation}
is the probability of drawing the subset $A$ when each of the $n$ qubits is
included independently with probability $p_\sigma$. The bandwidth therefore
sets the mean drawn size $\bar{m}=np_\sigma$, the \emph{mean Pauli weight}
of the kernel: a larger $\sigma$ concentrates the weight on smaller subsets.
Following Ref.~\cite{lerch2026} we scale the bandwidth as
$\sigma=\Theta(\sqrt n)$, which keeps $\bar m=\Theta(1)$, so the kernel
contains every body order but concentrates on low-body correlators at every
system size.
Equation~\eqref{eq:mmd} is thus an expectation over subsets drawn with
probability $w_A$, and the randomized estimators of Sec.~\ref{sec:setup}
estimate it by sampling $Z_A$ from this distribution instead of enumerating
all $2^n$ subsets.

\subsection{\label{sec:rand-iqp}The mixed IQP}

From here the circuit acts on two registers: the $n$ \emph{system} qubits,
which are measured and carry the model distribution, and $a$ \emph{ancilla}
qubits, which are not read. As in Eq.~\eqref{eq:iqp}, all qubits begin in
$\ket{0}$ and the circuit consists of commuting Pauli-$X$ rotations. The
\emph{system distribution} is the marginal obtained after tracing out the
ancillas. We show that this marginal is a uniform mixture of $L=2^a$
independently parametrized $n$-qubit IQP circuits. Each component, with its
own angle vector and Born distribution, is called a \emph{branch}. The
\emph{mixed IQP} extends this decomposition to configurable branch weights.

The ancilla circuit is built in two steps, following the \emph{compiled
IQP} construction of Ref.~\cite{slim2026}. First, the gates are taken
from the $n$-qubit circuit of Eq.~\eqref{eq:iqp}: the one- and two-qubit
rotations about the Pauli words $X_G$. Second, $a$ ancilla qubits are
adjoined, and each rotation is replicated once for every ancilla subset
$S\subseteq[a]$: the copy labeled $S$ rotates by its own \emph{compiled
angle} $\widetilde{\theta}_{G,S}$ about $X_G$ extended by Pauli-$X$
operators on the ancillas in $S$, and the $S=\emptyset$ copy is the
original rotation. Every generator is still a product of Pauli-$X$
operators, so the result is again an IQP circuit, now on $n+a$ qubits,
\begin{equation}
U_{\text{anc-IQP}}
=\exp\!\Big(i\sum_{G}\sum_{S\subseteq[a]}\widetilde{\theta}_{G,S}\,
 X_G^{\mathrm{sys}}\otimes X_S^{\mathrm{anc}}\Big),
\label{eq:miqp-unitary}
\end{equation}
where $X_G^{\mathrm{sys}}=\bigotimes_{j\in G}X_j$ acts on the system
register and $X_S^{\mathrm{anc}}=\bigotimes_{r\in S}X_{n+r}$ on the
ancillas in $S$ (the empty product being the identity). The copies with
$S\neq\emptyset$ act on more than two qubits, so the compiled circuit is
not itself order-$2$; the block structure below resolves it into $L$
ancilla-free order-$2$ circuits.

The compiled circuit is diagonal in the ancilla $X$ basis. Let
$\ket{\ell}_{X,\mathrm{anc}}$, $\ell\in\{0,1\}^a$, denote an $X$-basis
state. Each ancilla word acts on it as a sign,
$X_S^{\mathrm{anc}}\ket{\ell}_{X,\mathrm{anc}}
=(-1)^{S\cdot\ell}\ket{\ell}_{X,\mathrm{anc}}$, so when the ancillas are
in this state the $2^a$ copies of the rotation with support $G$ merge
into a single one- or two-qubit rotation about $X_G$, with the signed
angle sum
\begin{equation}
\theta_G^{(\ell)}=\sum_{S\subseteq[a]}(-1)^{S\cdot\ell}\,
 \widetilde{\theta}_{G,S}.
\label{eq:ciqp-angles}
\end{equation}
This Walsh--Hadamard map is invertible, with inverse
$\widetilde{\theta}_{G,S}=L^{-1}\sum_{\ell}(-1)^{S\cdot\ell}\,
\theta_G^{(\ell)}$, so the $L$ branch angle vectors can be chosen freely
and determine the compiled angles uniquely. Thus the physical unitary of
Eq.~\eqref{eq:miqp-unitary} has the block form
\begin{equation}
U_{\text{anc-IQP}}
=\sum_{\ell=0}^{L-1}
U(\bm{\theta}^{(\ell)})_{\mathrm{sys}}\otimes
\ket{\ell}_{X,\mathrm{anc}}\!\bra{\ell}_{X,\mathrm{anc}},
\label{eq:ciqp}
\end{equation}
where $U(\bm{\theta}^{(\ell)})$ is the ancilla-free order-$2$ IQP
circuit of Eq.~\eqref{eq:iqp}. The standard ancilla input is uniform in
this basis,
$\ket{0}_{\mathrm{anc}}^{\otimes a}
=L^{-1/2}\sum_\ell\ket{\ell}_{X,\mathrm{anc}}$. Tracing out the ancillas
therefore removes the cross terms between blocks and gives the uniform
mixture
\begin{equation}
p_{\Theta}(\bm{z})=\frac{1}{L}\sum_{\ell=0}^{L-1}
 |\braket{\bm{z}|U(\bm{\theta}^{(\ell)})|0^{\otimes n}}|^2,
\label{eq:moiqp}
\end{equation}
with branch angles
$\bm{\theta}^{(\ell)}=(\theta_G^{(\ell)})_G$ collected in
$\Theta=(\bm{\theta}^{(0)},\dots,\bm{\theta}^{(L-1)})$.

Allowing configurable mixture weights completes the mixed IQP family.
Replacing the uniform $X$-basis amplitudes of
$\ket{0}_{\mathrm{anc}}^{\otimes a}$ by
\begin{equation}
\ket{\chi_{\bm{\pi}}}_{\mathrm{anc}}
:=\sum_{\ell=0}^{L-1}\sqrt{\pi_\ell}\,
\ket{\ell}_{X,\mathrm{anc}},
\label{eq:ancilla-weights}
\end{equation}
with $\pi_\ell\ge0$ and $\sum_\ell\pi_\ell=1$, leaves the per-branch
circuits untouched and turns the data marginal into the
\emph{weighted} mixture
$\sum_{\ell}\pi_\ell\,|\braket{\bm{z}|U(\bm{\theta}^{(\ell)})|0^{\otimes n}}|^2$
(Appendix~\ref{app:train-embed}). The mixed IQP is this weighted family,
with trainable branch angles $\Theta$ and configurable branch weights
$\pi_\ell$: uniform weights recover the compiled IQP of
Ref.~\cite{slim2026}, and non-uniform weights let the model match a wider
range of data distributions. In this work the weights are not optimized;
they are fixed before training, either uniform or set to data-determined
group masses by the initialization (Sec.~\ref{sec:init},
Appendix~\ref{app:branch-weights}), and only the branch angles are
trained.

The uniform member is an ordinary $(n+a)$-qubit compiled IQP circuit with
the standard all-zero input. For nonuniform weights, preparing
$\ket{\chi_{\bm{\pi}}}_{\mathrm{anc}}$ is an additional
state-preparation step that is not assumed to belong to the commuting IQP
gate set. This distinction does not affect training or deployment. Every
branch is an IQP circuit whose correlators are estimated classically, and the
weighted MMD follows from their fixed $\pi_\ell$-weighted sum
(Appendix~\ref{app:correlators}). At deployment, each shot draws a branch
$\ell$ with probability $\pi_\ell$ and runs the $n$-qubit circuit
$U(\bm{\theta}^{(\ell)})$ once. This classical randomized routing reproduces
the same weighted system distribution as the general ancilla amplitude
state, without requiring its preparation.

\subsection{\label{sec:modesep}Branch diversity and mode separation}

A mixed IQP-QCBM exhibits \emph{mode separation} when its branches
represent distinct probability modes of the target distribution: each
branch carries its own distribution, and the mixture covers the data mode
by mode.

As its quantitative measure we define the squared \emph{branch diversity}
\begin{equation}
D_{\mathrm{branch}}^2(\Theta)
:=\sum_{A\subseteq[n]}w_A\,
\operatorname{Var}_{\ell\sim\pi}\!\left[z_A(\bm\theta^{(\ell)})\right],
\label{eq:branch-diversity}
\end{equation}
the kernel-weighted dispersion of the branch correlators about the mixture
correlator. Equivalently, it is one half of the $\pi_\ell$-weighted mean
pairwise MMD$^2$ between branches
[Eq.~\eqref{eq:dbranch}]. Below, $D_{\mathrm{branch}}$ denotes its
nonnegative square root, on the MMD scale. The squared diversity vanishes
exactly when all positive-weight branch distributions coincide and increases
as they differ in the correlators resolved by the training kernel; it does
not measure agreement with the target distribution.

\subsection{\label{sec:init-bp}\label{sec:init}Initialization strategies}

We study three initialization schemes, referred to throughout as the
\emph{data-agnostic}, \emph{global}, and \emph{cluster} initializations.
The starting parameter values matter because full-angle random
initialization produces a barren plateau, in which the gradients and loss
variance vanish exponentially with system
size~\cite{mcclean2018,cerezo2021,rudolph2024barriers,lerch2026}, whereas
the low-body MMD admits specific starting parameter settings whose
neighborhoods remain
trainable~\cite{grant2019,zhang2022,verdon2019learning,rudolph2024barriers,lerch2026}. Each
scheme is fixed by its \emph{center}, the parameter vector
$\Theta=(\bm\theta^{(0)},\dots,\bm\theta^{(L-1)})$ that the branches take
before the coincidence-breaking perturbation below; the trainability statements of
Sec.~\ref{sec:trainability} are made at these centers.

The data-agnostic and global schemes start every branch from the same angle
vector. By the Walsh--Hadamard map of Eq.~\eqref{eq:ciqp-angles}, this
\emph{branch-coincident} condition is equivalent, for every gate
$G$, to setting its ancilla-dependent compiled angles to zero,
\begin{equation}
\theta_G^{(0)}=\cdots=\theta_G^{(L-1)}
\quad\Longleftrightarrow\quad
\widetilde{\theta}_{G,S}=0\quad(S\neq\emptyset),
\label{eq:branch-coincident-center}
\end{equation}
so no compiled angle distinguishes the branches and
$D_{\mathrm{branch}}=0$. After specifying each center, we add the small
independent coincidence-breaking perturbation of Sec.~\ref{sec:train}. Because
the perturbation is random, it carries no information about the target: the
branches start slightly apart, but in random directions, and no branch is
associated with any data mode. The cluster scheme instead starts the
branches on different data modes.

\subsubsection{\label{sec:init-agnostic}Data-agnostic initialization}

This scheme uses no data. Subject to
the branch-coincident condition in Eq.~\eqref{eq:branch-coincident-center}, its
branch-independent compiled angles and weights are
\begin{equation}
\widetilde{\theta}_{j,\emptyset}=\frac{\pi}{4},
\qquad
\widetilde{\theta}_{jk,\emptyset}=0,
\qquad
\pi_\ell=\frac{1}{L}.
\label{eq:agnostic-center}
\end{equation}
Every branch therefore has $\theta_j^{(\ell)}=\pi/4$ and
$\theta_{jk}^{(\ell)}=0$, the unbiased single-circuit start of
Ref.~\cite{lerch2026}, and both the branch distributions and their mixture are
uniform over bit strings.

\subsubsection{\label{sec:init-global}Global initialization}

Following the single-circuit data-dependent initialization of
Ref.~\cite{recio2025}, this scheme uses the full training set $\mathcal D$
to initialize every branch identically (the ancillas are inactive at the
start). With $t_j\equiv t_{\{j\}}$ the
one-body moment of Eq.~\eqref{eq:tA} evaluated on $\mathcal D$, and subject
again to Eq.~\eqref{eq:branch-coincident-center}, the compiled
initialization is
\begin{equation}
\widetilde{\theta}_{j,\emptyset}=\frac{1}{2}\arccos t_j,
\qquad
\widetilde{\theta}_{jk,\emptyset}=0,
\qquad
\pi_\ell=\frac{1}{L},
\label{eq:global-center}
\end{equation}
so every branch reproduces the data's one-body moments,
$z_{\{j\}}(\bm\theta^{(\ell)})=\cos(2\theta_j^{(\ell)})=t_j$, while
$\theta_{jk}^{(\ell)}=0$. Higher-body data correlations are not matched
by this initialization and must be learned during training.

\subsubsection{\label{sec:init-dependent}Cluster initialization}

The cluster-initialized scheme installs mode separation at the start:
each branch is assigned its own group of the training data and
moment-matched to that group. Let $\{C_\ell\}_{\ell=0}^{L-1}$
be a partition of the training set $\mathcal{D}$ into $L$ groups, and
$t_j^{(\ell)}$ the one-body moment of Eq.~\eqref{eq:tA} evaluated over
$C_\ell$. The center sets each branch's one-body angles to
reproduce its group's moments and starts every two-body angle at zero,
\begin{equation}
\cos2\theta_j^{(\ell)}=t_j^{(\ell)},\qquad
\theta_{jk}^{(\ell)}=0\,,
\label{eq:datadep-center}
\end{equation}
so branch $\ell$ starts from a distribution in which the bits are
independent,
\begin{equation}
p_{\bm{\theta}^{(\ell)}}(\bm{z})
=\prod_{j=1}^{n}\frac{1+(-1)^{z_j}\,t_j^{(\ell)}}{2}\,,
\label{eq:datadep-product}
\end{equation}
with the mixture weight fixed by the group mass,
$\pi_\ell=|C_\ell|/|\mathcal{D}|$ [Eq.~\eqref{eq:ancilla-weights}].
Thus all within-branch bit correlations are learned during training.

If the group one-body moment vectors differ across clusters, the branch
products of Eq.~\eqref{eq:datadep-product} are distinct, and each branch
starts on its own cluster-associated region of the data. When the partition
aligns with the target's mode structure, this provides an initial degree of
mode separation. Here the groups are obtained by
spectral clustering~\cite{ng2001spectral,vonluxburg2007} of the training set
under the Hamming affinity
\begin{equation}
A_{\bm{x}\bm{y}}=\exp\!\left(-\frac{d_H(\bm{x},\bm{y})}{2\sigma_c^2}\right),
\label{eq:affinity}
\end{equation}
where $d_H(\bm{x},\bm{y})$ counts the bits on which the two strings differ
and $2\sigma_c^2$ is set to the median nonzero pairwise distance of the
training set: the same Hamming-distance similarity the low-body MMD
resolves (Sec.~\ref{sec:mmd}). The affinity only defines the groups and is
separate from the MMD training kernels ($\sigma_c\neq\sigma$).

\section{\label{sec:theory}\label{sec:trainability}Locally trainable initializations}

Building on the curvature analysis of ancilla-free IQP-QCBMs
~\cite{lerch2026,characterizing_iqp,deluca2026}, we establish local
trainability for mixed models with a polynomial number of branches. For the low-body MMD, the curvature of a branch parameter separates into a
target-independent model-sensitivity term and a target-dependent
data-mismatch term. This decomposition yields local barren-plateau
avoidance in an inverse-polynomial neighborhood of the data-agnostic center
for every target, and of the global and cluster centers under explicit
conditions on the assigned data groups and branch weights.

\subsection{\label{sec:tr-curv}Curvature criterion}

Barren plateaus are commonly defined by an exponentially vanishing
gradient variance; for standard parametrized circuits this is equivalent
to exponential concentration of the loss about its
mean~\cite{arrasmith2022,rudolph2024barriers}. Our certificates bound the
loss. Following the local framework of Ref.~\cite{lerch2026}, we call an
initialization \emph{locally free of barren plateaus} when the loss
variance remains inverse-polynomial,
$\mathrm{Var}_\Theta[\mathcal{L}]=\Omega(1/\mathrm{poly}(n))$, on a
neighborhood of inverse-polynomial radius around the analyzed center,
which rules out an exponentially flat landscape there. This is a local
statement about the initial landscape, not a guarantee of global
convergence or successful optimization, and it constrains the loss
variance rather than the variance of any single gradient component.

The variance need not be certified directly: a single loss curvature
suffices. In Lemma~\ref{lem:curv-bp} of Appendix~\ref{app:bp-lemma},
adapting Ref.~\cite{lerch2026}, we prove that if
$|\partial^2_{\theta_G^{(\ell)}}\mathcal{L}|=
\Omega(1/\mathrm{poly}(n))$ at the analyzed center, the variance bound
above follows on an inverse-polynomial patch. The quantity to compute is
therefore a second derivative of the loss. Our loss is the weighted
correlator sum of Eq.~\eqref{eq:mmd}; differentiating it twice with
respect to a single branch angle, with $\pi_\ell$ the weight of that
branch, separates the curvature into
\begin{equation}
\begin{aligned}
\partial^2_{\theta_G^{(\ell)}}\mathcal{L}
=\sum_{A:\,G\cdot A=1}2w_A\Big[
&\underbrace{\pi_\ell^2\big(g_A^{(G,\ell)}\big)^2}_{\text{model sensitivity}}\\
&{}+\underbrace{4\pi_\ell\,z_A(\bm{\theta}^{(\ell)})
\big(t_A-z_A(\Theta)\big)}_{\text{data mismatch}}
\Big],
\end{aligned}
\label{eq:curv-decomp}
\end{equation}
where $z_A(\bm\theta^{(\ell)})$ is branch $\ell$'s Pauli-$Z$
correlator, $z_A(\Theta)$ is the correlator of the full mixture, $t_A$ is the
corresponding target-data correlator, $w_A$ is its MMD weight, and
$g_A^{(G,\ell)}=\partial_{\theta_G^{(\ell)}}
z_A(\bm\theta^{(\ell)})$.

The two terms play different roles. The \emph{model-sensitivity} term
supplies the nonnegative part of the curvature: it is the squared response
of the branch correlator to $\theta_G^{(\ell)}$, independent of the
target, and large when a small parameter change moves a correlator the loss
uses. The \emph{data-mismatch} term supplies the signed part: it couples
the mixture's residual error $t_A-z_A(\Theta)$ to the branch correlator
$z_A(\bm\theta^{(\ell)})$, vanishes when either factor does, and can
reinforce or partially cancel the curvature. A nonvanishing curvature can
therefore come from a strong circuit response. The next two subsections
evaluate both terms at the initialization centers of Sec.~\ref{sec:init},
where the uniform schemes have $\pi_\ell=1/L$.

\Needspace{10\baselineskip}
\subsection{\label{sec:tr-agnostic}Trainability at the data-agnostic center}

At the data-agnostic center of Sec.~\ref{sec:init-agnostic}, the
data-mismatch term of Eq.~\eqref{eq:curv-decomp} drops out for every target,
because every nontrivial branch correlator vanishes there; only the
model-sensitivity term remains.

\begin{theorem}[Local trainability at the data-agnostic center]
\label{thm:rand-bp}
Consider the uniform ($\pi_\ell=1/L$) mixed IQP-QCBM of Eq.~\eqref{eq:moiqp}
with $L=O(\mathrm{poly}(n))$ branches and the low-body MMD loss
($\sigma=\Theta(\sqrt{n})$). At the data-agnostic center of
Eq.~\eqref{eq:agnostic-center}, the loss curvature of every one-body branch parameter
$\theta_j^{(\ell)}$ satisfies
\begin{equation}
c \;:=\;
\partial^2_{\theta_j^{(\ell)}}\mathcal{L}
\;=\; \frac{8\,w_{\{j\}}}{L^{2}}
\;=\; \Omega\!\left(\frac{1}{\mathrm{poly}(n)}\right).
\label{eq:rand-curv}
\end{equation}
Consequently, the center is locally free of barren plateaus:
$\mathrm{Var}_\Theta[\mathcal{L}]=\Omega(1/\mathrm{poly}(n))$ throughout an
inverse-polynomial neighborhood.
\end{theorem}

\emph{Proof sketch.} Only $A=\{j\}$ survives in
Eq.~\eqref{eq:curv-decomp}: its sensitivity is
$g_{\{j\}}^{(j,\ell)}=-2$, while every $|A|\geq2$ term contains
$\cos(\pi/2)=0$. This gives Eq.~\eqref{eq:rand-curv}; because
$w_{\{j\}}=\Theta(1/n)$ and $L=O(\mathrm{poly}(n))$, the curvature is
inverse-polynomial. Appendix~\ref{app:bp-agnostic} gives the full proof, and
Sec.~\ref{sec:exp-theory} verifies the exact $L^{-2}$ scaling.

\subsection{\label{sec:tr-datadep}Trainability at the global and cluster centers}

Both schemes start branch $\ell$ on the one-body moments of its assigned
data $C_\ell$ with all two-body angles zero
(Sec.~\ref{sec:init-dependent}: $C_\ell=\mathcal D$ for the global scheme,
one cluster per branch for cluster initialization); $t_A^{(\ell)}$ denotes
the empirical moment of $C_\ell$ for subset $A$. The guarantee below assumes
the implemented weight rule: the groups partition the training set and each
branch weight is its group's empirical mass, $\pi_\ell=|C_\ell|/N$, so that
$\sum_\ell\pi_\ell t_A^{(\ell)}=t_A$ holds exactly for every subset (for the
global scheme every group equals the training set and the identity holds for
any weights). Beyond this, two assumptions on the groups certify local
barren-plateau avoidance.

\begin{assumption}[Approximately factorizable groups]
\label{as:factor}
There exists a constant $C>0$, independent of $n$, such that every group
$C_\ell$ and every subset $A$ satisfy
\begin{equation*}
\left|t_A^{(\ell)}-\prod_{j\in A}t_j^{(\ell)}\right|
\leq \left(\frac{C}{n}\right)^{|A|/2}.
\end{equation*}
\end{assumption}
This condition bounds the part of each correlator that is not explained by
the product of its one-body marginals; it does not require
$t_A^{(\ell)}$ itself to be small. The residual correlation is therefore
suppressed geometrically with the body order. This is a sufficient condition,
not a necessary one.

\begin{assumption}[A noncollapsed, well-weighted group]
\label{as:peak}
For at least one group $\ell^\star$ and one qubit $j^\star$,
\begin{equation*}
\pi_{\ell^\star}=\omega(1/n),
\qquad
1-\big(t_{j^\star}^{(\ell^\star)}\big)^2=\Theta(1).
\end{equation*}
\end{assumption}
The weight condition makes the positive one-body curvature large enough to
dominate the $O(1/n^3)$ residual below. The second condition prevents its
model sensitivity from vanishing. Other groups may be collapsed; only one
witness group must satisfy both conditions. Both assumptions are
asymptotic statements about how the grouped data behave as $n$ grows, so
no experiment at a single system size can confirm them.
Appendix~\ref{app:assumption-check} reports finite-size diagnostics of the
corresponding quantities at the sizes studied; these assess plausibility,
not the scaling with $n$.

\begin{theorem}[Local trainability at the global and cluster centers, informal]
\label{thm:datadep-bp-main}
For the weighted mixed IQP-QCBM with $L=O(\mathrm{poly}(n))$ branches, the
low-body MMD ($\sigma=\Theta(\sqrt{n})$), and branch weights equal to the
empirical group masses, consider the data-dependent center
\begin{equation*}
\cos\!\big(2\theta_j^{(\ell)}\big)=t_j^{(\ell)},
\qquad
\theta_{jk}^{(\ell)}=0.
\end{equation*}
Under Assumptions~\ref{as:factor} and~\ref{as:peak}, the loss curvature at
qubit $j^\star$ of branch $\ell^\star$ satisfies
\begin{equation}
\begin{aligned}
c:=\partial^2_{\theta_{j^\star}^{(\ell^\star)}}\mathcal{L}
&\ \geq 8w_{\{j^\star\}}\pi_{\ell^\star}^2
\big[1-(t_{j^\star}^{(\ell^\star)})^2\big]-O(1/n^3)\\
&\ =\Omega(1/\mathrm{poly}(n)).
\end{aligned}
\label{eq:datadep-curv-main}
\end{equation}
\end{theorem}

By the curvature criterion of Sec.~\ref{sec:tr-curv}, this curvature places
the model locally free of barren plateaus in an inverse-polynomial
neighborhood of the data-dependent center. We evaluate the curvature and its
leading positive term exactly at the initialization centers of the two
$n=16$ benchmarks in Sec.~\ref{sec:exp-theory}.

\emph{Proof sketch.} The witness branch supplies the positive one-body
model-sensitivity term in Eq.~\eqref{eq:datadep-curv-main}. Because the
weights are the group masses of a partition, the identity
$\sum_\ell\pi_\ell t_A^{(\ell)}=t_A$ cancels the weighted group moments
exactly: the one-body mismatch vanishes, and every higher-body mismatch
reduces to a weighted sum of the per-group residuals
$t_A^{(\ell)}-\prod_{j\in A}t_j^{(\ell)}$ between each group moment and the
factorized correlator realized at the center. Assumption~\ref{as:factor}
bounds these residuals by $(C/n)^{|A|/2}$ at every body order, giving the
$O(1/n^3)$ remainder without truncating the MMD.
Assumption~\ref{as:peak} makes the positive term dominate.
Appendix~\ref{app:bp-datadep} gives the full proof.

\section{\label{sec:usable}Mixture expressivity and the need for branch diversity}

Reducing the loss below the minimum attainable by a single ancilla-free circuit requires
the branches to represent different distributions, yet the gradients that
create these differences are zero when the branches are identical and
remain small while the branches are similar. The requirement has a classical
parallel in \emph{mode collapse}: an adversarial generative model can
concentrate its samples on a few modes of the target, and multi-generator
variants counter this by training several generators under an objective that
explicitly enforces diversity among them, so that different generators
capture different modes~\cite{Ghosh_2018_CVPR}. The branches of the mixture
play the role of these generators, with the additional obstacle that their
separating gradients vanish at coincidence. This section proves the two
statements: the diversity floor (Sec.~\ref{sec:modesep-theory}) and the
suppression of the branch-separating gradients (Sec.~\ref{sec:grad-supp}).
Together with Sec.~\ref{sec:trainability}, they explain why cluster
initialization can be important: its branches start from different
distributions.

\subsection{\label{sec:modesep-theory}Diversity required beyond the best ancilla-free circuit}

An ancilla-free IQP circuit is not
universal~\cite{kurkin2025note,kurkin2025universality}, whereas
Proposition~\ref{prop:universal} shows that the mixed model is. This larger
representational range improves on the MMD of the best ancilla-free IQP
circuit only when the
branches represent different distributions:
Proposition~\ref{prop:diversity} shows that the required diversity is at least
the margin gained below the ancilla-free floor. Branch count therefore adds
capacity, but branch diversity is necessary to use it.

\begin{proposition}[Universality]
\label{prop:universal}
For each probability distribution $p$ on $\{0,1\}^n$ there is a
mixed IQP-QCBM with $L=|\operatorname{supp}(p)|$ branches whose data
distribution equals $p$. This universality can be realized either by classical
randomized branch routing or by a branch-controlled circuit with
$a=\lceil\log_2 L\rceil$ ancilla index qubits prepared in the general
amplitude state of Eq.~\eqref{eq:ancilla-weights}.
\end{proposition}

The proof assigns one deterministic branch to every string in the support of
$p$ (Appendix~\ref{app:universal}). It can therefore use $L=2^n$ branches
and $a=n$ ancilla index qubits: the proposition is a representability
construction for the weighted branch family, not a new efficient
universality result for ordinary compiled IQP circuits, and for IQP
circuits with hidden units a stronger universality statement is available
in Ref.~\cite{kurkin2025note}. The practical question is whether far fewer branches
can divide the target into broad modes, with each branch representing many
outcomes. Cluster initialization is designed to supply this division from
the data.

Let $\Delta_{\text{anc-free}}$ denote the minimum MMD attainable by one
ancilla-free IQP circuit on the fixed interaction graph [Eq.~\eqref{eq:floor}].

\begin{proposition}[Diversity floor]
\label{prop:diversity}
For any (weighted) mixture of ancilla-free IQP circuits on a fixed graph, with weights $\pi_\ell$ on the simplex,
\begin{equation}
\sqrt{\mathcal{L}(\Theta)}\;\ge\;\Delta_{\text{anc-free}}-D_{\mathrm{branch}}\,,
\label{eq:diversity-bound}
\end{equation}
so improving below the ancilla-free floor requires
$D_{\mathrm{branch}}\geq\Delta_{\text{anc-free}}
-\sqrt{\mathcal{L}(\Theta)}>0$. Coincident branches therefore cannot improve
on that floor.
\end{proposition}

\emph{Proof sketch.} The mixture correlator is the weighted mean of the branch
correlators, and the definition of $D_{\mathrm{branch}}$ ensures that some
reachable branch lies within $D_{\mathrm{branch}}$ of this mean. The triangle
inequality gives
$\Delta_{\text{anc-free}}\leq\sqrt{\mathcal{L}(\Theta)}
+D_{\mathrm{branch}}$, proving Eq.~\eqref{eq:diversity-bound}. See
Appendix~\ref{app:diversity} for the full proof.

Diversity is necessary but not sufficient: the branches must differ along the
target's mode structure, as in the mode separation of
Sec.~\ref{sec:modesep}. Cluster initialization supplies this target-aligned
diversity, whereas the data-agnostic and global schemes start with coincident
branches. The next subsection analyzes the gradients needed to leave that
coincidence.

\subsection{\label{sec:grad-supp}Suppression of branch-separating gradients at and near coincidence}

At branch coincidence all branches have identical sensitivities. A
branch-distinguishing angle $\widetilde{\theta}_{G,S}$ ($S\neq\emptyset$)
has a gradient given by their signed Walsh sum, whose signs sum to zero;
the common angle $\widetilde{\theta}_{G,\emptyset}$ has no such
cancellation.

\begin{theorem}[Ancilla-gradient suppression, informal]
\label{thm:grad-suppress-main}
For the uniform ($\pi_\ell=1/L$) mixed IQP-QCBM with the low-body MMD loss
of Eq.~\eqref{eq:mmd}, let
$\bar{\bm{\theta}}=L^{-1}\sum_\ell\bm{\theta}^{(\ell)}$ and define the branch
spread
$\delta=\max_\ell\|\bm{\theta}^{(\ell)}-\bar{\bm{\theta}}\|_\infty$.
For every branch-distinguishing compiled angle
$\widetilde{\theta}_{G,S}$ with $S\neq\emptyset$,
\begin{equation}
\big|\partial_{\widetilde{\theta}_{G,S}}\mathcal{L}\big|\;\le\;
2\delta\sum_A w_A\,C_A\,\big|z_A(\Theta)-t_A\big|\;=\;O(m_0\,\delta),
\label{eq:grad-suppress-main}
\end{equation}
where $C_A\le8m_0$ and the generator count
$m_0=\mathrm{poly}(n)$. Thus the branch-distinguishing gradient vanishes at
$\delta=0$; for a fixed architecture it grows at most linearly with the
existing spread, and uniformly in the system size the bound is
$O(m_0\delta)$.
\end{theorem}

Exact coincidence is therefore invariant under deterministic gradient
descent. The random perturbation used by the data-agnostic and global schemes
makes $\delta>0$ and supplies an $O(m_0\delta)$ gradient that optimization may
amplify; Theorem~\ref{thm:grad-suppress-main} bounds this instantaneous
gradient, not the escape time. The perturbation must also be small enough
that the start remains inside the patch on which
Theorem~\ref{thm:rand-bp} applies. At the data-agnostic center,
Corollary~\ref{cor:admissible-fluctuation} in
Appendix~\ref{app:diversity} provides such a margin: an explicit
admissible perturbation radius, nonincreasing in the branch count $L$,
within which the loss variance remains inverse-polynomial and every
branch-distinguishing gradient stays $O(1/n)$. This radius is
a conservative sufficient choice; it neither identifies the maximal
linear regime nor certifies that the finite run perturbation lies inside
it. Theorem~\ref{thm:grad-suppress-main} is proved formally in
Appendix~\ref{app:diversity}, and Sec.~\ref{sec:exp-theory} tests its
linear gradient prediction.

Cluster initialization begins with target-aligned separation and therefore
does not rely on this random amplification, while retaining the conditional
local guarantee of Sec.~\ref{sec:tr-datadep}. Section~\ref{sec:result} tests
whether this separation persists and lowers the test MMD$^2$.

\begin{figure*}[!t]
\includegraphics[width=\textwidth]{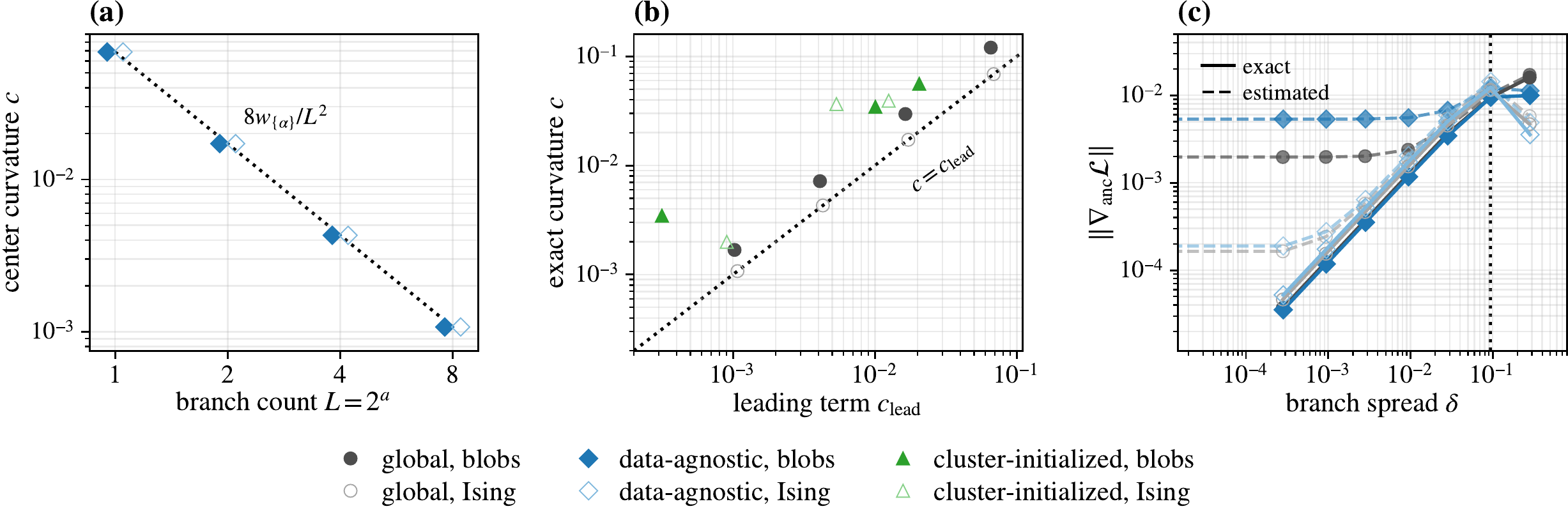}
\caption{\label{fig:theory-check} Exact tests of the three local predictions
on the two $n=16$ datasets. (a)~Theorem~\ref{thm:rand-bp}, at the
data-agnostic center: the one-body curvature follows the $8w_{\{j\}}/L^2$ prediction of
Theorem~\ref{thm:rand-bp}. (b)~Theorem~\ref{thm:datadep-bp-main}, at the global and
cluster-initialized centers: the exact witness curvature $c$ against the leading positive term
$c_{\mathrm{lead}}$ defined in the text. Equation~\eqref{eq:datadep-curv-main}
requires only $c\geq c_{\mathrm{lead}}-R_n$; the dotted diagonal
$c=c_{\mathrm{lead}}$ is therefore a stronger finite-size reference.
(c)~Theorem~\ref{thm:grad-suppress-main}, at branch-coincident starts
($a=1$): the exact branch-separating gradient vanishes at coincidence and grows
linearly with the existing spread (Pearson $r>0.99$) through the spread of
the run initialization, marked by the dotted vertical line; the estimator
resolves this gradient above its sampling floor.
The curvature decomposition for (b) and the full ancilla-count sweep for (c)
are given in Appendix~\ref{app:theory-check}.}
\end{figure*}

\section{\label{sec:result}From initialization to mode separation}

The experiments examine how initialization affects the local loss landscape,
subsequent optimization, final test MMD$^2$, and the distributions learned by
individual branches. Exact calculations on the two $n=16$ datasets first test the
predicted center curvatures and the suppression of branch-separating
gradients near coincident branches (Sec.~\ref{sec:exp-theory}). Training on
all four benchmarks then tests whether cluster initialization accelerates
convergence and reduces the variation over the tested seeds
(Sec.~\ref{sec:exp-train}). Finally, we relate the test MMD$^2$ to branch
diversity and inspect the trained distribution of each branch to determine
which part of the target distribution it represents
(Sec.~\ref{sec:exp-modesep}). Section~\ref{sec:setup} specifies the datasets
and the common protocol.

\input{experiment_setup}

\begin{figure*}[!t]
\includegraphics[width=\textwidth]{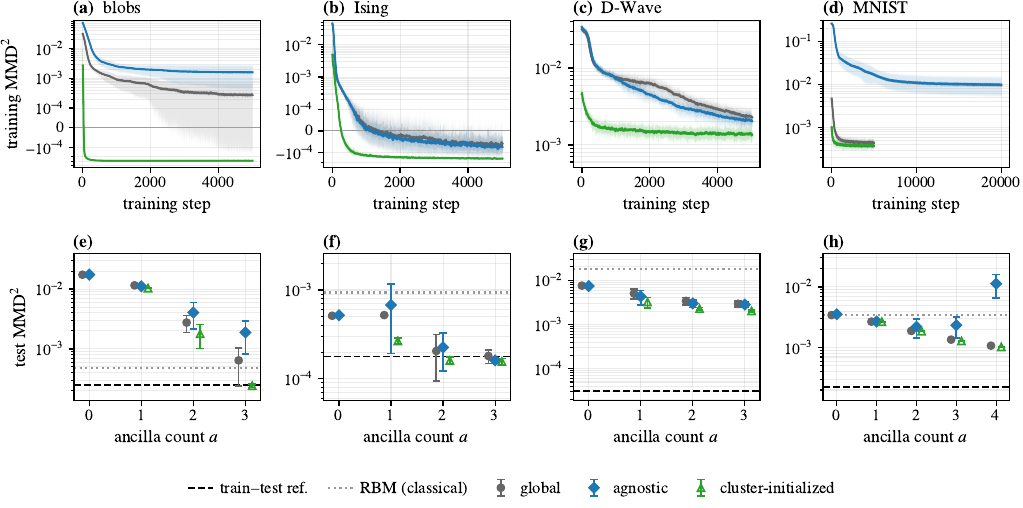}
\caption{\label{fig:train-results} Training dynamics and test MMD$^2$ on the four
benchmarks; columns are datasets, at the headline mixture size $a=3$
($a=4$ for MNIST). (a)--(d)~Training MMD$^2$ against the training step:
bold curves are $25$-step moving averages of the seed mean, faint curves the
raw mean, and bands one standard deviation. The unbiased finite-sample loss
samples Pauli-$Z$ operators and estimates model correlators and data moments;
after removing self-correlations, values near zero may be slightly negative.
(e)--(h)~Test MMD$^2$ against
the ancilla count $a$ ($L=2^a$), averaged over the sweep
$\bar m=1,\dots,6$, with the train--test sampling reference (dashed) and the
best classical RBM (dotted). Means and bands use ten training seeds for the
$n=16$ datasets and five for MNIST and D-Wave.}
\end{figure*}

\subsection{\label{sec:exp-theory}Verification of the local landscape}

Exact calculations on the $n=16$ blobs and Ising benchmarks confirm the
trainability predictions, as shown in Figs.~\ref{fig:theory-check}(a) and
\ref{fig:theory-check}(b). At the data-agnostic center, the target-dependent
term vanishes, so the blobs and Ising give the same exact curvature
$8w_{\{j\}}/L^2$ for $L=1,2,4,8$. For the data-dependent centers,
let $c_{\mathrm{lead}}:=8w_{\{j^\star\}}\pi_{\ell^\star}^2
[1-(t_{j^\star}^{(\ell^\star)})^2]$. Equation~\eqref{eq:datadep-curv-main} then
states $c\ge c_{\mathrm{lead}}-R_n$. Figure~\ref{fig:theory-check}(b)
uses the stronger reference $c=c_{\mathrm{lead}}$ to show how much curvature
remains before allowing for the finite-size remainder. Every witness point lies
above that line. At
the global center, $c/c_{\mathrm{lead}}\simeq1$ for Ising and
ranges from $1.64$ to $1.84$ for the blobs. Cluster initialization gives
larger ratios: $2.73$--$10.99$ for the blobs and $2.21$--$6.83$ for Ising.
The $a=3$ blobs remain above the stronger reference even though their eight clusters are
each concentrated around one mode and the witness marginal approaches saturation.
Appendix~\ref{app:theory-check} gives the complete sensitivity--mismatch
decomposition and per-configuration values.

The branch-separating gradient exhibits the linear suppression derived in
Theorem~\ref{thm:grad-suppress-main}: it vanishes at coincidence and grows
linearly with the branch spread ($r>0.99$ in all four $a=1$ cases;
Fig.~\ref{fig:theory-check}(c)). This sweep tests the linear relation through
the run initialization, beyond what the conservative admissible radius of
Corollary~\ref{cor:admissible-fluctuation} certifies. The suppression also
raises the required estimator precision: in the full ancilla-count sweep, the
exact data-agnostic gradient on the blobs lies below the estimator floor at
$a=3$, whereas the Ising gradient remains resolved
(Fig.~\ref{fig:app-branch-spread}). Thus the deterministic direction may be
present yet statistically unresolved at fixed sampling budget.

\subsection{\label{sec:exp-train}Optimization from coincident and separated initializations}

All three schemes show decreasing training loss, but initial mode separation
changes how quickly the additional branches become useful. The
cluster-initialized curves fall most steeply and plateau earliest; the global
and data-agnostic curves descend more gradually from near-coincident starts
(Fig.~\ref{fig:train-results}(a)--(d)). This ordering follows the
branch-gradient mechanism: cluster initialization begins with distinct branches, whereas the coincident schemes must amplify a small perturbation.

Cluster initialization achieves the lowest mean test MMD$^2$, or matches it
within the observed seed variation, on all four benchmarks. Its test MMD$^2$ is comparable to the train--test
sampling reference on the blobs and Ising, is the lowest tested value on
D-Wave, and matches the global scheme on MNIST
(Fig.~\ref{fig:train-results}(e)--(h)). At the headline mixture sizes it also
lies below the reported classical RBM reference on every benchmark.
Section~\ref{sec:exp-modesep} examines how this advantage is related to the
mode separation developed by the trained branches.

Across the tested seeds, cluster initialization also gives the narrowest
bands on the blobs, Ising, and D-Wave benchmarks
(Fig.~\ref{fig:train-results}). This pattern is consistent with the
coincidence-breaking mechanism: the global and data-agnostic starts rely on a
random perturbation to generate their initially suppressed separating
gradients, so its seed-dependent direction and magnitude can affect the early gradients relative to estimator noise and hence the early trajectory; cluster initialization
does not rely on this random coincidence breaking. The comparison uses ten seeds for each
$n=16$ dataset and five for each larger dataset and does not establish
general stability.

Table~\ref{tab:anc-angles} quantifies the branch-distinguishing motion added
after each initialization by the median over the largest $1\%$ of
$|\widetilde{\theta}_{G,S}^{\,\mathrm{final}}
-\widetilde{\theta}_{G,S}^{\,\mathrm{init}}|$ for $S\neq\emptyset$.
Global and data-agnostic models build this displacement from near coincidence,
where Theorem~\ref{thm:grad-suppress-main} suppresses the relevant gradients;
cluster initialization begins separated and avoids this initial cancellation.

\begin{table}[!t]
\caption{\label{tab:anc-angles}Branch-distinguishing displacement
$|\widetilde{\theta}_{G,S}^{\,\mathrm{final}}
-\widetilde{\theta}_{G,S}^{\,\mathrm{init}}|$, $S\neq\emptyset$: median
of the largest $1\%$, mean $\pm$ standard deviation over seeds.}
\begin{ruledtabular}
\footnotesize
\setlength{\tabcolsep}{1.4pt}
\begin{tabular}{@{}l c c c@{}}
Dataset ($a$) & global & cluster-init. & data-agn. \\
\colrule
Blobs ($3$) & $0.363\pm0.027$ & $0.057\pm0.002$ & $0.459\pm0.062$ \\
Ising ($1$) & $0.376\pm0.110$ & $0.288\pm0.022$ & $0.358\pm0.122$ \\
Ising ($3$) & $0.243\pm0.018$ & $0.156\pm0.011$ & $0.240\pm0.015$ \\
MNIST ($4$) & $0.117\pm0.002$ & $0.072\pm0.001$ & $0.280\pm0.023$ \\
D-Wave ($3$) & $0.104\pm0.005$ & $0.057\pm0.006$ & $0.102\pm0.005$ \\
\end{tabular}
\end{ruledtabular}
\end{table}

\begin{figure*}[!t]
\includegraphics[width=\textwidth]{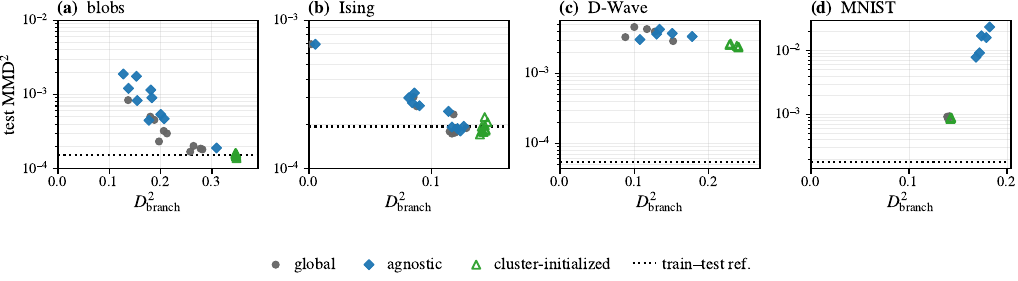}
\caption{\label{fig:diversity} Squared branch diversity
[Eq.~\eqref{eq:branch-diversity}] against test MMD$^2$, both at the same
kernel bandwidth, for individual trained models at $a=3$ (blobs and
D-Wave), $a=2$ (Ising), and $a=4$ (MNIST). One point per initialization
scheme and training seed; dotted lines mark the train--test sampling
references. Both
coordinates are exact for $n=16$ and estimated from correlators otherwise.
The Spearman correlations between the two coordinates are $-0.94$ (blobs),
$-0.58$ (Ising), $-0.77$ (D-Wave), and $+0.61$ (MNIST), with
percentile-bootstrap $95\,\%$ confidence intervals $(-0.97,-0.87)$,
$(-0.83,-0.20)$, $(-0.97,-0.30)$, and $(+0.01,+0.92)$ from $10^4$
resamples.}
\end{figure*}

The Ising $a=1$ result makes this optimization difference clearest. Cluster
initialization undergoes its largest displacement in the table,
$0.288\pm0.022$, with little seed variation. The global and data-agnostic
starts require comparable motion, $0.376\pm0.110$ and $0.358\pm0.122$, but
their displacement varies about fivefold more, consistent with their wider test
MMD$^2$ error bars in Fig.~\ref{fig:train-results}(f). Coincidence is not
prohibitive when smaller changes suffice: on MNIST, global initialization
matches the cluster test MMD$^2$ after a displacement of $0.117\pm0.002$.
These results indicate an initialization-dependent delay rather than a
prohibition: larger reorganizations rely more strongly on amplifying the
random perturbation, whereas cluster initialization begins with directed
separation.

\subsection{\label{sec:exp-modesep}Mode separation and model performance}

This section tests the central empirical claim of this work: additional
branches become useful when they generate distinct distributions, each
capturing a particular feature of the target.
We ask whether the trained branches specialize to different modes or regions
of the target distribution, and whether mixtures that develop this mode
separation achieve lower test MMD$^2$. Branch diversity quantifies how much
the branch distributions differ, whereas the distributions of the individual
branches reveal
whether those differences correspond to structure in the data; Figs.~\ref{fig:diversity}
and \ref{fig:modesep} provide these complementary views.

Target-aligned branch separation, rather than diversity alone, accompanies
lower test MMD$^2$ across the four benchmarks. Figure~\ref{fig:diversity}
plots each trained model's branch diversity $D_{\mathrm{branch}}^2$ against
its test MMD$^2$. On the blobs, Ising, and D-Wave, test MMD$^2$ decreases
with the scalar diversity measure, with Spearman correlations of $-0.94$,
$-0.58$, and $-0.77$ across the plotted models (the caption of
Fig.~\ref{fig:diversity} gives the confidence intervals), and the
cluster-initialized models combine the largest diversity with the lowest
test MMD$^2$. On MNIST the association reverses, with a Spearman
correlation of $+0.61$: a large
$D_{\mathrm{branch}}^2$ does not imply a low test MMD$^2$, because the
metric measures how much the branches differ, not whether those differences
represent the target. Across the four datasets, useful branch diversity must
therefore follow the target's structure.

\begin{figure*}[!t]
\includegraphics[width=\textwidth]{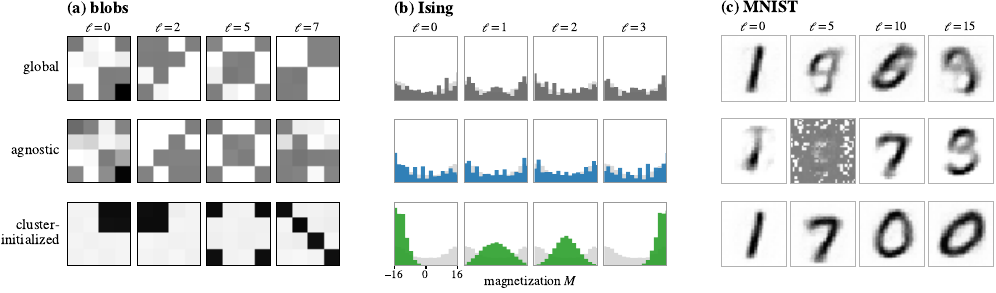}
\caption{\label{fig:modesep} Distributions of selected branches after
training: four branches spanning the branch index of each mixture. Rows
denote initialization schemes and columns denote branches $\ell$.
(a)~Binary blobs ($a=3$) and (c)~binarized MNIST ($a=4$): each cell is
the mean image of one branch, in which each pixel shows the probability
that the branch outputs a one at that pixel, with darker pixels denoting
higher probability. (b)~2D Ising ($a=2$,
all four branches): exact branch magnetization densities (colored) over
the training density (gray). Every configuration is trained at multiple
seeds (ten at $n=16$, five on MNIST); to display the differences between
the schemes conservatively, each row shows the seed with the smallest
branch diversity $D^{2}_{\mathrm{branch}}$ [Eq.~\eqref{eq:dbranch}] of its
scheme. Complete
panels containing every branch are shown
in Fig.~\ref{fig:app-modesep-full} of Appendix~\ref{app:experiments}.}
\end{figure*}

Useful mode separation gives individual branches recognizable roles in the
target distribution: a branch learns a particular image pattern or physical
sector. Figure~\ref{fig:modesep} shows representative branches, and
Fig.~\ref{fig:app-modesep-full} gives the complete per-branch panels. On the
blobs, the cluster-initialized branch means are sharp and reproduce the eight
binary patterns one for one, whereas the global and data-agnostic mixtures
contain averages or duplicates of them and leave modes uncovered. On MNIST
the picture resolves two earlier observations at once. The global and
cluster-initialized branch mean images differ in digit identity and
stroke shape,
although both reach this specialization with modest branch-distinguishing
displacement during training (Table~\ref{tab:anc-angles}). The high diversity of the
data-agnostic models in Fig.~\ref{fig:diversity} is inflated by four
branches that remain nearly uniform: noise-like distributions differ from
every digit branch, so they raise the diversity statistic while generating no
structure of the data. The Ising benchmark provides a direct physical example
of why mode separation matters. In Fig.~\ref{fig:modesep}(b), which shows the
least-diverse trained seed of each scheme, the global and data-agnostic
branches remain concentrated in overlapping, weakly magnetized regions of the
total magnetization $M=\sum_i s_i$. Their mixtures therefore underrepresent
the strongly magnetized sectors and retain a higher test MMD$^2$. Even in its
least-diverse run, cluster initialization instead assigns branches to
complementary negative, central, and positive magnetization sectors; together
they cover the symmetric target distribution and attain the lower test
MMD$^2$.

Together, the two analyses show what turns branch diversity into useful
expressivity. In the mixtures with the lowest test MMD$^2$, the branches do
not merely differ;
they specialize to distinct target modes, image patterns, or physical sectors
and thereby cover complementary parts of the distribution. This target-aligned
specialization is the mode separation that makes the additional branches
effective.

\section{\label{sec:discussion}Discussion}

The results distinguish local trainability from usable expressivity.
Theorem~\ref{thm:rand-bp} and
Theorem~\ref{thm:datadep-bp-main} show that the mixed IQP-QCBM can,
under their stated assumptions, be initialized in a neighborhood without a
local barren plateau,
whereas Proposition~\ref{prop:diversity} and
Theorem~\ref{thm:grad-suppress-main} specify whether the expressive directions
introduced by the branches are accessible. Branch diversity is
necessary to improve below the MMD of the best ancilla-free circuit, but diversity alone
is insufficient: the branches must acquire complementary, target-relevant
roles. The data-agnostic MNIST results illustrate this distinction, since
several branches become different without learning recognizable data
structure.

The trainability guarantees are local and conditional. They neither
characterize the full loss landscape nor guarantee convergence to a useful
solution. At $n=16$ we evaluate every subset and the exact center curvature. 
On MNIST and D-Wave, the available diagnostics describe only typical
low-body residuals at a single system size. They do not establish the
$n$-independent constant required by Assumption~\ref{as:factor}
(Appendix~\ref{app:assumption-check}). The MNIST comparison makes the remaining limitation
concrete: globally initialized branches specialize successfully, whereas
data-agnostic branches train more slowly and can converge to uninformative
distributions. Avoiding a local barren plateau does not determine which
minimum is reached by training.

A limitation shared by the theory and the numerical evaluation is the
emphasis of the loss on low-body statistics. The Hamming Gaussian kernel is
characteristic and formally contains correlators of every body order, but
$\bar m=O(1)$ concentrates most of its weight on low-body subsets. This
weighting is an inductive bias in the sense studied for kernel
methods~\cite{kubler2021inductive}: what the model fits well is set by the
match between the emphasized correlators and the structure of the target.
Accordingly, a small estimated MMD$^2$ primarily reflects agreement in the
low-body statistics emphasized by the kernel, rather than guaranteeing
agreement of the full output distribution. Examining the individual branches
in Sec.~\ref{sec:exp-modesep} partly addresses this limitation by showing whether
the learned differences correspond to recognizable target structure.

Increasing the number of branches also raises the question of how much
complexity remains within each branch. The binary-blobs benchmark is an
extreme case: because the target is itself a mixture of product distributions
[Eq.~\eqref{eq:blobs}], one branch per mode reproduces it without two-body
interactions, making the resulting branches classically simple. MNIST and the
spin-glass benchmark do not show this reduction: removing all trained
two-body angles increases their test MMD$^2$ by factors of five to eight
(Appendix~\ref{app:zero2b}), showing that the fitted distributions use
non-product structure.
Whether this structure is sufficient to retain a classical sampling
separation remains open, as does the broader trade-off between mixture size,
modeling power, and sampling complexity.

Finally, cluster initialization is one mechanism for creating branch
specialization, not a necessary condition for it. It fixes the direction of
specialization through a data partition before optimization, and that
partition need not be optimal for the IQP model. The MNIST results show that
globally initialized branches can also develop distinct digit-like
distributions during training. A natural next step is therefore to design
objectives or training rules that promote target-aligned specialization
dynamically rather than prescribing it entirely at initialization, in the
spirit of the load-balancing objectives that promote expert specialization in
classical sparsely gated mixtures of experts~\cite{shazeer2017outrageously}. The branch
decomposition is also orthogonal to generalizations of the classically
trainable circuit family itself: spectral Born machines extend the IQP Born
machine through group Fourier analysis toward integer-structured
data~\cite{huang2026spectral}, and whether the mixture mechanism studied here
carries over to such models is open. More
generally, the practical value of the mixed IQP-QCBM depends not on branch
count alone, but on whether optimization assigns its branches complementary
roles in the target distribution.

\section{\label{sec:conclusion}Conclusion}

Viewing the mixed IQP-QCBM as a weighted mixture of ancilla-free branches
reveals two distinct challenges. First, the model must be initialized in a
region locally free of barren plateaus. Second, training must make the branches
represent different distributions so that the mixture can outperform a single
ancilla-free circuit. For $L=O(\operatorname{poly}(n))$ branches, we establish locally
trainable neighborhoods around the data-agnostic initialization and, under
explicit data-dependent conditions, around the global and cluster
initializations. We further show that achieving a lower MMD than the best
ancilla-free circuit requires nonzero branch diversity.

Cluster initialization provides branch separation from the start by assigning
different data clusters to different branches. Exact calculations on the two
$n=16$ benchmarks, the binary blobs and the 2D Ising model, reproduce the
predicted $L^{-2}$ curvature scaling and show
that the branch-separating gradient vanishes at exact coincidence and grows
linearly with the branch spread nearby. Across binary blobs, a
two-dimensional Ising model, binarized MNIST, and a $484$-spin glass,
cluster initialization either gives the fastest convergence and the lowest mean
test MMD$^2$, or matches it within the observed seed variation. Examining
the individual branches further shows that the mixtures with the lowest test MMD$^2$
are those whose branches cover distinct image patterns, magnetization
sectors, or digit structures. MNIST also demonstrates that the same
specialization can emerge dynamically from a perturbed global start, whereas
diversity that is not aligned with the data does not necessarily improve the model.

The practical value of a mixed IQP-QCBM therefore lies not in the branch count alone, but in whether its branches acquire complementary, target-aligned roles. Cluster initialization offers one way to provide this specialization from the outset, without relying on the weak branch-separating gradients near coincidence. More broadly, our results suggest that mixed IQP-QCBMs should be designed not only to avoid barren plateaus, but also to create and preserve branch-specific learning signals. Adaptive data partitioning, diversity-promoting objectives, and trainable mixture weights are natural directions for making such specialization emerge more reliably during training. An important next step is to examine whether these benefits persist when trained mixed IQP-QCBMs are deployed at larger qubit counts on quantum hardware, where finite sampling and device noise may affect the learned branch structure.

\begin{acknowledgments}
The authors thank the Quantum AI Team at NORMA Inc. for support, and Hyun-Soo Kim for helpful discussions and feedback. Large language models (Anthropic Claude and OpenAI GPT) were used to assist with drafting and editing the manuscript and with the
development of the analysis code; all AI-assisted text, code, and results
were reviewed and verified by the authors, who take full responsibility
for the content of this work.
\end{acknowledgments}

\renewcommand{\bibsection}{\section*{References}}
\bibliography{main}

\clearpage
\onecolumngrid
\appendix

\section*{Contents of the appendices}
\begingroup
\setlength{\parindent}{0pt}\setlength{\parskip}{1.5pt}
\newcommand{\apptocA}[2]{\textbf{Appendix~\ref{#1}: #2}\nobreak\leaders\hbox to .6em{\hss.\hss}\hfill\textbf{\pageref{#1}}\par}
\newcommand{\apptocB}[2]{\hspace*{1.4em}\ref{#1}.\ #2\nobreak\leaders\hbox to .6em{\hss.\hss}\hfill\pageref{#1}\par}
\apptocA{app:estimator}{Model, estimator, and deployment}
\apptocB{app:correlators}{Randomized estimation of IQP correlators}
\apptocB{app:ciqp}{Compilation to the ancilla circuit}
\apptocB{app:train-embed}{Branch weights and deployment equivalence}
\smallskip
\apptocA{app:bp}{Trainability of the mixed IQP-QCBM: proofs}
\apptocB{app:bp-curvview}{Branch-wise view of the loss and its curvature}
\apptocB{app:bp-lemma}{Trainability lemma: center curvature implies inverse-polynomial loss variance}
\apptocB{app:bp-agnostic}{Data-agnostic center: inverse-polynomial curvature}
\apptocB{app:bp-datadep}{Data-dependent center: curvature}
\apptocB{app:assumption-check}{Finite-size diagnostics of grouped data}
\smallskip
\apptocA{app:expressivity}{Expressivity of the mixture}
\apptocB{app:universal}{Universality via a trivial mixture}
\apptocB{app:diversity}{Branch diversity is necessary to surpass the best ancilla-free circuit}
\smallskip
\apptocA{app:experiments}{Details of the numerical experiments}
\apptocB{app:datasets}{Dataset details}
\apptocB{app:hyperparams}{IQP-QCBM hyperparameters}
\apptocB{app:branch-weights}{Branch weights and clustering settings}
\apptocB{app:train-details}{Classical RBM baseline}
\apptocB{app:theory-check}{Protocol for the numerical tests of the theorems}
\apptocB{app:zero2b}{Ablation of the trained two-body angles}
\apptocB{app:numerical-bandwidth}{MMD bandwidth dependence}
\smallskip
\endgroup
\clearpage

\section{\label{app:estimator}Model, estimator, and deployment}

The symbols used throughout the appendices are collected in
Table~\ref{tab:notation}.

\begin{table*}[!ht]
\caption{\label{tab:notation}Notation used throughout the appendices. Generators
$G$ are gate supports (a vertex or edge of the interaction graph); the ancilla
count $a$ is distinct from the operator subset $A$, and the randomization samples
$\bm s,\bm u$ are distinct from the data Born string $\bm z$.}
\renewcommand{\arraystretch}{1.5}%
\begin{ruledtabular}
\begin{tabular}{l p{0.78\textwidth}}
\multicolumn{2}{c}{\textbf{Core objects}}\\
\colrule
Symbol & Definition\\
\colrule
$[n]$ & Visible-qubit index set $\{1,\dots,n\}$.\\
$A,\ A^c$ & A qubit subset (the support of a correlator) and its complement $[n]\setminus A$.\\
$Z_A$ & Pauli-$Z$ word on $A$, the product $\bigotimes_{i\in A}Z_i$ of single-qubit $Z_i$.\\
$X_G$ & Pauli-$X$ word on a generator $G$, the product $\bigotimes_{j\in G}X_j$.\\
$U(\bm\theta)$ & Single-IQP circuit $\prod_G e^{i\theta_G X_G}$ applied to the input $\ket{0}^{\otimes n}$.\\
$\theta_G$ & Angle of generator $G$.\\
$z_A(\bm\theta)$ & Model correlator $\bra{0}U^\dagger Z_A U\ket{0}$, valued in $[-1,1]$; for a single qubit $z_j\equiv z_{\{j\}}$.\\
$t_A$ & Data moment $\mathbb{E}_{\bm x\sim p_{\mathrm{data}}}[(-1)^{\sum_{i\in A}x_i}]$; for a single qubit $t_j\equiv t_{\{j\}}$.\\
$t_A^{(\ell)}$ & Data moment of training group $C_\ell$, $\langle Z_A\rangle_{C_\ell}$; one-body case $t_j^{(\ell)}$.\\
$g_A^{(G)}$ & Gradient of the correlator with respect to $\theta_G$, that is $\partial z_A/\partial\theta_G$ (written $g_A$ when unambiguous).\\
$\mathcal{L}$ & Low-body MMD loss, the weighted sum $\sum_A w_A(z_A-t_A)^2$.\\
$w_A$ & MMD weight for $A$, equal to $(1-p_\sigma)^{n-|A|}p_\sigma^{|A|}$; depends only on $|A|$.\\
$p_\sigma$ & Kernel parameter $\tfrac12(1-e^{-1/(2\sigma^2)})$; equal to $\Theta(1/n)$ at $\sigma=\Theta(\sqrt n)$.\\
$\Theta^\star,\ r$ & Initialization center and patch half-width, with $\Theta\sim\Theta^\star+\mathrm{Unif}(-r,r)$.\\
\colrule
\multicolumn{2}{c}{\textbf{Mixture and ancilla}}\\
\colrule
$a,\ L$ & Ancilla count and number of branches, related by $L=2^a$.\\
$\ell,\ \pi_\ell$ & Branch label $\ell\in\{0,1\}^a$ and its weight; $\sum_\ell\pi_\ell=1$, uniform $\pi_\ell=1/L$.\\
$\bm\theta^{(\ell)}$ & Angle vector of branch $\ell$.\\
$\Theta$ & Full mixture parameter, the collection $(\bm\theta^{(0)},\dots,\bm\theta^{(L-1)})$.\\
$m,\ m_0$ & Total number of trainable angles, $m=L\,m_0$ with $m_0$ per branch.\\
$z_A(\Theta)$ & Mixture correlator, the branch average $\sum_\ell\pi_\ell\,z_A(\bm\theta^{(\ell)})$.\\
$\widetilde\theta_{G,S}$ & Compiled angle for ancilla subset $S\subseteq[a]$ of the $(n+a)$-qubit mixed IQP (Appendix~\ref{app:ciqp}).\\
\colrule
\multicolumn{2}{c}{\textbf{Graph and generators}}\\
\colrule
$E$ & Edge set: the two-qubit interactions allowed by the connectivity graph.\\
$G$ & Generator: the qubits a gate acts on, a vertex or edge of $E$.\\
$G\cdot A,\ G\cdot\bm s$ & Binary inner products $|G\cap A|$ and $\sum_{j\in G}s_j$, both taken mod $2$.\\
$\bm s,\ \bm u$ & Uniform Monte-Carlo randomization samples on the system and ancilla registers.\\
$\bm z$ & Data Born string in $\{0,1\}^n$ [Eq.~\eqref{eq:born}].\\
\end{tabular}
\end{ruledtabular}
\end{table*}

\emph{$X$-basis setup (used by all subsections).} Every circuit below is of IQP
form: a product of commuting Pauli-$X$ rotations applied to $\ket{0}^{\otimes\cdot}$
and measured in the computational basis. We evaluate its correlators directly in
the eigenbasis of these gates, the Pauli-$X$ basis, without converting the gates to
$Z$-rotations. Writing a single branch as a product over its commuting generators,
\begin{equation}
U(\bm{\theta})=\prod_{G}\exp\!\big(i\,\theta_G X_G\big),
\qquad
X_G=\bigotimes_{j\in G}X_j,
\label{eq:estA-unitary}
\end{equation}
each generator $G$ is the set of qubits its gate acts on (a vertex or edge of the
interaction graph), so $\theta_G$ collects the one- and two-body angles
$\theta_j,\theta_{jk}$ of Eq.~\eqref{eq:iqp}; throughout, $G\cdot\bm s$ denotes the
parity $\sum_{j\in G}s_j$ and $G\cdot A$ the overlap $|G\cap A|$, both modulo two. The shared input is the all-zeros state expanded over the
$X$-eigenbasis $\ket{\bm s}_X=H^{\otimes n}\ket{\bm s}$, on which the generators and
the measured operator act as
\begin{equation}
\ket{0}^{\otimes n}=2^{-n/2}\sum_{\bm s\in\{0,1\}^n}\ket{\bm s}_X,
\qquad
X_G\ket{\bm s}_X=(-1)^{G\cdot\bm s}\ket{\bm s}_X,
\qquad
Z_A\ket{\bm s}_X=\ket{\bm s\oplus A}_X,
\label{eq:estA-xbasis}
\end{equation}
the first because $\braket{\pm|0}=2^{-1/2}$ on every qubit, the second because each
$X_j$ has eigenvalue $(-1)^{s_j}$ on $\ket{s_j}_X$, and the third because each
$Z_j$ exchanges $\ket{+}\leftrightarrow\ket{-}$. These three relations are reused
by the compiled circuit of Appendix~\ref{app:ciqp} under the substitution
$n\to n+a$, on which the mixed IQP is again an IQP-form circuit.

\subsection{\label{app:est-single}\label{app:correlators}Randomized estimation of IQP correlators}

This subsection records the classical randomized estimator that evaluates one
branch correlator $z_A(\bm\theta^{(\ell)})$, the estimator referred to as
Ref.~\cite{vandennest2011} in the main text. To lighten notation we
drop the branch label within this subsection and write $\bm\theta$ for a single
branch's angles; the result applies to each $\bm\theta^{(\ell)}$ in turn.

\emph{Reduction to an average of cosines.} The branch correlator is the
$X$-basis expectation of $Z_A$ after the measured operator is Heisenberg-evolved
through the commuting gates. Conjugating $Z_A$ by one rotation passes it through
unchanged or doubles its angle, which the factor $1-(-1)^{G\cdot A}\in\{0,2\}$
records in a single line,
\begin{equation}
\exp\!\big(-i\theta_G X_G\big)\,Z_A\,\exp\!\big(i\theta_G X_G\big)
=\exp\!\big(-i\theta_G\,(1-(-1)^{G\cdot A})\,X_G\big)\,Z_A .
\label{eq:estA-heisenberg}
\end{equation}
The generators commute, so evolving through all of them gives
\begin{equation}
U(\bm\theta)^\dagger Z_A\, U(\bm\theta)
=\Big(\prod_{G}\exp\!\big(-i\theta_G\,(1-(-1)^{G\cdot A})\,X_G\big)\Big)Z_A .
\label{eq:estA-evolved}
\end{equation}
Evaluating $\bra{0}^{\otimes n}\!\cdots\!\ket{0}^{\otimes n}$ in the $X$-eigenbasis
of Eq.~\eqref{eq:estA-xbasis} is then immediate: $Z_A$ leaves $\ket{0}^{\otimes n}$
invariant (it only relabels the uniform superposition by $\bm s\mapsto\bm s\oplus A$),
and the surviving exponential is diagonal on $\ket{\bm s}_X$. Collecting its phase
into the per-sample integrand
\begin{equation}
f_A(\bm\theta,\bm s)
=\cos\!\Big(\sum_{G}\theta_G\,(1-(-1)^{G\cdot A})\,(-1)^{G\cdot\bm s}\Big),
\label{eq:estA-integrand}
\end{equation}
the involution $\bm s\mapsto\bm s\oplus A$ cancels the imaginary part termwise, so the
correlator is the uniform average of $f_A$,
\begin{equation}
\boxed{\;
z_{A}(\bm\theta)
=\mathbb{E}_{\bm s\sim\mathrm{Unif}(\{0,1\}^n)}\big[f_A(\bm\theta,\bm s)\big].
\;}
\label{eq:estA-estimator}
\end{equation}
Equation~\eqref{eq:estA-estimator} is the randomized estimand of
Ref.~\cite{vandennest2011}.

\emph{Unbiased sample-mean estimator.} Averaging $f_A$ over $M$ independent uniform
samples $\bm s^{(1)},\dots,\bm s^{(M)}$ gives the unbiased Monte-Carlo estimator
\begin{equation}
\widehat{z}_{A}(\bm\theta)
=\frac{1}{M}\sum_{m=1}^{M}f_A(\bm\theta,\bm s^{(m)}),
\qquad
\mathbb{E}\big[\widehat{z}_{A}\big]=z_{A}(\bm\theta).
\label{eq:estA-samplemean}
\end{equation}
Since $f_A\in[-1,1]$, the variance is at most $1/M$ per operator and the estimate
concentrates at the standard $O(M^{-1/2})$ rate independent of $n$. Evaluating $f_A$
is one pass over the $m_0$ gates, so with $\mathcal{A}$ the set of measured operators
(the subsets $A$ whose weight $w_A$ is above threshold in the loss) the total
classical cost is
\begin{equation}
O\!\big(M\,|\mathcal{A}|\,m_0\big),
\label{eq:estA-cost}
\end{equation}
linear in the number of samples and operators and polynomial in the circuit size.
This is the classical channel through which the MMD loss of Eq.~\eqref{eq:mmd} is
assembled: each correlator $z_A(\bm\theta)$ entering the weighted $\ell_2$ distance
is the estimand of Eq.~\eqref{eq:estA-estimator}, evaluated by the sample mean of
Eq.~\eqref{eq:estA-samplemean}.

\subsection{\label{app:est-ancilla}\label{app:ciqp}Compilation to the ancilla circuit}

This subsection compiles the $L=2^a$ branches into the uniform-weight
$(n+a)$-qubit IQP circuit of Eq.~\eqref{eq:ciqp}, so that one IQP-form circuit
estimates the uniform mixture correlator $z_A(\Theta)$. The derivation reuses
the $X$-basis correlator reduction of the
preamble [Eqs.~\eqref{eq:estA-xbasis}--\eqref{eq:estA-estimator}]
under $n\to n+a$, $\bm s\to(\bm s,\bm u)$, $G\to G\otimes S$, $A\to(A,\bm 0)$, where
$\bm u\in\{0,1\}^a$ are the ancilla randomization bits and $S\subseteq[a]$ an
ancilla subset.

\emph{Compiled circuit and its generators.} Adapting the Walsh--Hadamard mixed IQP
compilation of Ref.~\cite{slim2026} to Eq.~\eqref{eq:iqp}, write
$X_S^{\mathrm{anc}}=\bigotimes_{r\in S}X_{n+r}$ (empty product the identity). The
compiled angles are the inverse Walsh--Hadamard transform of the branch angles,
\begin{equation}
\widetilde{\theta}_{G,S}
=\frac{1}{L}\sum_{\ell=0}^{L-1}(-1)^{S\cdot\ell}\,\theta_G^{(\ell)},
\qquad S\subseteq[a],
\label{eq:estB-angles}
\end{equation}
with $S\cdot\ell$ the binary inner product modulo two. The mixed IQP of
Eq.~\eqref{eq:ciqp} is then the ordinary $(n+a)$-qubit IQP circuit
\begin{equation}
U_{\text{anc-IQP}}(\Theta)
=\prod_{G}\prod_{S\subseteq[a]}
\exp\!\big(i\,\widetilde{\theta}_{G,S}\,X_G^{\mathrm{sys}}\otimes X^{\mathrm{anc}}_{S}\big),
\label{eq:estB-unitary}
\end{equation}
whose generators are the commuting Pauli-$X$ words $X_G^{\mathrm{sys}}\!\otimes\!X^{\mathrm{anc}}_S$
on the joint register.

\emph{$X$-basis block form.} The physical compiled circuit is already
block-diagonal in the eigenbasis of its ancilla generators. The inverse angle
map follows from the Walsh orthogonality relation
\begin{equation}
\sum_{S\subseteq[a]}(-1)^{S\cdot(\ell\oplus\ell')}=L\,\delta_{\ell,\ell'},
\label{eq:estB-walsh-orth}
\end{equation}
which inverts Eq.~\eqref{eq:estB-angles}. On an ancilla $X$-basis state
$\ket{\ell}_{X,\mathrm{anc}}$, the eigenvalue of $X_S^{\mathrm{anc}}$ is
$(-1)^{S\cdot\ell}$, and the effective data angle is
\begin{equation}
\sum_{S\subseteq[a]}(-1)^{S\cdot\ell}\,\widetilde{\theta}_{G,S}
=\theta_G^{(\ell)},
\label{eq:estB-walsh}
\end{equation}
so the physical compiled unitary has the block form
\begin{equation}
U_{\text{anc-IQP}}
=\sum_{\ell=0}^{L-1}U(\bm{\theta}^{(\ell)})_{\mathrm{sys}}\otimes
\ket{\ell}_{X,\mathrm{anc}}\!\bra{\ell}_{X,\mathrm{anc}},
\label{eq:estB-block}
\end{equation}
the branch-controlled form of Eq.~\eqref{eq:ciqp}.

\emph{Joint randomized estimator.} The mixed IQP is a single IQP circuit on $n+a$
qubits with input $\ket{0}^{\otimes(n+a)}$, so Eq.~\eqref{eq:estA-estimator}
applies under the substitutions above; a compiled generator anticommutes with
$Z_A\otimes I^{\mathrm{anc}}$ iff $G\cdot A$ is odd, a condition on $G$ alone. The
joint estimand is
\begin{equation}
z_{A}(\Theta)
=\mathbb{E}_{(\bm s,\bm u)\sim\mathrm{Unif}(\{0,1\}^{n}\times\{0,1\}^{a})}
\Big[\cos\!\Big(\sum_{G}(1-(-1)^{G\cdot A})\!\!\sum_{S\subseteq[a]}\widetilde{\theta}_{G,S}\,
(-1)^{G\cdot\bm s}(-1)^{S\cdot\bm u}\Big)\Big],
\label{eq:estB-estimator}
\end{equation}
with unbiased Monte-Carlo estimator over $M$ uniform joint draws $(\bm s^{(m)},\bm
u^{(m)})$, the compiled counterpart of Eq.~\eqref{eq:estA-samplemean},
\begin{equation}
\widehat{z}^{\,\mathrm{mix}}_{A}(\Theta)
=\frac{1}{M}\sum_{m=1}^{M}
\cos\!\Big(\sum_{G}(1-(-1)^{G\cdot A})\!\!\sum_{S\subseteq[a]}\widetilde{\theta}_{G,S}\,
(-1)^{G\cdot\bm s^{(m)}}(-1)^{S\cdot\bm u^{(m)}}\Big),
\qquad
\mathbb{E}\big[\widehat{z}^{\,\mathrm{mix}}_{A}\big]=z_{A}(\Theta).
\label{eq:estB-samplemean}
\end{equation}
Measuring $Z_A$ on the data register of the mixed IQP (identity on the ancillas)
therefore estimates the mixture correlator $z_A(\Theta)$ of the $n$-qubit data
marginal.

\emph{Reduction to a per-branch object.} Fixing the ancilla draw $\bm u=\ell$
collapses the inner Walsh sum over $S$ via Eq.~\eqref{eq:estB-walsh},
\begin{equation}
\sum_{S\subseteq[a]}\widetilde{\theta}_{G,S}(-1)^{S\cdot\ell}=\theta_G^{(\ell)},
\end{equation}
the effective data angle of branch $\ell$. Conditioning on the uniform ancilla
bits, the joint estimand factorizes into an outer uniform average over $\ell$ of
the inner data-only average,
\begin{align}
z_{A}(\Theta)
&=\frac{1}{2^a}\sum_{\ell\in\{0,1\}^a}
\mathbb{E}_{\bm s\sim\mathrm{Unif}(\{0,1\}^n)}
\big[f_A(\bm\theta^{(\ell)},\bm s)\big]
\nonumber\\
&=\frac{1}{L}\sum_{\ell=0}^{L-1}z_{A}(\bm\theta^{(\ell)}),
\qquad L=2^a,
\label{eq:estB-perbranch}
\end{align}
where the inner expectation is exactly the single-branch estimand
Eq.~\eqref{eq:estA-estimator} for branch $\ell$. The joint mixed IQP estimator thus
reduces to the uniform branch average, recovering the mixture correlator
$z_A(\Theta)$ of Eq.~\eqref{eq:moiqp-corr} from one $(n+a)$-qubit IQP-form circuit.

\subsection{\label{app:assemble}\label{app:train-embed}Branch weights and deployment equivalence}

This subsection assembles the two estimators into the deployment statement used in
the main text: the trained mixed IQP marginal, the weighted mixture, and a decomposed
randomized estimator that never builds the joint register all share one estimand.
It fixes the branch weights and proves the estimator equivalence.

\emph{(i) Compilation and marginal.} An indexed family of $L=2^a$ independent
$n$-qubit IQP circuits with branch angles $\{\theta_G^{(\ell)}\}_{\ell=0}^{L-1}$ is
mapped to a single uniform-weight $(n+a)$-qubit compiled IQP circuit through
the inverse Walsh--Hadamard angle
map of Eq.~\eqref{eq:estB-angles}, following Ref.~\cite{slim2026}. The empty subset
$S=\varnothing$ collects the branch-averaged angle and the nonempty
subsets carry the branch-distinguishing content. The physical compiled
unitary has the $X$-basis block form of Eq.~\eqref{eq:estB-block} and is run
from the standard input
$\ket{0}_{\mathrm{sys}}^{\otimes n}\ket{0}_{\mathrm{anc}}^{\otimes a}$.
Because
$\ket{0}_{\mathrm{anc}}^{\otimes a}
=L^{-1/2}\sum_\ell\ket{\ell}_{X,\mathrm{anc}}$, every block has equal
amplitude. Tracing out the ancillas removes the off-diagonal terms and yields
the $n$-qubit data marginal
\begin{equation}
p_\Theta(\bm z)=\frac{1}{L}\sum_{\ell=0}^{L-1}
|\braket{\bm z|U(\bm\theta^{(\ell)})|0^{\otimes n}}|^2,
\label{eq:estC-marginal}
\end{equation}
the uniform mixture of Eq.~\eqref{eq:moiqp}. The cluster-initialized scheme of
Sec.~\ref{sec:init-dependent} fits one independent IQP circuit per data group to
obtain the $\{\theta_G^{(\ell)}\}$ and then applies Eq.~\eqref{eq:estB-angles} to
assemble the compiled branch-angle parameterization that joint training
refines; this generalizes the single-circuit data-dependent construction of
Ref.~\cite{recio2025} from $L=1$ to an arbitrary indexed mixture.

\emph{Hardness caveat.} The joint $(n+a)$-qubit output of
Eq.~\eqref{eq:estB-unitary} is IQP-form, and it is this joint object to which the
commuting-circuit sampling-hardness
arguments~\cite{shepherd2009,bremner2011,bremner2025stabilizer} apply. The
generative model uses only
the $n$-qubit data marginal of Eq.~\eqref{eq:estC-marginal}; joint hardness does
not by itself imply hardness of that marginal, since marginalization can reduce
distributional complexity. The construction proves exact deployment equivalence of
the uniform mixed marginal to one IQP-form circuit, not a separate worst-case
hardness theorem for the marginal. Hardness is moreover a property of the
circuit ensemble rather than of an individual circuit: IQP instances with
special structure admit fast classical simulation, as demonstrated for the
Harvard/QuEra logical-processor IQP circuits~\cite{maslov2024fast}, so
membership in the IQP family supports no per-circuit hardness claim. The
nonuniform extension below has an
exact branch-controlled or randomized-routing realization, but we do not
identify preparation of its general ancilla state with the ordinary IQP gate
set.

\emph{(ii) Branch weights.} Replacing the standard ancilla input
$\ket{0}^{\otimes a}$ by the $X$-basis amplitude state
$\ket{\chi_{\bm\pi}}_{\mathrm{anc}}$ of
Eq.~\eqref{eq:ancilla-weights} leaves every branch unitary
$U(\bm\theta^{(\ell)})$ untouched and turns the data marginal into the
\emph{weighted} mixture
$\sum_\ell\pi_\ell\,|\braket{\bm z|U(\bm\theta^{(\ell)})|0^{\otimes n}}|^2$. The weights are
set by the cluster assignment: the cluster-initialized scheme sets $\pi_\ell$ to the
mass of group $\ell$ (Sec.~\ref{sec:init}), recovering $\pi_\ell=1/L$ when the groups are balanced. They
change only how the branch outputs are combined, not any branch unitary. In
the estimator, each branch contributes a number of correlator samples
$n_\ell\propto\pi_\ell$ with
$\sum_\ell n_\ell$ equal to the per-step sample budget, reweighting the branch
average without altering any branch angle. The ordinary compiled IQP of
part~(i) is the special case $\pi_\ell=1/L$ with $L=2^a$. For general
$\bm\pi$, preparing $\ket{\chi_{\bm\pi}}$ is a separate state-preparation
requirement; equivalently, one
can draw the branch index classically with probability $\pi_\ell$. A single
trained set of compiled angles therefore realizes any configured weights
without retraining the angles, but only the uniform case is identified here
with the ordinary compiled IQP circuit. The
weighted data marginal
\begin{equation}
p_\Theta(\bm z)=\sum_\ell\pi_\ell\,
|\braket{\bm z|U(\bm\theta^{(\ell)})|0^{\otimes n}}|^2
\label{eq:estC-weighted-marginal}
\end{equation}
above is Eq.~\eqref{eq:estB-block} evaluated on
$\ket{\chi_{\bm\pi}}$, and
the data correlator
\begin{equation}
z_A(\Theta)=\sum_{\bm z}p_\Theta(\bm z)(-1)^{\sum_{j\in A}z_j}
\end{equation}
is linear in $p_\Theta$. The branch-weighted average therefore follows by
linearity, so the mixed IQP data correlator is
\begin{equation}
z_{A}(\Theta)=\sum_{\ell=0}^{L-1}\pi_\ell\,z_{A}(\bm\theta^{(\ell)}),
\label{eq:moiqp-corr}
\end{equation}

\emph{(iii) Estimator equivalence.} We now show that the compiled estimator of
Appendix~\ref{app:ciqp} and a \emph{decomposed} estimator that never builds the
joint register share one estimand. Fix the per-step budget of $M$ correlator
samples and define the decomposed procedure:
\begin{enumerate}
\item choose integer counts $n_\ell$ for the active branches with
      $\sum_{\ell:\pi_\ell>0} n_\ell=M$ and $n_\ell\ge1$ whenever $\pi_\ell>0$
      (in practice largest-remainder rounding of $\pi_\ell M$);
\item for each active branch $\ell$, draw $n_\ell$ auxiliary bit strings
      $\bm s\sim\mathrm{Unif}(\{0,1\}^n)$ and form the single-circuit unbiased
      estimator $\widehat{z}_{A}(\bm\theta^{(\ell)})$ of
      Eq.~\eqref{eq:estA-samplemean};
\item combine as the weighted mean $\widehat{z}^{\,\mathrm{dec}}_{A}
      =\sum_\ell\pi_\ell\,\widehat{z}_{A}(\bm\theta^{(\ell)})$.
\end{enumerate}
The $\pi_\ell$ weighting enters exactly once, in step~3; the allocation of step~1
only distributes the budget and applies no second weighting. Equivalently one may
draw a branch index with probability $\pi_\ell$ per sample and average the
resulting single-branch cosine samples directly, using sampling frequency rather
than a post-factor to realize the same weights. For the fixed-allocation estimator,
unbiasedness of Eq.~\eqref{eq:estA-samplemean} gives
\begin{equation}
\mathbb{E}\big[\widehat{z}^{\,\mathrm{dec}}_{A}\big]
=\sum_{\ell}\pi_\ell\,
\mathbb{E}\big[\widehat{z}_{A}(\bm\theta^{(\ell)})\big]
=\sum_{\ell}\pi_\ell\,z_{A}(\bm\theta^{(\ell)}),
\label{eq:estC-dec-exp}
\end{equation}
for any such allocation. The joint branch-controlled estimator
$\widehat{z}^{\,\mathrm{mix}}_{A}$ of Eq.~\eqref{eq:estB-samplemean}, evaluated
with the weighted ancilla state of part~(ii), has expectation
$z_A(\Theta)=\sum_\ell\pi_\ell\,z_A(\bm\theta^{(\ell)})$ by
Eq.~\eqref{eq:moiqp-corr}. The two right-hand sides coincide,
\begin{equation}
\boxed{\;
\mathbb{E}\big[\widehat{z}^{\,\mathrm{dec}}_{A}\big]
=\sum_{\ell=0}^{L-1}\pi_\ell\,z_{A}(\bm\theta^{(\ell)})
=z_{A}(\Theta)
=\mathbb{E}\big[\widehat{z}^{\,\mathrm{mix}}_{A}\big].
\;}
\label{eq:estC-equiv}
\end{equation}
Equation~\eqref{eq:estC-equiv} is exact at the level of estimands, resting only on
the linearity of the mixture correlator [Eq.~\eqref{eq:moiqp-corr}] and the
unbiasedness of the single-circuit estimator [Eq.~\eqref{eq:estA-samplemean}]. This
justifies the randomized deployment of Sec.~\ref{sec:rand-iqp}: drawing a branch
$\ell$ with probability $\pi_\ell$ and running one $n$-qubit circuit
$U(\bm\theta^{(\ell)})$ reproduces the trained mixed IQP data marginal, with no ancilla
register executed at deployment.

\emph{(iv) Derivatives transfer between coordinates.} The optimizer updates the
compiled angles $\widetilde{\theta}_{G,S}$, whereas the correlators and their
derivatives are most easily evaluated per branch, in the angles
$\theta_G^{(\ell)}$ of parts~(i)--(iii). The two are related gate by gate by the
bijection inverse to Eq.~\eqref{eq:estB-angles}: stacking one gate's branch angles
$\bm{\theta}_G=(\theta_G^{(\ell)})_\ell$ and compiled coefficients
$\widetilde{\bm{\theta}}_G=(\widetilde{\theta}_{G,S})_S$,
\begin{equation}
\bm{\theta}_G=H\,\widetilde{\bm{\theta}}_G,
\qquad H_{\ell,S}=(-1)^{S\cdot\ell},
\label{eq:coord-map}
\end{equation}
with $H$ the $L\times L$ Walsh--Hadamard matrix. Distinct gates use disjoint
blocks, so the global map is block-diagonal with one identical $H$ per gate. By
Walsh orthogonality $(H^2)_{\ell,\ell'}=\sum_S(-1)^{S\cdot(\ell\oplus\ell')}
=L\,\delta_{\ell,\ell'}$,
\begin{equation}
H=H^{\mathsf T},\qquad H^2=L\,\mathbb{I},\qquad \mathrm{cond}(H)=1
\label{eq:coord-Hprops}
\end{equation}
($Q=H/\sqrt{L}$ is orthogonal). The loss is the same function of the model in both
coordinate systems, so by the chain rule its gradient and Hessian transform as
\begin{equation}
\nabla_{\widetilde{\bm{\theta}}}\mathcal{L}
=H^{\mathsf T}\nabla_{\bm{\theta}}\mathcal{L},
\qquad
\nabla^2_{\widetilde{\bm{\theta}}}\mathcal{L}
=H^{\mathsf T}\big(\nabla^2_{\bm{\theta}}\mathcal{L}\big)H,
\label{eq:coord-chain}
\end{equation}
block-diagonally over gates. Because $\mathrm{cond}(H)=1$ the transfer is perfectly
conditioned: a gradient or curvature computed in the per-branch
coordinates, where each $z_A(\bm{\theta}^{(\ell)})$ and its derivatives are the
single-circuit estimands of part~(i), maps to the compiled training coordinates
without loss, up to the explicit factor $H$. The trainability proof of
Appendix~\ref{app:bp} uses this transfer to carry the per-branch guarantee to
the compiled angles the optimizer updates, as the push-forward of the
per-branch patch (Lemma~\ref{lem:coord-transfer}).

\clearpage

\section{\label{app:bp}Trainability of the mixed IQP-QCBM: proofs}

This part establishes that the mixed IQP-QCBM inherits the trainability of
its single-circuit constituents. We isolate the only model-dependent input to the
argument, a nonvanishing curvature at the chosen center, and prove once and for
all that any such curvature, regardless of how it is produced, forces an
inverse-polynomial loss variance on an inverse-polynomial patch and hence
rules out exponential local concentration of the loss
(see Lemma~\ref{lem:curv-bp} below). The two centers of interest, the
data-agnostic center (Theorem~\ref{thm:rand-bp}) and the data-dependent
center (Appendix~\ref{app:bp-datadep}), then reduce to a single curvature
evaluation each. To be clear about provenance: the curvature-to-variance
conversion and its explicit constants are Theorem~2 of Ref.~\cite{lerch2026},
invoked here rather than re-derived; what is new is the verification that the
$L$-branch mixture preserves that theorem's hypotheses, the constants changing by
only a $\mathrm{poly}(L)$ factor ($\pi_\ell\le1$, $m=Lm_0$), so its
$\Omega(1/\mathrm{poly}(n))$ variance bound carries over.

\subsection{\label{app:bp-curvview}Branch-wise view of the loss and its curvature}

We work throughout with the low-body MMD loss
\begin{equation}
\mathcal{L}(\Theta)=\sum_{A\subseteq[n]} w_A\big(z_A(\Theta)-t_A\big)^2,
\qquad
z_A(\Theta)=\sum_{\ell=0}^{L-1}\pi_\ell\,z_A(\bm{\theta}^{(\ell)}),
\label{eq:bp-loss}
\end{equation}
with weights $\pi_\ell$ on the simplex ($\pi_\ell=1/L$ for the standard
all-zero ancilla input), where the second equality is the expectation-value linearity of
Eq.~\eqref{eq:moiqp-corr} (Appendix~\ref{app:train-embed}), each branch correlator
obeys $z_A(\bm{\theta}^{(\ell)})\in[-1,1]$ as the expectation of a
$\pm1$-eigenvalue Pauli word, and $w_A=\Theta(n^{-|A|})$ is the low-body weight at
$\sigma=\Theta(\sqrt{n})$.

\emph{Single-circuit derivative regularity.} The one model-dependent input the
trainability lemma needs is the smoothness of the branch correlators. Each
$z_A(\bm\theta^{(\ell)})$ is an average of cosines [Eq.~\eqref{eq:estA-estimator}]
with argument linear in every angle, hence smooth with uniformly bounded
derivatives,
\begin{equation}
\big|\partial^q z_A\big|\le 2^q,
\qquad
\partial^2_{\theta_G} z_A=-2\big(1-(-1)^{G\cdot A}\big)\,z_A,
\label{eq:estA-derivbound}
\end{equation}
the second equal to $-4z_A$ for an angle that acts on $z_A$ and to $0$ otherwise,
where $z_A$ is independent of $\theta_G$ when $G\cdot A=0$~\cite{lerch2026}. These per-branch
bounds are the regularity input carried through the lemma below.

\emph{Single-parameter curvature decomposes into a model-sensitivity term and a
data-mismatch term.} A parameter $\theta_G^{(\ell)}$ enters only the
$\ell$-th branch, so
$\partial_{\theta_G^{(\ell)}}z_A(\Theta)=\pi_\ell\,g_A^{(G,\ell)}$
with the branch sensitivity $g_A^{(G,\ell)}:=\partial_{\theta_G^{(\ell)}}
z_A(\bm{\theta}^{(\ell)})$. Differentiating Eq.~\eqref{eq:bp-loss} twice and using
the IQP identities $\partial^2_{\theta_G}z_A=-4\,z_A$ for $G\cdot A=1$
~\cite{lerch2026}, with both $g_A^{(G,\ell)}$ and
$\partial^2_{\theta_G}z_A$ vanishing unless $G\cdot A=1$, only those subsets
contribute (for a one-body generator $G=\{j\}$, this condition is $j\in A$,
which is the case both theorems
below use):
\begin{equation}
\partial^2_{\theta_G^{(\ell)}}\mathcal{L}
=\sum_{A:\,G\cdot A=1} 2w_A\!\left[
 \underbrace{\pi_\ell^2\big(g_A^{(G,\ell)}\big)^2}_{\text{model sensitivity}}
 +\underbrace{4\pi_\ell\,z_A(\bm{\theta}^{(\ell)})\big(t_A-z_A(\Theta)\big)}_{\text{data mismatch}}
 \right].
\label{eq:bp-curv}
\end{equation}
The two terms are the model-sensitivity and data-mismatch contributions of
Eq.~\eqref{eq:curv-decomp}, here resolved per branch.

\emph{The data-mismatch term contains genuinely new inter-branch cross terms.}
Expanding the mismatch with $z_A(\Theta)=\sum_{\ell'}\pi_{\ell'}z_A(\bm{\theta}^{(\ell')})$,
\begin{equation}
z_A(\bm{\theta}^{(\ell)})\big(t_A-z_A(\Theta)\big)
= z_A(\bm{\theta}^{(\ell)})\,t_A
-\pi_\ell\,z_A(\bm{\theta}^{(\ell)})^2
-\!\sum_{\ell'\neq\ell}\!\pi_{\ell'}\, z_A(\bm{\theta}^{(\ell)})z_A(\bm{\theta}^{(\ell')}).
\label{eq:bp-cross}
\end{equation}
The final sum is the only feature absent from a single circuit: products of
correlators from \emph{distinct} branches $\ell\neq\ell'$. Each product is
uniformly bounded,
$|z_A(\bm{\theta}^{(\ell)})z_A(\bm{\theta}^{(\ell')})|\le1$, exactly like the
diagonal $z_A^2$ term, so the cross terms do not enlarge the derivative bounds
beyond those of one circuit. Boundedness alone secures the regularity hypotheses
of the variance bound but does not by itself prevent concentration; the
low-body weighting $w_A=\Theta(n^{-|A|})$ suppresses high-order subsets so that
the $A$-sum is dominated by $O(1)$-body terms, and the surviving curvature at a
given center is controlled by the model-sensitivity term there. The next
subsection makes this precise: it takes the center curvature as a hypothesis and
returns a polynomial variance lower bound, leaving each center to supply only its
own curvature.

\subsection{\label{app:bp-lemma}Trainability lemma: center curvature implies inverse-polynomial loss variance}

The following lemma is the single trainability statement on which both barren-plateau
theorems rest. Its curvature-to-variance step is Theorem~2 of Ref.~\cite{lerch2026},
invoked here verbatim with its explicit constants rather than re-derived; the only
mixture-specific content is checking that that theorem's regularity hypotheses survive
the passage from a single ancilla-free IQP circuit to the $L$-branch mixture, which they do whenever
$L=O(\mathrm{poly}(n))$. The proof is therefore independent of which center is chosen and
of how the curvature is produced.

\begin{lemma}[Center curvature $\Rightarrow$ inverse-polynomial loss variance for the mixture]
\label{lem:curv-bp}
Let $\mathcal{L}(\Theta)=\sum_{A\subseteq[n]} w_A\big(z_A(\Theta)-t_A\big)^2$ be the
low-body MMD loss [$\sigma=\Theta(\sqrt{n})$, $w_A=\Theta(n^{-|A|})$] of an
$L$-branch mixed IQP-QCBM with $z_A(\Theta)=\sum_\ell\pi_\ell z_A(\bm{\theta}^{(\ell)})$,
where $L=O(\mathrm{poly}(n))$ and each branch is an IQP circuit with
$m_0=\mathrm{poly}(n)$ gates obeying the per-branch derivative bounds
$|\partial^q z_A|\le 2^q$. Suppose at a center $\Theta^\star$ there is a parameter
$\theta_G^{(\ell^\star)}$ with
\[
\big|\partial^2_{\theta_G^{(\ell^\star)}}\mathcal{L}(\Theta^\star)\big|
= c = \Omega\!\left(\frac{1}{\mathrm{poly}(n)}\right).
\]
Then the mixture regularity constants $\beta_1,\beta_2$ are
$\mathrm{poly}(n)$-bounded, and there is an inverse-polynomial patch half-width
$r=1/\mathrm{poly}(n)$ about $\Theta^\star$ on which
\[
\mathrm{Var}_\Theta[\mathcal{L}]
\ge (1-\Delta)\,\frac{r^4}{45}\,c^2
= \Omega\!\left(\frac{1}{\mathrm{poly}(n)}\right),
\]
where $\Delta\in(0,1)$ is a free accuracy parameter selected by the user (it sizes
the admissible patch through $r$; for a fixed admissible $r$ a smaller $\Delta$
improves the prefactor $1-\Delta$ while reducing the maximum admissible
radius), not a model-dependent quantity.
Hence the loss does not concentrate exponentially near $\Theta^\star$, and
the model is locally free of barren plateaus at $\Theta^\star$ in the
loss-variance sense of Sec.~\ref{sec:tr-curv}.
\end{lemma}

\begin{proof}
The argument has two init-independent parts; neither uses the value of
$\Theta^\star$. The first, the only mixture-specific step, verifies that the
regularity constants $\beta_1,\beta_2$ entering Ref.~\cite{lerch2026}'s Theorem~2
stay $\mathrm{poly}(n)$-bounded for the $L$-branch mixture. The second invokes that
theorem to convert the curvature hypothesis $c$ into a variance lower bound.

\emph{Regularity transfer to the mixture.} Let $\pi_\ell$ be the simplex weights of the branch mixture
$z_A(\Theta)=\sum_\ell\pi_\ell z_A(\bm{\theta}^{(\ell)})$ (with $\pi_\ell=1/L$ in
the uniform case). For any mixture parameter $x=\theta_\nu^{(\ell)}$ and any
integer $q\ge1$,
\begin{equation}
\partial_x^q z_A(\Theta)=\pi_\ell\,\partial_{\theta_\nu^{(\ell)}}^q
z_A(\bm{\theta}^{(\ell)}),
\label{eq:bp-mixture-derivative}
\end{equation}
whenever the single-branch derivative is nonzero; mixed derivatives across distinct
branches vanish by the linearity of $z_A(\Theta)$ in the $\pi_\ell$, and same-branch
mixed derivatives again carry the factor $\pi_\ell$. Each branch obeys the
single-circuit bounds of Ref.~\cite{lerch2026}, $|\partial^q z_A|\le2^q$ (in
particular $|\partial^2 z_A|\le4$), and $|z_A(\Theta)|,|t_A|\le1$, so the product
rule applied to Eq.~\eqref{eq:bp-loss} reproduces the single-circuit one-parameter
derivative bound
\begin{equation}
\left|\partial_x^{2k}\mathcal{L}\right|
\le 4p_\sigma(1-p_\sigma)\,4^{2k},
\qquad k\ge1,
\label{eq:bp-single-derivative-bound}
\end{equation}
with the weight-sum prefactor $4p_\sigma(1-p_\sigma)$ of Ref.~\cite{lerch2026}
carried over unchanged.

The two regularity constants $\beta_1,\beta_2$ entering Ref.~\cite{lerch2026}'s
Theorem~2 are built from the mixed fourth derivatives, through
$\gamma_y=\sup\big|\partial_x^2\partial_y^2\mathcal{L}\big|$ and
$\beta_1=\sum_{y\neq x}\gamma_y/6$, and from a single-parameter
higher-derivative constant. The single-circuit value of $\gamma_y$ does not
transfer to the mixture: the same generator $G$ appears once in every branch,
and the two copies $x=\theta_G^{(\ell)}$ and $y=\theta_G^{(\ell')}$ couple
through every subset $A$ sensitive to $G$, a set of total weight
$\Theta(p_\sigma)$ rather than $\Theta(p_\sigma^2)$. At $\Theta=\bm{0}$ with
two uniform branches and a one-body generator $\{j\}$, an explicit evaluation
gives $|\partial_x^2\partial_y^2\mathcal{L}|=8\sum_{A\ni j}w_A$ independently
of the targets, so no bound proportional to $p_\sigma^2$ can hold. We
therefore bound $\gamma_y$ directly. Expanding
$\partial_x^2\partial_y^2(z_A-t_A)^2$ by the product rule gives sixteen terms
$D_S(z_A-t_A)\,D_{S^c}(z_A-t_A)$, one per subset $S$ of the four derivatives;
the fourteen terms with both factors differentiated are bounded by
$2^{|S|}\,2^{4-|S|}=16$ each, the two terms with an undifferentiated factor by
$|z_A-t_A|\,2^4\le32$ each, and $\sum_A w_A=1$, so
\begin{equation}
\gamma_y \le 288,
\qquad
\beta_1=\sum_{y\neq x}\frac{\gamma_y}{6}
\le 48(m-1),
\label{eq:bp-beta1-mixture}
\end{equation}
uniformly in the angles. The single-parameter constant involves only the
subsets sensitive to a single generator, so $\beta_2=a_0^2\gamma^6/6$ with
$a_0=4p_\sigma(1-p_\sigma)$ and $\gamma=4$ carries over from the
single-circuit theorem unchanged. The sum over $y\neq x$ runs over the $m-1$
remaining mixture parameters, the only place the parameter count $m=Lm_0$
enters. With $m=Lm_0$ polynomial whenever $L=O(\mathrm{poly}(n))$, the
constants $\beta_1,\beta_2$ are polynomially bounded in $n$.

\emph{From curvature to variance.} We may therefore invoke
Ref.~\cite{lerch2026}'s Theorem~2 with the hypothesized curvature
$c=\Omega(1/\mathrm{poly}(n))$. The accuracy parameter $\Delta\in(0,1)$ is free
and is selected by the user; for a fixed admissible $r$, decreasing $\Delta$
improves the prefactor $1-\Delta$ of the variance bound below, while it also
reduces the maximum admissible patch radius. The two displays that follow
are, respectively, the patch-admissibility condition and the variance lower
bound of that theorem; the factor $r^4/45$ in the latter is the leading
fourth-order contribution to the local variance bound and is imported
unchanged. Its admissible-patch condition,
\[
r^2 \le \Delta\,\frac{c^2}{2\beta_1 c+\beta_2},
\]
has a right-hand side bounded below by $1/\mathrm{poly}(n)$ for any fixed
$\Delta\in(0,1)$, because the numerator $c^2$ is inverse-polynomially large while
$\beta_1,\beta_2$ are polynomially bounded. Hence there exists an
inverse-polynomial patch half-width $r=1/\mathrm{poly}(n)$, which also meets the
theorem's separate $\gamma$ side condition $r\le3/(2\gamma)=3/8$ (with
$\gamma=4$), as $r\to0$ with $n$. On this patch,
\[
\mathrm{Var}_\Theta[\mathcal{L}]
\ge (1-\Delta)\,\frac{r^4}{45}\,c^2
= \Omega\!\left(\frac{1}{\mathrm{poly}(n)}\right).
\]
The loss variance is bounded away from zero by an inverse polynomial in
$n$, so no barren plateau in the sense of Sec.~\ref{sec:tr-curv} occurs at
$\Theta^\star$.
\end{proof}

\subsection{\label{app:bp-agnostic}Data-agnostic center: inverse-polynomial curvature}

At the unbiased, data-agnostic center the single-parameter curvature of the
mixture loss is $c=8\,w_{\{j\}}/L^2=\Omega(1/\mathrm{poly}(n))$, which by the
trainability lemma rules out a barren plateau. We evaluate the curvature
Eq.~\eqref{eq:bp-curv} at this center and then invoke
Lemma~\ref{lem:curv-bp}.

\emph{Evaluation at the unbiased center.} Set, in every branch, single-qubit
angles to $\pi/4$ and two-qubit angles to $0$. Every nontrivial branch
correlator then vanishes, since each factor is $\cos(\tfrac{\pi}{2})=0$,
\begin{equation}
z_A(\bm{\theta}^{(\ell)})=\prod_{j\in A}\cos\!\big(\tfrac{\pi}{2}\big)=0
\qquad\text{for every nontrivial }A\text{ and every }\ell,
\end{equation}
hence $z_A(\Theta)=0$. The entire data-mismatch
block of Eq.~\eqref{eq:bp-curv}, including all cross terms of
Eq.~\eqref{eq:bp-cross}, vanishes. In the model-sensitivity sum only
$A=\{j\}$ survives, because every higher-body branch sensitivity carries a
vanishing cosine factor,
\begin{equation}
g_A^{(j,\ell)}=-2\sin\!\big(\tfrac{\pi}{2}\big)\!\!
\prod_{k\in A\setminus\{j\}}\!\!\cos\!\big(\tfrac{\pi}{2}\big)=0
\quad\text{for }|A|\ge2,
\qquad
g_{\{j\}}^{(j,\ell)}=-2,
\end{equation}
the single-body value as in Ref.~\cite{lerch2026}. Therefore
\begin{equation}
\partial^2_{\theta_j^{(\ell)}}\mathcal{L}\Big|_{\text{unbiased}}
= 2\,w_{\{j\}}\,\pi_\ell^2(-2)^2
= \frac{8\,w_{\{j\}}}{L^2}
=: c,
\label{eq:bp-agnostic-curv}
\end{equation}
i.e.\ exactly the single-circuit curvature $8\,w_{\{j\}}$ of
Ref.~\cite{lerch2026} suppressed by $1/L^2$. For the low-body MMD the single-body
weight satisfies $w_{\{j\}}=\Theta(1/n)$, so for $L=O(\mathrm{poly}(n))$ this
curvature is $c=\Omega(1/\mathrm{poly}(n))$, which establishes
Eq.~\eqref{eq:rand-curv}.

\emph{Conclusion.} At the center $\Theta^\star$ with all branches at the
unbiased angles, every one-body parameter $\theta_j^{(\ell)}$ has the
curvature in Eq.~\eqref{eq:bp-agnostic-curv}, which meets the hypothesis
$c=\Omega(1/\mathrm{poly}(n))$
of Lemma~\ref{lem:curv-bp}. By Lemma~\ref{lem:curv-bp}, the mixture loss therefore
has an inverse-polynomial loss variance on a patch about $\Theta^\star$ and
exhibits no barren plateau, which proves Theorem~\ref{thm:rand-bp}. \hfill$\square$

\emph{Relation to sparse single circuits.} For a sparsely connected interaction
graph the single-circuit model can already avoid barren plateaus under broader
initializations~\cite{lerch2026}; the mixture guarantee is therefore most
informative in the densely connected regime.

\subsection{\label{app:bp-datadep}Data-dependent center: curvature}

This subsection computes the loss curvature at the data-dependent center
of Sec.~\ref{sec:tr-datadep} and feeds it to the Master
Lemma~\ref{lem:curv-bp} to obtain the data-dependent counterpart of
Theorem~\ref{thm:rand-bp}. The data-agnostic analysis of the preceding
subsection anchors the curvature at the unbiased center
$\theta_j^{(\ell)}=\pi/4$, $\theta_{jk}^{(\ell)}=0$, at which the
data-mismatch term drops out for every target because every nontrivial
branch correlator vanishes; the data-dependent center is
data-fitted, so the curvature is supplied by a different mechanism. We mirror the
single-circuit data-dependent guarantee of Ref.~\cite{lerch2026} (its Theorem~4
and the appendix proving it), lifting it to the data-dependent mixture through an exact partition
identity. Beyond the two regularity conditions already discharged by the
trainability lemma, the argument adds the mass-weighted partition structure
that supplies this identity and the two hypotheses,
Assumptions~\ref{as:factor} and~\ref{as:peak}; it holds for the full low-body
MMD with no body truncation.

\emph{Setup.} We use the weighted-mixture loss of Eq.~\eqref{eq:mmd} in its
correlator form, with the mixture correlator linear in the branches
[Eq.~\eqref{eq:moiqp-corr}],
\begin{equation}
\mathcal{L}(\Theta)=\sum_{A\subseteq[n]} w_A\big(z_A(\Theta)-t_A\big)^2,
\qquad
z_A(\Theta)=\sum_{\ell=0}^{L-1}\pi_\ell\,z_A(\bm{\theta}^{(\ell)}),
\label{eq:bpr-loss}
\end{equation}
with $w_A=(1-p_\sigma)^{\,n-|A|}p_\sigma^{\,|A|}$, $p_\sigma=\Theta(1/n)$, so that
$w_A=\Theta(n^{-|A|})$ and each single-body weight $w_{\{j\}}=\Theta(1/n)$.
Throughout this subsection the groups $\{C_\ell\}$ form a partition of the
$N$ training points and the branch weights are the empirical group masses,
$\pi_\ell=|C_\ell|/N$, as realized by the cluster-initialized scheme; the
global scheme is the degenerate case $C_\ell=\mathcal{D}$ for every $\ell$,
in which $t_A^{(\ell)}=t_A$ and any simplex weights qualify. Write the
per-group data moment
$t_A^{(\ell)}:=\langle Z_A\rangle_{C_\ell}$; its one-body case is
$t_j^{(\ell)}\equiv t_{\{j\}}^{(\ell)}=\langle Z_j\rangle_{C_\ell}$. Because each
$\pi_\ell=|C_\ell|/N$-weighted average over a partition is the full $N$-average,
these weights supply the exact arithmetic identity
\begin{equation}
\sum_\ell \pi_\ell\, t_A^{(\ell)}=t_A,
\label{eq:bpr-telescope}
\end{equation}
which holds for every subset $A$, not only the low-body ones; this all-orders
identity removes the body truncation of earlier versions. The identity is an
explicit hypothesis of the theorem, not a convention: uniform weights
$\pi_\ell=1/L$ satisfy it exactly when the partition is balanced,
$|C_\ell|=N/L$, as in the equal-weight control of
Appendix~\ref{app:branch-weights}, whereas an unbalanced partition with
uniform weights violates it, and the theorem makes no claim there.

\emph{The analyzed center.} As in Ref.~\cite{lerch2026}'s Theorem~4, we anchor
the curvature at the data-dependent center obtained by per-group moment matching
in its factorizing form: each branch $\ell$ is seeded so that
\begin{equation}
\cos\!\big(2\theta_j^{(\ell)}\big)=t_j^{(\ell)}\quad\text{for all }j,
\qquad
\theta_{jk}^{(\ell)}=0\quad\text{for all edges }(j,k)\in E.
\label{eq:bpr-center}
\end{equation}
This is the per-branch lift of conditions~1--2 of Ref.~\cite{lerch2026}'s
Theorem~4: one-body angles matched to the group moments, all two-body angles
set to zero. With every two-body angle equal to zero, the branch correlator
factorizes exactly into its one-body marginals,
\begin{equation}
z_A(\bm{\theta}^{(\ell)})
=\prod_{j\in A}\cos\!\big(2\theta_j^{(\ell)}\big)
=\prod_{j\in A} t_j^{(\ell)},
\label{eq:bpr-factor}
\end{equation}
so that the incident edge-angle product
$P_j^{(\ell)}=\prod_{k:(j,k)\in E}\cos^2(2\theta_{jk}^{(\ell)})=1$
identically. The single scalar $1-(t_j^{(\ell)})^2$ thus takes over the role
played by the witness-curvature combination $P_j-t_j^2$ in the earlier
formulation, and no separate condition on $P_j$ is needed.

\emph{Relation to the implemented initialization.} The schemes of
Sec.~\ref{sec:init-dependent} start at this factorizing center exactly (each
branch's one-body angles set from its group moments, every two-body angle
zero), and the only departure is the coincidence-breaking perturbation of
Sec.~\ref{sec:train}: one independent uniform draw of half-width
$r_0=(\pi/2)/\sqrt{m}$ added to each of the $m$ trainable compiled angles.
The perturbation is smooth: if the witness branch satisfies
$\sum_{k:(j^\star,k)\in E}(\theta_{j^\star k}^{(\ell^\star)})^2=o(1)$, then
$P_{j^\star}^{(\ell^\star)}
=\prod_{k:(j^\star,k)\in E}\cos^2(2\theta_{j^\star k}^{(\ell^\star)})=1-o(1)$ and
the one-body witness factor $1-(t_{j^\star}^{(\ell^\star)})^2$ is replaced locally by
$P_{j^\star}^{(\ell^\star)}-(t_{j^\star}^{(\ell^\star)})^2$. The implemented draw
satisfies this smallness condition: the induced per-branch two-body
perturbation has $\sum_{k}(\theta_{j^\star k}^{(\ell^\star)})^2=O(1/n)$ in
expectation for the fully connected layout, where the gate count is
$m=\Theta(n^2 2^a)$. The implemented initialization therefore remains within a
vanishing perturbation of the analyzed center, and the theorem itself is
stated and proved at the exact factorizing center.

\emph{Assumptions.} The Master Lemma consumes two regularity conditions: (R1)
the polynomial-width requirement $L=O(\mathrm{poly}(n))$, and (R2) the per-branch
bounded-derivative bounds $|\partial^q z_A|\le2^q$. Beyond these and the
mass-weighted partition structure fixed in the setup [the partition identity
Eq.~\eqref{eq:bpr-telescope}], the data-dependent guarantee uses the
following properties of the grouped data.
\begin{itemize}
\item \emph{Approximately factorizable groups}
(Assumption~\ref{as:factor}): every assigned group's correlators factorize over
its one-body marginals up to
$\big|t_A^{(\ell)}-\prod_{j\in A}t_j^{(\ell)}\big|
\leq(C/n)^{|A|/2}$ for an $n$-independent constant $C$. This bounds the
residual beyond the corresponding product moment, not the magnitude of
$t_A^{(\ell)}$ itself. For sufficiently large $n$, $C/n<1$, so the residual
is geometrically suppressed with the body order.
\item \emph{A noncollapsed, sufficiently weighted group}
(Assumption~\ref{as:peak}): at least one
\emph{well-weighted} group $\ell^\star$, with weight $\pi_{\ell^\star}=\omega(1/n)$,
keeps a non-saturated marginal at some qubit $j^\star$,
$1-(t_{j^\star}^{(\ell^\star)})^2=\Theta(1)$. A group concentrated on a single bit
string saturates every marginal at $\pm1$ and would void the sensitivity term;
the weight floor is what lifts the single-circuit curvature bound of
Ref.~\cite{lerch2026} to the weighted mixture.
\end{itemize}
The first hypothesis is imposed within every assigned group, whereas the
second requires only one sufficiently weighted group with a non-saturated
coordinate. Clustering can make the first condition easier to satisfy because
within-cluster correlators can be closer to their corresponding product
moments.

\begin{theoremdatadepformal}[Trainability at the data-dependent center, formal]
Consider the weighted mixed IQP-QCBM of Eq.~\eqref{eq:ancilla-weights} with
$L=O(\mathrm{poly}(n))$ branches and the low-body MMD loss
($\sigma=\Theta(\sqrt{n})$), with groups partitioning the training set and
branch weights equal to the group masses, $\pi_\ell=|C_\ell|/N$, so that the
partition identity Eq.~\eqref{eq:bpr-telescope} holds (the global scheme
$C_\ell=\mathcal{D}$ is the degenerate case), initialized at the
data-dependent center Eq.~\eqref{eq:bpr-center}. Under
Assumptions~\ref{as:factor} and~\ref{as:peak}, the loss curvature at the
witness coordinate
$(j^\star,\ell^\star)$ satisfies
\begin{equation}
\partial^2_{\theta_{j^\star}^{(\ell^\star)}}\mathcal{L}(\Theta^\star)
\;\geq\;8w_{\{j^\star\}}\pi_{\ell^\star}^2
\big[1-(t_{j^\star}^{(\ell^\star)})^2\big]-O(1/n^3)
\;=\Omega(1/\mathrm{poly}(n)).
\label{eq:bpr-curv}
\end{equation}
Consequently, Lemma~\ref{lem:curv-bp} gives an inverse-polynomial patch around
the data-dependent center on which
$\mathrm{Var}_\Theta[\mathcal{L}]=\Omega(1/\mathrm{poly}(n))$.
\end{theoremdatadepformal}

\begin{proof}
Steps~1--2 isolate the positive one-body curvature and show that its mismatch
vanishes exactly. Steps~3--5 control the higher-body mismatch, and Step~6
shows that the model-sensitivity term dominates.

\emph{Step~1 (single-parameter curvature, general $\pi_\ell$).} The parameter
$\theta_{j^\star}^{(\ell^\star)}$ enters $\mathcal{L}$ only through branch
$\ell^\star$, with weight $\pi_{\ell^\star}$. As in Appendix~\ref{app:correlators}
and Eq.~\eqref{eq:estA-derivbound}, $\theta_{j^\star}$ enters $z_A$ only through
the factor $\cos2\theta_{j^\star}$, present iff $j^\star\in A$, so
$\partial^2_{\theta_{j^\star}}z_A=-4z_A$ for $j^\star\in A$ and $0$ otherwise. Differentiating
Eq.~\eqref{eq:bpr-loss} twice,
\begin{equation}
\partial^2_{\theta_{j^\star}^{(\ell^\star)}}\mathcal{L}
=\sum_{A\ni j^\star}2w_A\,\pi_{\ell^\star}^2
\big(g_A^{(j^\star,\ell^\star)}\big)^2
\;+\;8\,\pi_{\ell^\star}\!\!\sum_{A\ni j^\star}\!w_A\,
z_A(\bm{\theta}^{(\ell^\star)})\big(t_A-z_A(\Theta)\big),
\label{eq:bpr-d2}
\end{equation}
which is Eq.~\eqref{eq:bp-curv} evaluated at the witness branch, where
$g_A^{(j^\star,\ell^\star)}:=\partial_{\theta_{j^\star}^{(\ell^\star)}}
z_A(\bm{\theta}^{(\ell^\star)})$ is the branch sensitivity. The first block
collects nonnegative \emph{sensitivity} terms; the second is the \emph{mismatch}
block. At the factorizing center Eq.~\eqref{eq:bpr-center},
Eq.~\eqref{eq:bpr-factor} gives
\begin{equation}
g_A^{(j^\star,\ell^\star)}
=-2\sin(2\theta_{j^\star}^{(\ell^\star)})\!\!
\prod_{k\in A\setminus\{j^\star\}}\!\! t_k^{(\ell^\star)},
\qquad
\big(g_A^{(j^\star,\ell^\star)}\big)^2
=4\big(1-(t_{j^\star}^{(\ell^\star)})^2\big)\!\!
\prod_{k\in A\setminus\{j^\star\}}\!\!(t_k^{(\ell^\star)})^2,
\label{eq:bpr-gid}
\end{equation}
using $\sin^2(2\theta_{j^\star}^{(\ell^\star)})
=1-\cos^2(2\theta_{j^\star}^{(\ell^\star)})
=1-(t_{j^\star}^{(\ell^\star)})^2$. Each per-subset sensitivity summand is therefore
$8w_A\pi_{\ell^\star}^2(1-(t_{j^\star}^{(\ell^\star)})^2)
\prod_{k\in A\setminus\{j^\star\}}(t_k^{(\ell^\star)})^2\ge0$, which is exactly
the model-sensitivity term of Ref.~\cite{lerch2026} carrying the extra
mixture factor $\pi_{\ell^\star}^2$.

\emph{Step~2 (one-body mismatch vanishes exactly).} For $A=\{j^\star\}$ the
factorization Eq.~\eqref{eq:bpr-factor} gives
$z_{\{j^\star\}}(\bm{\theta}^{(\ell)})=t_{j^\star}^{(\ell)}$. Evaluating the one-body
model correlator at the center and applying the partition identity
Eq.~\eqref{eq:bpr-telescope} at $|A|=1$ then collapses it onto the target,
\begin{equation}
z_{\{j^\star\}}(\Theta^\star)=\sum_\ell\pi_\ell t_{j^\star}^{(\ell)}
=t_{\{j^\star\}},
\label{eq:bpr-onebody-target}
\end{equation}
where the second equality is the partition identity. The one-body mismatch
$t_{\{j^\star\}}-z_{\{j^\star\}}(\Theta^\star)$ is thus identically zero, cleaner
than in the single circuit, where it vanishes by the construction
$z_j=t_j$. Retaining only the $A=\{j^\star\}$ sensitivity term, the first-order
contribution to Eq.~\eqref{eq:bpr-d2} is
\begin{equation}
2w_{\{j^\star\}}\pi_{\ell^\star}^2
\big(g_{\{j^\star\}}^{(j^\star,\ell^\star)}\big)^2
=8\,w_{\{j^\star\}}\,\pi_{\ell^\star}^2
\big(1-(t_{j^\star}^{(\ell^\star)})^2\big)
\label{eq:bpr-first}
\end{equation}
and this becomes $\Theta(\pi_{\ell^\star}^2/n)$ under
Assumption~\ref{as:peak}.

\emph{Step~3 (high-body mismatch is controlled, not truncated).} Using the
factorization Eq.~\eqref{eq:bpr-factor} for the model correlator and the partition
identity Eq.~\eqref{eq:bpr-telescope} for the target, the mismatch at any subset
$A$ obeys
\begin{equation}
t_A-z_A(\Theta^\star)
=\sum_\ell\pi_\ell\Big(t_A^{(\ell)}-\prod_{j\in A}t_j^{(\ell)}\Big),
\qquad
\big|t_A-z_A(\Theta^\star)\big|
\le\sum_\ell\pi_\ell\left(\frac{C}{n}\right)^{|A|/2}
=\left(\frac{C}{n}\right)^{|A|/2},
\label{eq:bpr-mismatch}
\end{equation}
where the inequality is Assumption~\ref{as:factor} applied within each group and
$\sum_\ell\pi_\ell=1$. This bound holds for \emph{every} $|A|$, so no body
truncation is needed: the per-group factorizability controls the high-body
mismatch directly.

\emph{Step~4 (per-subset combination).} Fix $A\ni j^\star$ with $|A|=K\ge2$. By the
center factorization of Eq.~\eqref{eq:bpr-factor} the witness branch's correlator and
sensitivity share a single factor, the product
$x:=\prod_{k\in A\setminus\{j^\star\}}|t_k^{(\ell^\star)}|\ge0$ of the remaining
one-body marginals:
$|z_A(\bm{\theta}^{(\ell^\star)})|=|t_{j^\star}^{(\ell^\star)}|\,x$
and $g_A^{(j^\star,\ell^\star)}\propto x$ [Eq.~\eqref{eq:bpr-gid}]. The
per-subset term of Eq.~\eqref{eq:bpr-d2} is then a quadratic in $x$,
\begin{equation}
T_A=2w_A\pi_{\ell^\star}^2\big(g_A^{(j^\star,\ell^\star)}\big)^2
+8\pi_{\ell^\star}w_A\,z_A(\bm{\theta}^{(\ell^\star)})\big(t_A-z_A(\Theta^\star)\big).
\end{equation}
By Eq.~\eqref{eq:bpr-gid} the sensitivity block collapses to a single quadratic
in $x$,
\begin{equation}
2w_A\pi_{\ell^\star}^2\big(g_A^{(j^\star,\ell^\star)}\big)^2
=2w_A\pi_{\ell^\star}^2\cdot4
\big(1-(t_{j^\star}^{(\ell^\star)})^2\big)x^2
=8w_A\,a\,x^2,
\end{equation}
with the abbreviations
\begin{equation}
a:=\pi_{\ell^\star}^2\big(1-(t_{j^\star}^{(\ell^\star)})^2\big),
\qquad
b:=\pi_{\ell^\star}\,|t_{j^\star}^{(\ell^\star)}|
\left(\frac{C}{n}\right)^{K/2}.
\end{equation}
The mismatch block is lower-bounded using
$|z_A(\bm{\theta}^{(\ell^\star)})|=|t_{j^\star}^{(\ell^\star)}|\,x$ from
Eq.~\eqref{eq:bpr-factor} together with the bound Eq.~\eqref{eq:bpr-mismatch},
which gives a term linear in $x$,
\begin{equation}
8\pi_{\ell^\star}w_A\,z_A(\bm{\theta}^{(\ell^\star)})\big(t_A-z_A(\Theta^\star)\big)
\ge-8\pi_{\ell^\star}w_A\,|t_{j^\star}^{(\ell^\star)}|\,x
\left(\frac{C}{n}\right)^{K/2}
=-8w_A\,b\,x.
\end{equation}
Hence $T_A\ge8w_A\,x(a x-b)$. Minimizing $x(ax-b)$ over $x\ge0$ gives the value
$-b^2/(4a)$, so the combined summand is bounded below by
\begin{equation}
T_A\;\ge\;8w_A\Big(\!-\frac{b^2}{4a}\Big)
=-\,\frac{2w_A\,b^2}{a}
=-\,2w_A\,\frac{(t_{j^\star}^{(\ell^\star)})^2}
{1-(t_{j^\star}^{(\ell^\star)})^2}
\left(\frac{C}{n}\right)^K.
\label{eq:bpr-permin}
\end{equation}
The factor $\pi_{\ell^\star}^2$ in $b^2$ cancels the same factor in $a$, so
Eq.~\eqref{eq:bpr-permin} is \emph{mixture-independent}, reproducing
Ref.~\cite{lerch2026}'s per-subset bound verbatim.

\emph{Step~5 (summing the negative parts).} Summing Eq.~\eqref{eq:bpr-permin} over
$A\ni j^\star$ with $|A|\ge2$ collects the prefactor
$2(t_{j^\star}^{(\ell^\star)})^2/
[1-(t_{j^\star}^{(\ell^\star)})^2]$ times
\begin{equation}
\mathcal{S}:=\sum_{A\ni j^\star,\,|A|\ge2}
w_A\left(\frac{C}{n}\right)^{|A|}.
\end{equation}
With $w_A=(1-p_\sigma)^{n-|A|}p_\sigma^{|A|}$, the count $\binom{n-1}{K-1}$ of
$j^\star$-incident $K$-subsets and $p_\sigma=\Theta(1/n)$ give
\begin{equation}
\mathcal S\leq(n-1)\left(\frac{p_\sigma C}{n}\right)^2
=O(1/n^3).
\label{eq:bpr-sum-proof}
\end{equation}
The sum is independent of the mixture weights.

\emph{Step~6 (sensitivity dominates).} Combining the exact first-order term
Eq.~\eqref{eq:bpr-first} with the high-body lower bound from Steps~4--5,
\begin{equation}
\partial^2_{\theta_{j^\star}^{(\ell^\star)}}\mathcal{L}(\Theta^\star)
\;\ge\;8\,w_{\{j^\star\}}\,\pi_{\ell^\star}^2
\big(1-(t_{j^\star}^{(\ell^\star)})^2\big)
-\frac{2(t_{j^\star}^{(\ell^\star)})^2}
{1-(t_{j^\star}^{(\ell^\star)})^2}\,\mathcal{S}
\label{eq:bpr-combine}
\end{equation}
Under Assumption~\ref{as:peak}, the prefactor of $\mathcal S$ is $O(1)$ and
the positive term is $\Theta(\pi_{\ell^\star}^2/n)$. Since
$\pi_{\ell^\star}=\omega(1/n)$, this term dominates the $O(1/n^3)$ remainder,
and therefore
$\partial^2_{\theta_{j^\star}^{(\ell^\star)}}\mathcal{L}(\Theta^\star)
=\Omega(\pi_{\ell^\star}^2/n)$.

\emph{Conclusion.} The witness curvature
$\Omega(1/\mathrm{poly}(n))$ is exactly the hypothesis of
Lemma~\ref{lem:curv-bp}. The lemma gives an
inverse-polynomial patch on which
$\mathrm{Var}_\Theta[\mathcal{L}]=\Omega(1/\mathrm{poly}(n))$ for the full
low-body MMD, proving Theorem~\ref{thm:datadep-bp-main}$'$.
\end{proof}

\emph{Coordinate of the patch and transfer to the compiled circuit.} The
curvature lower bound and the variance estimate, like those of the data-agnostic
center, are stated in the per-branch (decomposed) angles $\{\theta_G^{(\ell)}\}$,
to which the witness coordinate $x=(j^\star,\ell^\star)$, the admissible patch of
half-width $r$, and the regularity constants $\beta_1,\beta_2$ all refer. The
per-branch guarantee transfers verbatim to the deployed mixed IQP because the
Walsh--Hadamard map of Eq.~\eqref{eq:estB-angles} into the mixed IQP angles
$\{\widetilde{\theta}_{G,S}\}$ is a fixed, deterministic reparametrization and the
loss is the same function of the model in both coordinate systems (the correlator
equivalence of Appendix~\ref{app:ciqp}), so no transform of the bound is required.
Separately, Lemma~\ref{lem:coord-transfer} below makes the transfer explicit:
it identifies the image of the per-branch patch in the compiled angles, on
which the identical variance bound holds, and shows that the witness curvature
is not diluted by the change of coordinates.

\begin{lemma}[Coordinate transfer of the trainability guarantee]
\label{lem:coord-transfer}
Fix a gate $G$ and stack its $L=2^a$ branch angles into the vector
$\bm{\theta}_G=(\theta_G^{(\ell)})_{\ell}$ and its $L$ compiled coefficients into
$\widetilde{\bm{\theta}}_G=(\widetilde{\theta}_{G,S})_{S}$. By
Eqs.~\eqref{eq:coord-map}--\eqref{eq:coord-chain} of
Appendix~\ref{app:train-embed} the two are related by the well-conditioned Walsh
bijection $\bm{\theta}_G=H\widetilde{\bm{\theta}}_G$ ($H=\sqrt{L}\,Q$ with $Q$
orthogonal, $\mathrm{cond}(H)=1$), under which the Hessian transforms by the
congruence $\nabla^2_{\widetilde{\bm{\theta}}}\mathcal{L}
=H^{\mathsf T}(\nabla^2_{\bm{\theta}}\mathcal{L})H$. Then:
(i)~the variance guarantee of Lemma~\ref{lem:curv-bp}, stated for the
coordinate-wise uniform distribution on the per-branch hypercube patch of
half-width $r$, holds with the identical lower bound for the push-forward of
that distribution under the bijection, a compiled-angle patch in which every
coordinate satisfies $|\delta\widetilde{\theta}_{G,S}|\le r$; and
(ii)~any per-branch witness curvature $c_\star=\Omega(1/\mathrm{poly}(n))$
maps to a unit direction of the compiled block with curvature at least
$Lc_\star$ (no positivity hypothesis on the other branches), so the certifying
curvature is not diluted by the compilation.
\end{lemma}

\begin{proof}
By Appendix~\ref{app:train-embed} the map is block-diagonal over gates, so it
suffices to treat one block.

\emph{(i) Push-forward of the patch.} Let $\Theta=\Theta^\star+\bm{\xi}$ with
$\bm{\xi}\sim\mathrm{Unif}([-r,r]^m)$ be the perturbation for which
Lemma~\ref{lem:curv-bp} certifies
$\mathrm{Var}[\mathcal{L}]\ge(1-\Delta)r^4c_\star^2/45$, and let the compiled
parameter of the same model be
$\widetilde{\Theta}=\widetilde{\Theta}^\star+H^{-1}\bm{\xi}$ blockwise. The
loss is the same function of the model in both parametrizations (the
correlator equivalence of Appendix~\ref{app:ciqp}), so the two descriptions
define the same random variable and its variance is unchanged. The compiled
patch is the image of the hypercube, a rotated and $1/\sqrt{L}$-scaled
parallelepiped rather than a product of intervals over the compiled axes; each
of its coordinates obeys
$|\delta\widetilde{\theta}_{G,S}|
=L^{-1}\big|\sum_{\ell}(-1)^{S\cdot\ell}\xi_{G,\ell}\big|\le r$, so the patch
lies inside the compiled hypercube of the same half-width.

\emph{(ii) Witness direction.} The data-agnostic theorem or data-dependent
corollary supplies one per-branch axis with
$(\nabla^2_{\bm{\theta}}\mathcal{L})_{\ell^\star,\ell^\star}\ge c_\star$
[Eqs.~\eqref{eq:rand-curv},\,\eqref{eq:bpr-curv}], with no constraint on the other
entries. Its image $\widehat{\bm{u}}=\sqrt{L}\,H^{-1}\bm{e}_{\ell^\star}$ is a unit
direction in the compiled block ($\|H^{-1}\bm{e}_{\ell^\star}\|=L^{-1/2}$ since
$\mathrm{cond}(H)=1$), and by the congruence $H\widehat{\bm{u}}=\sqrt{L}\,
\bm{e}_{\ell^\star}$ gives
\begin{equation}
\widehat{\bm{u}}^{\mathsf T}(\nabla^2_{\widetilde{\bm{\theta}}}\mathcal{L})
\widehat{\bm{u}}=L\,(\nabla^2_{\bm{\theta}}\mathcal{L})_{\ell^\star,\ell^\star}
\ge L\,c_\star\ge c_\star .
\end{equation}
The perfect conditioning ($Q=H/\sqrt{L}$ orthogonal) is what turns the witness
axis into a genuine unit direction; the factor $L\ge1$ only helps.

Part~(ii) is a statement about the Hessian congruence alone. We do not claim
that Lemma~\ref{lem:curv-bp} re-applies with the compiled axes in place of the
per-branch axes: its hypothesis is a coordinate second derivative and its
patch an axis-aligned hypercube, both tied to the coordinate frame, and an
independent uniform perturbation of the compiled coordinates is a different
distribution from the push-forward in~(i). The guarantee asserted in the
compiled angles is the push-forward one.
\end{proof}

Lemma~\ref{lem:coord-transfer} applies verbatim to Theorem~\ref{thm:rand-bp} of
the data-agnostic center: that guarantee is also proved per-branch and deployed
through the same mixed IQP, so perturbing in the per-branch angles and mapping
the perturbation through the Walsh--Hadamard bijection carries the same
asymptotic guarantee to the compiled circuit.

\emph{Finite-size interpretation.} Assumption~\ref{as:factor} concerns a
family of datasets controlled by one $n$-independent constant $C$. To
characterize a single dataset, write the non-factorizing residual of group
$\ell$ on subset $A$ as
$r_A^{(\ell)}:=t_A^{(\ell)}-\prod_{j\in A}t_j^{(\ell)}$. At a fixed $n$ the
smallest constant for which Assumption~\ref{as:factor} holds is
\begin{equation}
C_n:=n\,q_n,\qquad
q_n:=\max_{\ell,\,|A|\geq2}\big|r_A^{(\ell)}\big|^{2/|A|},
\label{eq:factor-constant}
\end{equation}
since $|r_A^{(\ell)}|\leq q_n^{|A|/2}=(C_n/n)^{|A|/2}$ then holds for every
group and subset. This $C_n$ is only the effective constant required at that
size; one system size cannot establish that it remains bounded as $n$ grows.
Appendix~\ref{app:assumption-check} therefore separates an exact finite-size
calculation on the two $n=16$ datasets from descriptive low-body diagnostics
on the larger benchmarks.

\FloatBarrier
\subsection{\label{app:assumption-check}Finite-size diagnostics of grouped data}

We first measure a typical residual within each group,
\begin{equation}
\epsilon_k^{(\ell)}
=\operatorname*{median}_{|A|=k}
\big|r_A^{(\ell)}\big|,
\label{eq:eps-k}
\end{equation}
and aggregate it as
$\epsilon_k=\sum_\ell\pi_\ell\epsilon_k^{(\ell)}$. Its body-order-rescaled
finite-size constant is
\begin{equation}
C_{k,\mathrm{eff}}:=n\epsilon_k^{2/k},
\qquad q_{k,\mathrm{eff}}:=C_{k,\mathrm{eff}}/n.
\label{eq:ceff-k}
\end{equation}
We compute the same quantities from the $95$th percentile to expose the tail.
For $k=2,\ldots,6$, subsets are exhaustive at $n=16$ and sampled uniformly
($4\times10^4$ per order) on the larger datasets. Each group is paired with an
independent-bit surrogate having the same size and one-body marginals, which
sets the finite-sample reference in Fig.~\ref{fig:app-assumption}.

For the two $n=16$ datasets we additionally enumerate every subset at every
body order $k=2,\ldots,16$ and take the maximum over both subsets and groups,
as required by Eq.~\eqref{eq:factor-constant}. Across all reported groupings,
$q_n$ ranges from $0.743$ to $0.897$ on the blobs and from $0.754$ to $0.809$
on Ising. Thus the effective residual envelope decreases geometrically with
body order at this size. The corresponding $C_n$ values, however, are
$11.9$--$14.3$ and $12.1$--$12.9$; these numbers do not establish the
$n$-independent constant required by Assumption~\ref{as:factor}. The exact
center curvatures are therefore evaluated directly in
Sec.~\ref{sec:exp-theory}, rather than inferred from unmeasured asymptotic
scaling.

For MNIST and D-Wave, exhaustive maximization over all subsets is infeasible.
Figure~\ref{fig:app-assumption} instead shows typical low-body behavior: at the
largest mixtures, $\max_{k\leq6}q_{k,\mathrm{eff}}$ is $0.39$ and $0.40$ for
the medians, and $0.62$ and $0.71$ for the $95$th percentiles. Clustering most
strongly reduces the two-body residual, while several higher-body curves
approach their finite-sample references. These measurements characterize the
benchmarks at their fixed sizes; they neither verify the uniform maximum in
Assumption~\ref{as:factor} nor establish an $n$-independent constant for a
growing data family.

The witness condition is separately nondegenerate in every reported
configuration. In the most massive group,
$1-(t_{j^\star}^{(\ell)})^2$ reaches $1.00$ on MNIST and D-Wave, $0.28$ on Ising,
and $0.23$ on the blobs at $a=3$. These are finite-size witness values, not
evidence for their asymptotic scaling.

\begin{figure*}[!ht]
\includegraphics[width=\textwidth]{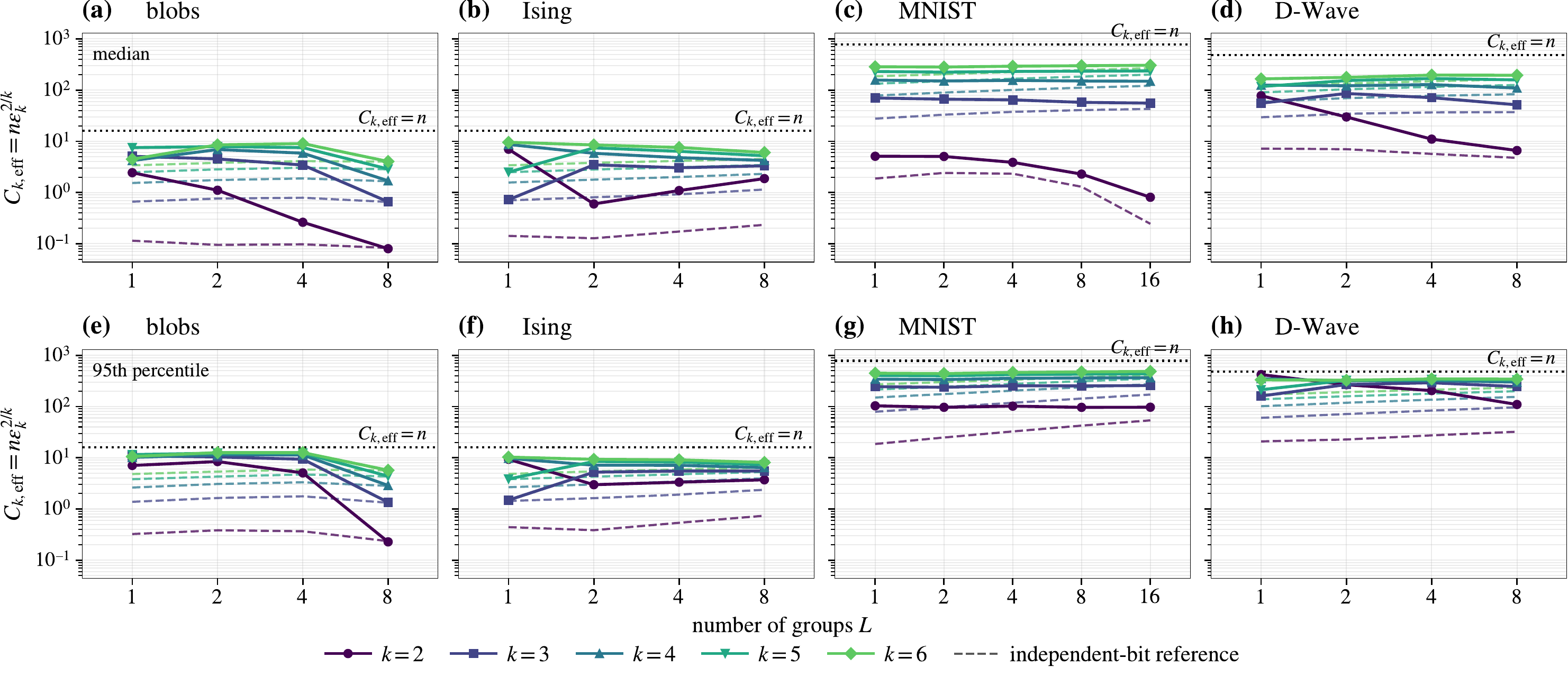}
\caption{\label{fig:app-assumption} Finite-size effective constants
$C_{k,\mathrm{eff}}=n\epsilon_k^{2/k}$ for the grouped benchmark data.
Columns are datasets and $L=1$ denotes the ungrouped global initialization;
the top and bottom rows use the group-weighted median and $95$th-percentile
residuals, respectively. Solid curves show the data and dashed curves the
independent-bit finite-sample references. The horizontal line
$C_{k,\mathrm{eff}}=n$ marks $q_{k,\mathrm{eff}}=1$. These typical low-body
summaries do not test whether one $n$-independent $C$ controls a growing data
family.}
\end{figure*}

\section{\label{app:expressivity}Expressivity of the mixture}

The mixture increases expressivity along two complementary axes, and this part
establishes both. First, the ancilla register lifts the ancilla-free
representational limit: a trivial mixture already represents every distribution,
so the mixed IQP family is universal (Sec.~\ref{app:universal}). Second,
within that universal family it is \emph{inter-branch diversity} (how much the
branches differ) that lets the mixture surpass the best ancilla-free circuit, and we
quantify the diversity a given improvement requires (Sec.~\ref{app:diversity}).
Trainability, established in Appendix~\ref{app:bp}, is what makes this
expressivity reachable by optimization; here we characterize it.

\subsection{\label{app:universal}Universality via a trivial mixture}

Recall (Sec.~\ref{sec:rand-iqp}) that an ancilla-free IQP circuit is not
universal~\cite{kurkin2025universality}; we give a constructive version through
the mixture. A deliberately trivial mixture, one branch per support point of the
target, already represents \emph{any} distribution, so the mixed IQP family is
universal. The construction has no generative value on its own; it serves to
locate the role of the ancillas.

The building block is a \emph{one-local delta generator}. For a bit string
$\bm{x}\in\{0,1\}^n$, take the ancilla-free IQP circuit of Eq.~\eqref{eq:iqp} with
one-body angles $\theta_j=\tfrac{\pi}{2}x_j$ and all two-body angles set to zero.
Each qubit then evolves independently as
$\exp\!\big(i\tfrac{\pi}{2}x_j X_j\big)\ket{0}$, which is $\ket{0}$ when $x_j=0$
and $i\ket{1}$ when $x_j=1$, so a computational-basis measurement returns $\bm{x}$
with certainty. The circuit therefore realizes the point mass
$\delta_{\bm{x}}$, using only the simplest (one-local) IQP gates, which are also
classically trivial to sample.

\begin{proof}[Proof of Proposition~\ref{prop:universal}]
Enumerate the support $\operatorname{supp}(p)=\{\bm{x}^{(1)},\dots,\bm{x}^{(L)}\}$
with masses $p_\ell=p(\bm{x}^{(\ell)})$. Form the mixed IQP mixture of
Sec.~\ref{sec:rand-iqp} in which branch $\ell$ is the one-local delta generator
for $\bm{x}^{(\ell)}$, and set the mixture weights to $\pi_\ell=p_\ell$ through the
ancilla state of Eq.~\eqref{eq:ancilla-weights} (when $L$ is not a power of two,
the surplus $2^{\lceil\log_2 L\rceil}-L$ branches carry zero weight). By
construction the mixture distribution is
\begin{equation}
\sum_{\ell=1}^{L}\pi_\ell\,\delta_{\bm{x}^{(\ell)}}=\sum_{\ell=1}^{L}p_\ell\,
 \delta_{\bm{x}^{(\ell)}}=p\,.
\end{equation}
The same distribution is obtained either by drawing branch $\ell$ with
probability $\pi_\ell$ or by supplying the branch-controlled circuit with the
general ancilla amplitude state of Eq.~\eqref{eq:ancilla-weights}. This proves
universality of the weighted mixed IQP family. For uniform weights the latter
is an ordinary compiled IQP circuit; general weights require the additional
state preparation and are not claimed to belong to the commuting IQP gate set.
\end{proof}

Two remarks fix the meaning of the result. First, the universality it establishes
is information-theoretic, not efficient: a generic distribution has
$L=|\operatorname{supp}(p)|$ up to $2^n$, requiring as many as $a=n$ ancillas, and
the construction merely memorizes the support: one delta branch per support
point, with no efficiency. A stronger universality statement for IQP
circuits with hidden units is given in Ref.~\cite{kurkin2025note}; the
proposition is a representability construction for the weighted branch
family, not a new result for ordinary IQP circuits. Second, the substantive question
taken up in the body of the paper is how \emph{few} ancillas, with structured
(non-delta) branches and the cluster-initialized scheme of Sec.~\ref{sec:init},
suffice to capture real data. In the language of Proposition~\ref{prop:diversity},
the delta mixture is the extreme of branch diversity, every branch a distinct
point mass, whereas a useful model attains the same coverage with far fewer,
broader branches.

\subsection{\label{app:diversity}Branch diversity is necessary to surpass the best ancilla-free circuit}

The data-agnostic center of Sec.~\ref{sec:init-agnostic} is branch-coincident,
and a mixture whose branches coincide reduces to a single ancilla-free
circuit. Here we make
this quantitative: the MMD loss attainable by the mixture is
controlled from below by how much its branches differ, so that branch diversity is
\emph{necessary}, though not sufficient, to do better than the best
ancilla-free IQP circuit.

We work in the expectation-value representation of the loss, Eq.~\eqref{eq:mmd}. Write the
expectation-value vector of a distribution $p$ as $z(p)=(z_A(p))_{A\subseteq[n]}$ with
$z_A(p)=\mathbb{E}_{\bm{x}\sim p}[(-1)^{\sum_{j\in A}x_j}]$, and equip the
expectation-value space with the kernel-induced weighted norm
\begin{equation}
\|u\|_w^2=\sum_{A\subseteq[n]}w_A\,u_A^2,
\label{eq:wnorm}
\end{equation}
with the $w_A$ of Eq.~\eqref{eq:mmd}. In this norm the loss is a squared
distance,
\begin{equation}
\mathcal{L}(\Theta)=\|z(p_\Theta)-t\|_w^2,
\end{equation}
where $t=(t_A)$ is the data expectation value vector of Eq.~\eqref{eq:tA}, and the
squared MMD between any two distributions $p,q$ is the same weighted distance
between their expectation-value vectors,
\begin{equation}
\mathrm{MMD}^2(p,q)=\|z(p)-z(q)\|_w^2.
\end{equation}
Let
\begin{equation}
\mathcal{S}=\{\,z(\bm{\theta})\;:\;\bm{\theta}\in\mathbb{R}^{m_0}\,\}
\label{eq:single-set}
\end{equation}
be the set of expectation-value vectors reachable by a single \emph{ancilla-free} IQP circuit on the
fixed interaction graph; since each $z_A(\bm{\theta})$ is a continuous function of
the angles and the angles range over a torus, $\mathcal{S}$ is compact. The
ancilla-free approximation floor is the closest such vector to the data,
\begin{equation}
\Delta_{\text{anc-free}}=\min_{s\in\mathcal{S}}\|t-s\|_w
=\min_{\bm{\theta}}\sqrt{\mathcal{L}(\bm{\theta})}\,,
\label{eq:floor}
\end{equation}
and the branch diversity of a mixture $p_\Theta$ with weights $\pi_\ell$ is the
weighted dispersion of its branch expectation values about their mean,
\begin{equation}
D_{\mathrm{branch}}^2
=\sum_{\ell}\pi_\ell\,\big\|z(\bm{\theta}^{(\ell)})-z(p_\Theta)\big\|_w^2
=\tfrac12\sum_{\ell,\ell'}\pi_\ell\pi_{\ell'}\,
 \mathrm{MMD}^2\!\big(q_\ell,q_{\ell'}\big),
\label{eq:dbranch}
\end{equation}
where $q_\ell$ is branch $\ell$'s Born distribution. The second equality is the
standard variance identity; it expresses $D_{\mathrm{branch}}^2$ as one half of
the $\pi_\ell$-weighted double sum of pairwise squared MMDs, equivalently as
$\tfrac12\,\mathbb{E}_{\ell,\ell'\sim\pi}[\mathrm{MMD}^2(q_\ell,q_{\ell'})]$ for
independent branch draws, which is estimable directly from branch samples.

Proposition~\ref{prop:diversity} of the main text is proved as follows.

\begin{proof}[Proof of Proposition~\ref{prop:diversity}]
By the expectation-value linearity of the mixture (Appendix~\ref{app:train-embed},
Eq.~\eqref{eq:moiqp-corr}), $z(p_\Theta)=\sum_\ell\pi_\ell\,z(\bm{\theta}^{(\ell)})$
is a convex combination of the branch vectors $z(\bm{\theta}^{(\ell)})\in\mathcal{S}$.
Its distance to $\mathcal{S}$ is therefore bounded by the dispersion:
\begin{align}
\mathrm{dist}\big(z(p_\Theta),\mathcal{S}\big)
&\le \min_\ell\big\|z(p_\Theta)-z(\bm{\theta}^{(\ell)})\big\|_w
\le \sum_\ell\pi_\ell\big\|z(p_\Theta)-z(\bm{\theta}^{(\ell)})\big\|_w
\nonumber\\
&\le \Big(\sum_\ell\pi_\ell\big\|z(p_\Theta)-z(\bm{\theta}^{(\ell)})\big\|_w^2\Big)^{1/2}
= D_{\mathrm{branch}},
\label{eq:dist-bound}
\end{align}
where the second inequality uses $\min\le$ weighted mean and the third is
Cauchy--Schwarz. Let $s^\star\in\mathcal{S}$ attain
$\mathrm{dist}(z(p_\Theta),\mathcal{S})$ (the minimum exists by compactness of
$\mathcal{S}$). The triangle inequality in the $\|\cdot\|_w$ norm then gives
\begin{equation}
\sqrt{\mathcal{L}(\Theta)}=\|t-z(p_\Theta)\|_w
\ge \|t-s^\star\|_w-\|z(p_\Theta)-s^\star\|_w
\ge \Delta_{\text{anc-free}}-D_{\mathrm{branch}},
\end{equation}
since $\|t-s^\star\|_w\ge\Delta_{\text{anc-free}}$ by Eq.~\eqref{eq:floor} and
$\|z(p_\Theta)-s^\star\|_w=\mathrm{dist}(z(p_\Theta),\mathcal{S})\le D_{\mathrm{branch}}$
by Eq.~\eqref{eq:dist-bound}. The collapse and margin statements follow by setting
$D_{\mathrm{branch}}=0$ and by rearranging Eq.~\eqref{eq:diversity-bound}.
\end{proof}

Three remarks fix the scope of the proposition. First, the bound is a
\emph{necessary} condition: it lower-bounds the diversity required for a given
improvement but does not assert that increasing $D_{\mathrm{branch}}$ alone reduces
the loss, since branches may differ in directions orthogonal to $t-\mathcal{S}$.
Second, the result has content only when $\Delta_{\text{anc-free}}>0$, that is when
the target expectation-value vector lies outside the ancilla-free set $\mathcal{S}$, the
non-universality regime in which an ancilla-free IQP circuit cannot reproduce the
data~\cite{kurkin2025universality}; for targets inside $\mathcal{S}$ a single ancilla-free circuit already
suffices and branch coincidence costs nothing. Third, the
statement lives entirely in the low-body MMD metric: $\Delta_{\text{anc-free}}$
is itself an MMD floor, and Eq.~\eqref{eq:diversity-bound} does not lift to total
variation. The proposition thus complements Theorem~\ref{thm:rand-bp}: trainability
makes the mixture optimizable, while the diversity floor identifies branch
differentiation as what the optimization must achieve to exceed a single
circuit. The cluster-initialized scheme of Sec.~\ref{sec:init} provides a
data-aligned starting point for that differentiation.

Proposition~\ref{prop:diversity} establishes that branch diversity is
representationally necessary; we now show that deterministic first-order
gradient descent cannot initiate it at exact branch coincidence and that the
separating signal is suppressed nearby. Working on the uniform-weight
mixed IQP of Eqs.~\eqref{eq:estB-angles}--\eqref{eq:estB-block} (standard
all-zero ancilla input,
$\pi_\ell=1/L$, $L=2^a$), we resolve the loss gradient into the compiled-angle
basis. For a fixed generator $G$ the $S=\varnothing$ compiled angle
$\widetilde{\theta}_{G,\varnothing}$ is the branch-average (system) parameter,
while the $S\neq\varnothing$ compiled angles are the ancilla-coupling
(branch-distinguishing) parameters. Throughout, $g_A^{(G,\ell)}:=
\partial_{\theta_G^{(\ell)}}z_A(\bm{\theta}^{(\ell)})$ is the per-branch
single-circuit expectation-value gradient introduced with
Eq.~\eqref{eq:bp-curv}.

\begin{theoremgradformal}[Ancilla-gradient suppression, formal]
For the uniform ($\pi_\ell=1/L$) mixed IQP-QCBM with the low-body MMD loss
$\mathcal{L}(\Theta)=\sum_A w_A\big(z_A(\Theta)-t_A\big)^2$
[Eq.~\eqref{eq:bp-loss}], fix a generator $G$ and define the Walsh transform
of its per-branch sensitivities,
\begin{equation}
\widehat{g}_A(S):=\sum_{\ell=0}^{L-1}(-1)^{S\cdot\ell}\,g_A^{(G,\ell)},
\end{equation}
with mean
$\bar{g}_A:=\tfrac1L\widehat{g}_A(\varnothing)=\tfrac1L\sum_\ell g_A^{(G,\ell)}$.
Then:
\begin{enumerate}
\item[(i)] \emph{(Exact decomposition.)}
For every $S\subseteq[a]$ the expectation-value gradient is the corresponding
Walsh coefficient,
\begin{equation}
\partial_{\widetilde{\theta}_{G,S}}z_A(\Theta)=\tfrac1L\,\widehat{g}_A(S),
\end{equation}
and the loss gradient inherits the same decomposition,
\begin{equation}
\partial_{\widetilde{\theta}_{G,S}}\mathcal{L}
=\tfrac2L\sum_A w_A\big(z_A(\Theta)-t_A\big)\,\widehat{g}_A(S).
\end{equation}
\item[(ii)] \emph{(Exact zero at coincidence.)}
If the branches coincide, $\bm{\theta}^{(\ell)}=\bar{\bm{\theta}}$ for all
$\ell$, then $\partial_{\widetilde{\theta}_{G,S}}\mathcal{L}=0$ for every
$S\neq\varnothing$, simultaneously for all ancilla-coupling parameters and
independently of the target $t$, while the system gradient
$\partial_{\widetilde{\theta}_{G,\varnothing}}\mathcal{L}
=2\sum_A w_A\big(z_A(\Theta)-t_A\big)\,\bar{g}_A$
is the ordinary single-circuit gradient at $\bar{\bm{\theta}}$ and is
not suppressed by Walsh cancellation.
\item[(iii)] \emph{($O(m_0\delta)$ suppression nearby.)}
Let $\delta:=\max_\ell\|\bm{\theta}^{(\ell)}-\bar{\bm{\theta}}\|_\infty$ with
$\bar{\bm{\theta}}=\tfrac1L\sum_\ell\bm{\theta}^{(\ell)}$. There exists a
constant $C_A\le 8m_0$, uniform in $|A|$ with $m_0=\mathrm{poly}(n)$ the
number of generators on the fixed graph, such that every ancilla-coupling
gradient is bounded by the branch spread, for every $S\neq\varnothing$,
\begin{equation}
\big|\partial_{\widetilde{\theta}_{G,S}}\mathcal{L}\big|
\le 2\delta\sum_A w_A\,C_A\,|z_A(\Theta)-t_A|=O(m_0\,\delta).
\end{equation}
The full prefactor is
$\le 16m_0\,\sup_A|z_A(\Theta)-t_A|$ with $m_0=\mathrm{poly}(n)$, so the suppression is
operative once $\delta=o(1/\mathrm{poly}(n))$, as quantified in
Corollary~\ref{cor:admissible-fluctuation}. The system gradient, by contrast,
remains at the ordinary single-circuit scale unless the branch-average circuit is
near a stationary point; on any region where that system block is bounded below,
the ratio
$\|\nabla_{\mathrm{anc}}\mathcal{L}\|/\|\nabla_{\mathrm{sys}}\mathcal{L}\|$ is
$O(m_0\delta)$ on the $G$-block.
\item[(iv)] \emph{(Parseval characterization.)}
The total first-order sensitivity of $z_A$ to the ancilla-coupling block of $G$
equals the empirical variance of the per-branch gradients across branches,
\begin{equation}
\sum_{S\neq\varnothing}\big(\partial_{\widetilde{\theta}_{G,S}}z_A(\Theta)\big)^2
=\frac1L\sum_{\ell}\big(g_A^{(G,\ell)}\big)^2-\bar{g}_A^2
=\mathrm{Var}_\ell\!\big[g_A^{(G,\ell)}\big].
\label{eq:grad-parseval}
\end{equation}
\end{enumerate}
\end{theoremgradformal}

Part~(iii) is Theorem~\ref{thm:grad-suppress-main} of the main text; part~(ii)
is its $\delta=0$ limit, and parts~(i) and~(iv) expose the Walsh mechanism
that produces the suppression.

\begin{proof}
\emph{Step 1 (chain rule into the compiled angles).}
The forward Walsh--Hadamard map of Eq.~\eqref{eq:estB-walsh} expresses each branch
angle as a sign-weighted sum of the compiled angles,
\begin{equation}
\theta_G^{(\ell)}=\sum_{S\subseteq[a]}(-1)^{S\cdot\ell}\,
\widetilde{\theta}_{G,S},
\end{equation}
so that $\partial\theta_G^{(\ell)}/\partial\widetilde{\theta}_{G,S}=(-1)^{S\cdot\ell}$,
and $\widetilde{\theta}_{G,S}$ enters $z_A$ only through the branch angles
$\{\theta_G^{(\ell)}\}_\ell$. With $z_A(\Theta)=\tfrac1L\sum_\ell
z_A(\bm{\theta}^{(\ell)})$ from Eq.~\eqref{eq:bp-loss} and the chain rule,
\begin{equation}
\partial_{\widetilde{\theta}_{G,S}}z_A(\Theta)
=\sum_{\ell}\frac1L\,g_A^{(G,\ell)}\,(-1)^{S\cdot\ell}
=\frac1L\,\widehat{g}_A(S).
\label{eq:grad-step1}
\end{equation}
For $S=\varnothing$ this is $\bar{g}_A$, the mean per-branch gradient; for
$S\neq\varnothing$ it is the $S$-th Walsh fluctuation coefficient of the
per-branch gradients.

\emph{Step 2 (the loss gradient inherits the decomposition).}
Differentiating $\mathcal{L}=\sum_A w_A(z_A(\Theta)-t_A)^2$ and inserting
Eq.~\eqref{eq:grad-step1},
\begin{equation}
\partial_{\widetilde{\theta}_{G,S}}\mathcal{L}
=\sum_A 2w_A\big(z_A(\Theta)-t_A\big)\,
\partial_{\widetilde{\theta}_{G,S}}z_A(\Theta)
=\frac2L\sum_A w_A\big(z_A(\Theta)-t_A\big)\,\widehat{g}_A(S),
\label{eq:grad-step2}
\end{equation}
which is statement (i).

\emph{Step 3 (exact zero at branch coincidence).}
If $\bm{\theta}^{(\ell)}=\bar{\bm{\theta}}$ for all $\ell$, then
$g_A^{(G,\ell)}=g_A$ is independent of $\ell$, so Walsh orthogonality
$\sum_\ell(-1)^{S\cdot\ell}=L\,\delta_{S,\varnothing}$ [the
$\ell'=\bm{0}$ case of Eq.~\eqref{eq:estB-walsh-orth}, applied in the
$\ell$-sum after transposing it with $H=H^{T}$] gives
$\widehat{g}_A(S)=g_A\sum_\ell(-1)^{S\cdot\ell}=g_A\,L\,\delta_{S,\varnothing}$.
Hence for every $S\neq\varnothing$ both $\partial_{\widetilde{\theta}_{G,S}}z_A
=0$ and, by Eq.~\eqref{eq:grad-step2}, $\partial_{\widetilde{\theta}_{G,S}}
\mathcal{L}=0$, regardless of $t$ and simultaneously for all coupling
parameters, while
$\partial_{\widetilde{\theta}_{G,\varnothing}}\mathcal{L}
=2\sum_A w_A(z_A(\Theta)-t_A)\,\bar{g}_A$ equals the single-circuit gradient at
$\bar{\bm{\theta}}$, which is not subject to the Walsh cancellation. This is
statement (ii).

\emph{Step 4 ($O(m_0\delta)$ suppression near coincidence).}
The per-branch gradient $g_A^{(G,\ell)}=\partial_{\theta_G}z_A(\bm{\theta}^{(\ell)})$
depends on every generator angle $\theta_{G'}^{(\ell)}$ whose generator $G'$
anticommutes with $Z_A$ (i.e.\ $G'\cdot A$ odd). In general the $L$ branch angle
vectors $\bm{\theta}^{(\ell)}$ may differ in all $m_0$ coordinates (the spread
$\delta=\max_\ell\|\bm{\theta}^{(\ell)}-\bar{\bm{\theta}}\|_\infty$ is measured in
$\|\cdot\|_\infty$ over the full branch vector), so the Lipschitz step must be
taken in the full branch angle vector, not in $\theta_G$ alone. By the uniform
per-branch regularity bound $|\partial^2 z_A|\le4$ of the trainability lemma
(Lemma~\ref{lem:curv-bp}, valid at any angles and used throughout
Appendix~\ref{app:bp}), every mixed second derivative obeys
$|\partial_{\theta_{G'}}\partial_{\theta_G}z_A|\le4$,
independently of $|A|$. Summing over the at most $m_0$ generators that
anticommute with $Z_A$ gives the branch-vector Lipschitz constant of
$g_A^{(G,\ell)}$,
\begin{equation}
L_A\le 4\,\#\{G':G'\cdot A\ \text{odd}\}\le 4m_0,\qquad m_0=\mathrm{poly}(n).
\end{equation}
The single-coordinate bound $L_A\le4$ would hold only if $\theta_G$ alone varied
across branches; in general the $m_0$ factor is required, matching the $m$-factor
the trainability lemma accumulates across parameters [e.g.\ $\beta_1\le48(m-1)$ of
Eq.~\eqref{eq:bp-beta1-mixture}].

Write each per-branch gradient as its branch mean plus a deviation,
$g_A^{(G,\ell)}=\bar{g}_A+\varepsilon_A^{(\ell)}$ with $\sum_\ell\varepsilon_A^{(\ell)}=0$.
Every $\bm{\theta}^{(\ell)}$ lies within $\delta$ of $\bar{\bm{\theta}}$, hence
within $2\delta$ of any other branch, so the deviation is bounded by the largest
pairwise difference,
\begin{equation}
|\varepsilon_A^{(\ell)}|
\le\max_{\ell'}\big|g_A^{(G,\ell)}-g_A^{(G,\ell')}\big|
\le L_A\max_{\ell'}\|\bm{\theta}^{(\ell)}-\bm{\theta}^{(\ell')}\|_\infty
\le 2L_A\delta=:C_A\delta,
\end{equation}
with $C_A\le8m_0$ uniformly in $|A|$. For $S\neq\varnothing$ the constant mean
cancels under Walsh orthogonality, $\sum_\ell(-1)^{S\cdot\ell}\bar{g}_A=0$, so only
the deviations survive in the compiled coefficient,
$\widehat{g}_A(S)=\sum_\ell(-1)^{S\cdot\ell}\varepsilon_A^{(\ell)}$, and
\begin{equation}
\Big|\tfrac1L\widehat{g}_A(S)\Big|
\le\tfrac1L\sum_\ell|\varepsilon_A^{(\ell)}|
\le\max_\ell|\varepsilon_A^{(\ell)}|\le C_A\delta.
\end{equation}
Substituting into Eq.~\eqref{eq:grad-step2} bounds the full ancilla-coupling
gradient,
\begin{equation}
\big|\partial_{\widetilde{\theta}_{G,S}}\mathcal{L}\big|
\le\tfrac2L\sum_A w_A|z_A(\Theta)-t_A|\,|\widehat{g}_A(S)|
\le2\delta\sum_A w_A C_A|z_A(\Theta)-t_A|=O(\delta).
\end{equation}
The system (branch-average) gradient, by contrast, is the ordinary single-circuit
gradient at $\bar{\bm{\theta}}$ and is not suppressed by the Walsh cancellation;
when it is bounded below, the ratio of the two yields statement~(iii).

\emph{Step 5 (Parseval characterization).}
Walsh--Parseval applied to the per-branch gradients (the orthogonality of
Eq.~\eqref{eq:estB-walsh-orth}) gives
\begin{equation}
\sum_{S\subseteq[a]}\widehat{g}_A(S)^2
=\sum_{S}\sum_{\ell,\ell'}(-1)^{S\cdot(\ell\oplus\ell')}
g_A^{(G,\ell)}g_A^{(G,\ell')}
=L\sum_\ell\big(g_A^{(G,\ell)}\big)^2.
\end{equation}
Isolating the $S=\varnothing$ term gives $\widehat{g}_A(\varnothing)^2=(L\bar{g}_A)^2$;
dividing by $L^2$ as in Eq.~\eqref{eq:grad-step1} then yields
\begin{equation}
\sum_{S\neq\varnothing}\big(\partial_{\widetilde{\theta}_{G,S}}z_A(\Theta)\big)^2
=\frac{1}{L^2}\sum_{S\neq\varnothing}\widehat{g}_A(S)^2
=\frac1L\sum_\ell\big(g_A^{(G,\ell)}\big)^2-\bar{g}_A^2
=\mathrm{Var}_\ell\!\big[g_A^{(G,\ell)}\big],
\end{equation}
which is Eq.~\eqref{eq:grad-parseval} and statement (iv). Here
$\mathrm{Var}_\ell[g_A^{(G,\ell)}]=\tfrac1L\sum_\ell(g_A^{(G,\ell)})^2-\bar{g}_A^2$
is the population variance of the per-branch gradients over the $L$ branches.
\end{proof}

The theorem records a \emph{relative}, subspace-level suppression. The system
(branch-average) block trains as an ordinary single circuit; the ancilla-coupling
block is first-order flat exactly at branch coincidence and within $O(\delta)$ of
it, by Eq.~\eqref{eq:grad-parseval} in proportion to the dispersion of the
per-branch gradients $\mathrm{Var}_\ell[g_A^{(G,\ell)}]$. It is not a global
barren plateau and makes no claim for well-separated branches, where
$\delta=\Theta(1)$ and the coupling gradient is generically $\Theta(1)$. The
consequence for optimization is that gradient descent started from a collapsed
(global) initialization receives no first-order signal to differentiate the
branches. The diversity required by Proposition~\ref{prop:diversity} cannot
be acquired by a deterministic first-order update while exact coincidence is
maintained. It can emerge after a perturbation or higher-order motion breaks
the coincidence; cluster initialization instead supplies data-aligned separation
without relying on this coincidence-breaking step.

\emph{Remark (curvature decides escapability).}
The first-order flatness of statement (ii) is data-independent (it follows from
Walsh orthogonality alone, for any target), so it certifies that the
ancilla-coupling subspace is gradient-flat at coincidence but does not by itself
separate an escapable saddle from a trapped minimum; that distinction is
second-order. At a branch-coincident point we use the IQP correlator identity
$\partial^2_{\theta_G^{(\ell)}}z_A=-4z_A$ (the identity underlying
Eq.~\eqref{eq:bp-curv}) together with
$\partial\theta_G^{(\ell)}/\partial\widetilde{\theta}_{G,S}=(-1)^{S\cdot\ell}$
(whose square is $1$). The ancilla-coupling curvature is then, for every
$S\neq\varnothing$,
\begin{equation}
\partial^2_{\widetilde{\theta}_{G,S}}\mathcal{L}
=-8\!\!\sum_{\substack{A:\\ \partial^2_{\theta_G}z_A\neq0}}\!\!
w_A\,z_A(\Theta)\big(z_A(\Theta)-t_A\big),
\end{equation}
independent of $S$ and set in sign by the data mismatch; for a one-body generator
$G=\{j\}$ the restricted sum is over $A\ni j$, matching
Eq.~\eqref{eq:bp-curv}. A negative value makes the
coincident point a saddle along the ancilla directions, the only
first-order-invisible route by which branch diversity can be acquired. This
second-order analysis is governed by the curvature machinery of
Lemma~\ref{lem:curv-bp}, as instantiated at the unbiased and data-dependent centers in
Theorems~\ref{thm:rand-bp} and~\ref{thm:datadep-bp-main}$'$, which we do not re-derive
here.

The same mechanism accounts qualitatively for the measured ancilla-to-system
coupling strengths across datasets: when the modes are not separated in the
low-body expectation values (the MNIST handwritten digits) the per-branch gradients differ little,
$\mathrm{Var}_\ell[g_A^{(G,\ell)}]$ is small, and the trained coupling angles
remain near their initial values, whereas for Hamming-separated or spin-glass
datasets (binary blobs, the 2D Ising model, the D-Wave samples) the branches carry
genuinely different gradients and the coupling block activates.

The following corollary makes the trade-off between diversity and trainability quantitative.

\begin{corollary}[Admissible coincidence-breaking fluctuation]
\label{cor:admissible-fluctuation}
At the unbiased center of Theorem~\ref{thm:rand-bp}, let
$m=Lm_0$, $c_j=8w_{\{j\}}/L^2$, and
\[
B_1(L)=48(Lm_0-1).
\]
For any fixed $\Delta\in(0,1)$, one may choose an inverse-polynomial patch
half-width $r_L$ satisfying
\begin{equation}
r_L^2\le
\frac{\Delta c_j^2}{2B_1(L)c_j+\beta_2},
\qquad
r_L\le\min\!\left\{\frac38,\frac{1}{nm_0}\right\}.
\label{eq:cor-radius}
\end{equation}
This is sufficient for Lemma~\ref{lem:curv-bp} because
$\beta_1\le B_1(L)$, and the choice can be made nonincreasing in $L$.
Throughout this patch the branch spread obeys
\begin{equation}
\delta\le 2r_L,
\label{eq:cor-delta-bound}
\end{equation}
and every branch-distinguishing gradient satisfies
\begin{equation}
\big|\partial_{\widetilde{\theta}_{G,S}}\mathcal{L}\big|
\le32m_0\delta\le64m_0r_L\le\frac{64}{n},
\qquad S\neq\varnothing.
\label{eq:cor-grad-bound}
\end{equation}
The variance lower bound of Lemma~\ref{lem:curv-bp} remains
inverse-polynomial on this smaller patch. The corollary controls only a local
near-coincident region: it neither identifies its maximal radius nor claims
that the finite perturbations used in the experiments lie inside it.
\end{corollary}

\begin{proof}
At the unbiased center $\Theta^{*}$ all branches share one angle vector, so the
perturbed branch angles satisfy
$\|\bm{\theta}^{(\ell)}-\Theta^{*}\|_\infty\le r_L$
for every $\ell$, and their mean $\bar{\bm{\theta}}=\tfrac1L\sum_\ell
\bm{\theta}^{(\ell)}$ likewise satisfies
$\|\bar{\bm{\theta}}-\Theta^{*}\|_\infty\le r_L$. The triangle inequality
gives Eq.~\eqref{eq:cor-delta-bound}. Moreover, since
$\sum_Aw_A=1$, $|z_A-t_A|\le2$, and $C_A\le8m_0$,
Theorem~\ref{thm:grad-suppress-main}$'$ gives
$|\partial_{\widetilde{\theta}_{G,S}}\mathcal{L}|
\le32m_0\delta$, which proves Eq.~\eqref{eq:cor-grad-bound}.

It remains to verify the stated choice of radius.
Substituting $c_j=8w_{\{j\}}/L^2$ and $B_1(L)=48(Lm_0-1)$ into the first bound
of Eq.~\eqref{eq:cor-radius} gives the conservative squared radius
\begin{equation}
R_L^2=
\frac{64\Delta w_{\{j\}}^2}
{768\,w_{\{j\}}(Lm_0-1)L^2+\beta_2L^4}.
\label{eq:cor-r2-monotone}
\end{equation}
For fixed $n$, its denominator is increasing in $L$, so $R_L$ is
nonincreasing. Since $L$ and $m_0$ are polynomial, $R_L$ is bounded below by
an inverse polynomial. Taking
$r_L=\min\{R_L,3/8,1/(nm_0)\}$ therefore preserves both the admissibility
condition and an inverse-polynomial variance bound, while making
Eq.~\eqref{eq:cor-grad-bound} explicit.
\end{proof}

\section{\label{app:experiments}Details of the numerical experiments}

This appendix records the numerical settings used in Sec.~\ref{sec:result}
(the four datasets, the classical baseline, the branch weights and sample
allocation, and the dataset-specific training settings).

\subsection{\label{app:datasets}Dataset details}

\emph{Binary blobs ($n=16$).} This dataset is a binary analog of Gaussian
blobs~\cite{recio2025}. A sample is generated by choosing one of eight fixed
$16$-bit patterns uniformly and flipping each bit independently with
probability $\eta=0.05$:
\begin{equation}
p_{\text{data}}(\bm{z})=\frac{1}{8}\sum_{m=1}^{8}
 \eta^{\,d_H(\bm{z},\bm{b}^{(m)})}\,(1-\eta)^{\,n-d_H(\bm{z},\bm{b}^{(m)})},
\label{eq:blobs}
\end{equation}
where $d_H$ is the Hamming distance. The resulting distribution has eight
well-separated modes. We use $5000$ training and $10000$ test strings.

\emph{2D Ising ($n=16$).} The target is the thermal distribution of a
classical Ising model on a periodic $4\times4$ square lattice at temperature
$T=3$. With $s_j=1-2z_j\in\{\pm1\}$,
\begin{equation}
\mathcal{E}(\bm{z})=-\!\!\sum_{\langle j,k\rangle}\!\! J_{jk}\,s_j s_k,
\qquad
p_{\text{data}}(\bm{z})\propto e^{-\mathcal{E}(\bm{z})/T},
\label{eq:ising}
\end{equation}
where $\langle j,k\rangle$ runs over the nearest-neighbor bonds of the lattice
and the fixed couplings $J_{jk}$ are drawn independently and uniformly from
$[0,2]$.
There are no local fields, so the distribution is invariant under a global
spin flip. Samples are generated with independent Metropolis--Hastings chains
following Ref.~\cite{recio2025}; we use $5000$ training and $50000$ test
configurations.

\emph{MNIST ($28\times28$, $n=784$).} Each gray-scale MNIST digit
image~\cite{lecun1998} is converted to a $784$-bit string. For a pixel value
$g_j\in\{0,\ldots,255\}$, we set $z_j=1$ when $g_j>128$ and $z_j=0$
otherwise. From the original $60000$-image training corpus,
the first $50000$ images define the training pool and the final $10000$ define
a disjoint test pool. We draw $2000$ images from each pool.

\emph{D-Wave spin glass ($n=484$).} Each sample is a $484$-bit spin
configuration produced on a D-Wave processor with Pegasus connectivity and a
$100\,\mu\mathrm{s}$ quench. The dataset was collected by Scriva
\textit{et al.}~\cite{scriva2023} and introduced as a generative-modeling
benchmark in Ref.~\cite{recio2025}. From its $10000$ training and $60000$ test
configurations, we use disjoint subsets of $2000$ configurations each.

Figure~\ref{fig:app-modesep-full} completes the branch panels summarized in
the main text before the appendix turns to optimization settings.

\begin{figure*}[!ht]
\includegraphics[width=\textwidth]{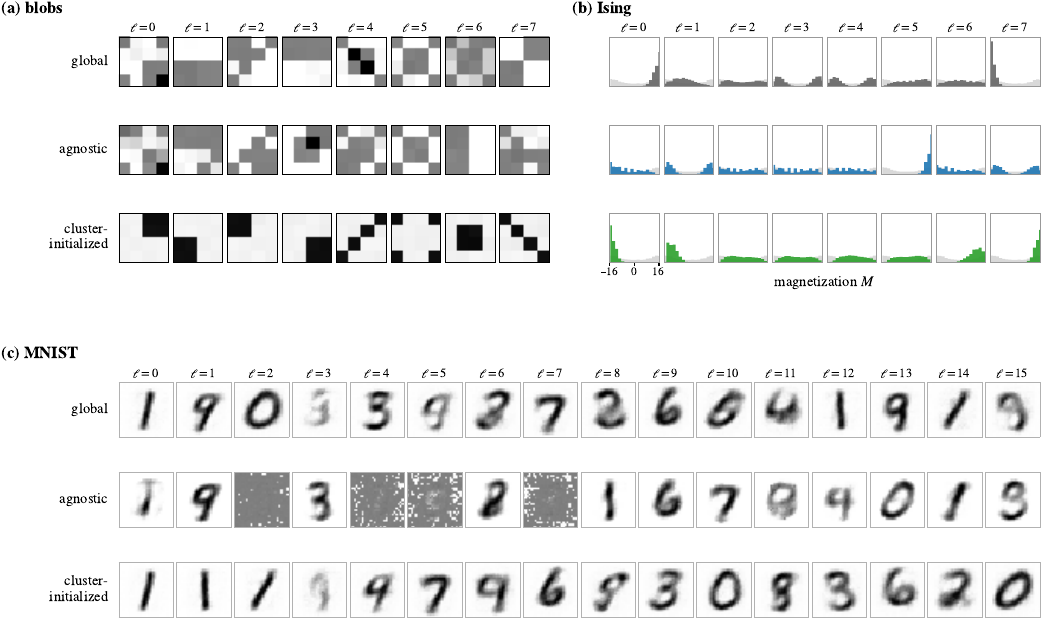}
\caption{\label{fig:app-modesep-full} Every branch of the mixtures shown in
Fig.~\ref{fig:modesep}: (a)~binary blobs and (b)~2D Ising at $a=3$ (eight
branches each), (c)~binarized MNIST at $a=4$ (sixteen branches). Rows,
colours, and the per-scheme lowest-diversity seed choice are as in
Fig.~\ref{fig:modesep}; the Ising panel shows the $a=3$ mixture at seed
$7$.}
\end{figure*}

\subsection{\label{app:hyperparams}IQP-QCBM hyperparameters}

Table~\ref{tab:hyperparams} gives the complete circuit, optimization, and
estimator settings used for the numerical benchmarks. The MNIST
data-agnostic runs are listed separately because they use a smaller learning
rate, $\eta=5\times10^{-4}$, and a longer optimization budget. These settings
were used because the shared setting did not converge consistently across
the tested seeds. At the shared rate $\eta=10^{-3}$, a subset of the five
data-agnostic runs fails to converge on this dataset
(Table~\ref{tab:mnist-da-rate}):
retrained for $20{,}000$ steps at $\eta=10^{-3}$, three of the five seeds
plateau at $a=1$ with a test MMD$^2$ more than an order of magnitude above the
converged value, the across-seed standard deviation is about $10^{3}$
times that at $\eta=5\times10^{-4}$, and the seeds that fail differ with the
ancilla count. At $\eta=5\times10^{-4}$ this plateau failure is absent for
$a\le3$ in all five runs, and the
$20{,}000$-step budget, against $5000$ for the other schemes, is then needed
for the data-agnostic center to reach its plateau. At $a=0$ the single shared
circuit already converges within $5000$ steps, so the $5000$-step value is
reported there. Training bandwidths are ordered
by $\bar m=(2,6)$ and obtained from Eq.~\eqref{eq:weights}. The
Monte-Carlo sample sizes of the training and evaluation estimators are
listed with the shared settings in Table~\ref{tab:hyperparams}.

\begin{table}[!ht]
\caption{\label{tab:hyperparams}Complete IQP-QCBM settings for the
numerical benchmarks: per-dataset circuit (top), optimization and
evaluation (middle), and settings shared across the benchmarks (bottom).
The parameter count is $(N_1+N_2)2^{a^\star}$, with $a^\star$ in
parentheses.}
\begin{ruledtabular}
\footnotesize
\setlength{\tabcolsep}{4pt}
\renewcommand{\arraystretch}{1.0}
\begin{tabular}{l c l r c r}
Dataset & $n$ & layout & \shortstack{gates/branch\\$(N_1+N_2)$}
& $a$ sweep & params ($a^\star$) \\
\colrule
Binary blobs & $16$ & fully connected order-2 & $16+120$ & $0$--$3$
  & $1088$ ($3$) \\
2D Ising & $16$ & fully connected order-2 & $16+120$ & $0$--$3$
  & $1088$ ($3$) \\
MNIST & $784$ & row/column order-2 & $784+21{,}168$ & $0$--$4$
  & $351{,}232$ ($4$) \\
D-Wave spin glass & $484$ & fully connected order-2 & $484+116{,}886$ & $0$--$3$
  & $938{,}960$ ($3$) \\
\end{tabular}
\end{ruledtabular}
\begin{ruledtabular}
\footnotesize
\setlength{\tabcolsep}{4pt}
\renewcommand{\arraystretch}{1.0}
\begin{tabular}{l c l c c r}
Configuration & $\eta$ & steps & \multicolumn{2}{c}{training $\sigma$ $(\bar m=2,6)$}
  & test channel \\
\colrule
Binary blobs & $10^{-3}$ & $5000$ & \multicolumn{2}{c}{$(1.3183,0.6006)$} & exact \\
2D Ising & $10^{-3}$ & $5000$ & \multicolumn{2}{c}{$(1.3183,0.6006)$} & exact \\
MNIST, global/cluster & $10^{-3}$ & $5000$ & \multicolumn{2}{c}{$(9.8868,5.6935)$} & correlator \\
MNIST, DA & $5\times10^{-4}$ & $5000/20{,}000$ & \multicolumn{2}{c}{$(9.8868,5.6935)$} & correlator \\
D-Wave spin glass & $10^{-3}$ & $5000$ & \multicolumn{2}{c}{$(7.7621,4.4627)$} & correlator \\
\end{tabular}
\end{ruledtabular}
\begin{ruledtabular}
\footnotesize
\setlength{\tabcolsep}{4pt}
\renewcommand{\arraystretch}{1.12}
\begin{tabular}{l p{0.74\textwidth}}
Optimizer and repetitions
  & Adam~\cite{kingma2015adam}; ten independent training seeds per
    configuration for binary blobs and Ising, and five for MNIST and
    D-Wave. \\
Training loss
  & Mean of the two MMD$^2$ kernels listed in Table~\ref{tab:hyperparams}. \\
Training estimate
  & $1000$ Pauli-$Z$ operators per kernel and step; their expectations use
    $3000$ uniform random bit strings in Eq.~\eqref{eq:estA-samplemean}. \\
Test sweep
  & $\bar m=1,\ldots,6$; exact evaluation for $n=16$. \\
Estimated test
  & For $n>16$: $2000$ Pauli-$Z$ operators and $5000$ uniform random bit
    strings per estimator seed; median over five seeds. \\
Initial perturbation
  & Each angle receives independent noise from
    $\operatorname{Unif}[-\pi/(2\sqrt m),\,\pi/(2\sqrt m)]$, where $m$ is the
    total trainable angle count. \\
\end{tabular}
\end{ruledtabular}
\end{table}

\begin{table}[!ht]
\caption{\label{tab:mnist-da-rate}Shared versus reduced learning rate for
the data-agnostic MNIST runs: test MMD$^2$ ($\times10^{-3}$, averaged over
the six evaluation bandwidths as in the main text) after $20{,}000$ steps,
for each of the five training seeds. At the shared $\eta=10^{-3}$ a subset
of seeds plateaus well above the converged value, and which seeds fail
changes with $a$; at $\eta=5\times10^{-4}$ this plateau failure is absent
for $a\le3$. The seed spread at $a=4$, present at both rates, is the
data-agnostic exception discussed in Sec.~\ref{sec:exp-train}.}
\begin{ruledtabular}
\footnotesize
\begin{tabular}{l c cccc}
$\eta$ & seed & $a=1$ & $a=2$ & $a=3$ & $a=4$ \\
\colrule
$10^{-3}$ (shared) & $7$ & $2.74$ & $9.49$ & $3.47$ & $20.0$ \\
 & $11$ & $2.65$ & $1.98$ & $3.53$ & $15.0$ \\
 & $23$ & $46.0$ & $1.83$ & $19.4$ & $23.5$ \\
 & $1234$ & $49.4$ & $1.88$ & $19.8$ & $22.9$ \\
 & $2025$ & $47.7$ & $1.96$ & $14.9$ & $23.8$ \\
\colrule
$5\times10^{-4}$ (used) & $7$ & $2.67$ & $3.58$ & $2.99$ & $6.36$ \\
 & $11$ & $2.68$ & $1.85$ & $3.49$ & $7.16$ \\
 & $23$ & $2.71$ & $1.85$ & $1.71$ & $13.2$ \\
 & $1234$ & $2.70$ & $1.83$ & $2.15$ & $12.1$ \\
 & $2025$ & $2.71$ & $1.90$ & $1.33$ & $17.7$ \\
\end{tabular}
\end{ruledtabular}
\end{table}

\begin{table}[!ht]
\caption{\label{tab:all-branch-weights}Realized branch weights for the four
numerical benchmarks. Panel (a) gives the complete vector
$\bm\pi^{(r)}=(\pi_0,\ldots,\pi_{2^a-1})$ for $a\le3$; panel (b) gives the
longer MNIST $a=4$ vectors. Entries are
ordered by binary branch index and rounded to four decimal places; exactly,
$\pi_\ell=|C_\ell|/N_{\rm train}$. Global and data-agnostic mixtures use
$\pi_\ell=2^{-a}$. A--J label the ten $n=16$ runs and A--E the five
large-system runs; joined labels indicate identical rounded vectors.}
\begin{ruledtabular}
\footnotesize
\setlength{\tabcolsep}{3pt}
\begin{tabular}{l c c l}
\multicolumn{4}{l}{\textit{(a) Runs with $a\le3$}} \\
Target ($N_{\rm train}$) & $a$ & Run & Branch-weight vector $\bm\pi^{(r)}$ \\
\colrule
Binary blobs ($5000$)
 & $1$ & A--J & $(0.5576,0.4424)$ \\
 & $2$ & A & $(0.3798,0.2678,0.1728,0.1796)$ \\
 & $2$ & B,G & $(0.3830,0.2644,0.1786,0.1740)$ \\
 & $2$ & C,D,F,H,J & $(0.3830,0.2644,0.1746,0.1780)$ \\
 & $2$ & E & $(0.3900,0.2688,0.1666,0.1746)$ \\
 & $2$ & I & $(0.3834,0.2722,0.1708,0.1736)$ \\
 & $3$ & A--J & $(0.1260,0.1300,0.1256,0.1214,0.1230,0.1220,0.1230,0.1290)$ \\
\colrule
2D Ising ($5000$)
 & $1$ & A--J & $(0.5008,0.4992)$ \\
 & $2$ & A,C,D,F,I,J & $(0.3258,0.1804,0.1700,0.3238)$ \\
 & $2$ & B,E,H & $(0.3258,0.1802,0.1698,0.3242)$ \\
 & $2$ & G & $(0.3258,0.1792,0.1696,0.3254)$ \\
 & $3$ & A & $(0.1872,0.0892,0.1060,0.1014,0.1034,0.1008,0.1006,0.2114)$ \\
 & $3$ & B & $(0.1906,0.0984,0.0986,0.1148,0.1080,0.0992,0.0942,0.1962)$ \\
 & $3$ & C & $(0.1962,0.0986,0.0988,0.1014,0.1044,0.0960,0.0990,0.2056)$ \\
 & $3$ & D & $(0.1922,0.0974,0.1056,0.1048,0.1072,0.1014,0.0952,0.1962)$ \\
 & $3$ & E & $(0.1904,0.0996,0.1072,0.1006,0.1056,0.1030,0.0928,0.2008)$ \\
 & $3$ & F & $(0.1944,0.1040,0.1036,0.1074,0.1002,0.1026,0.0906,0.1972)$ \\
 & $3$ & G & $(0.1922,0.0994,0.1100,0.1070,0.0958,0.1062,0.0906,0.1988)$ \\
 & $3$ & H & $(0.1942,0.0958,0.0996,0.1068,0.1088,0.1000,0.0932,0.2016)$ \\
 & $3$ & I & $(0.1976,0.1026,0.0984,0.1082,0.0946,0.1054,0.0980,0.1952)$ \\
 & $3$ & J & $(0.1914,0.0966,0.1064,0.1058,0.1106,0.1002,0.0914,0.1976)$ \\
\colrule
MNIST ($2000$)
 & $1$ & A--E & $(0.5030,0.4970)$ \\
 & $2$ & A,B,C,E & $(0.2230,0.2875,0.2345,0.2550)$ \\
 & $2$ & D & $(0.2235,0.2870,0.2350,0.2545)$ \\
 & $3$ & A,D & $(0.1100,0.1370,0.1440,0.1365,0.0940,0.1635,0.1140,0.1010)$ \\
 & $3$ & B & $(0.1100,0.1370,0.1440,0.1365,0.0945,0.1635,0.1140,0.1005)$ \\
 & $3$ & C & $(0.1100,0.1370,0.1435,0.1370,0.0940,0.1640,0.1140,0.1005)$ \\
 & $3$ & E & $(0.1100,0.1365,0.1440,0.1370,0.0940,0.1640,0.1135,0.1010)$ \\
\colrule
D-Wave ($2000$)
 & $1$ & A--E & $(0.5065,0.4935)$ \\
 & $2$ & A,D & $(0.2750,0.1870,0.3280,0.2100)$ \\
 & $2$ & B,C,E & $(0.2755,0.1865,0.3280,0.2100)$ \\
 & $3$ & A & $(0.0730,0.1060,0.1685,0.1310,0.1590,0.1350,0.1300,0.0975)$ \\
 & $3$ & B,E & $(0.1730,0.1695,0.1070,0.1310,0.0965,0.0960,0.1290,0.0980)$ \\
 & $3$ & C & $(0.1730,0.1690,0.1070,0.1315,0.0965,0.0960,0.1290,0.0980)$ \\
 & $3$ & D & $(0.0735,0.1060,0.1690,0.1300,0.1590,0.1350,0.1300,0.0975)$ \\
\colrule
\multicolumn{4}{l}{\textit{(b) MNIST, $a=4$ ($N_{\rm train}=2000$)}} \\
\multicolumn{1}{c}{Run} & \multicolumn{3}{c}{Branch-weight vector $\bm\pi^{(r)}$} \\
\colrule
\multicolumn{1}{c}{A} & \multicolumn{3}{l}{\scriptsize
  $(0.0510,0.0515,0.0435,0.0720,0.0695,0.0765,0.0750,0.0565,0.0510,0.0495,0.0715,0.0785,0.0515,0.0770,0.0720,0.0535)$} \\
\multicolumn{1}{c}{B} & \multicolumn{3}{l}{\scriptsize
  $(0.0495,0.0520,0.0440,0.0770,0.0750,0.0690,0.0590,0.0685,0.0515,0.0695,0.0510,0.0800,0.0535,0.0735,0.0725,0.0545)$} \\
\multicolumn{1}{c}{C} & \multicolumn{3}{l}{\scriptsize
  $(0.0500,0.0560,0.0440,0.0760,0.0745,0.0720,0.0660,0.0595,0.0455,0.0720,0.0515,0.0770,0.0755,0.0535,0.0755,0.0515)$} \\
\multicolumn{1}{c}{D} & \multicolumn{3}{l}{\scriptsize
  $(0.0540,0.0520,0.0435,0.0740,0.0680,0.0640,0.0795,0.0580,0.0505,0.0545,0.0725,0.0760,0.0765,0.0520,0.0720,0.0530)$} \\
\multicolumn{1}{c}{E} & \multicolumn{3}{l}{\scriptsize
  $(0.0520,0.0515,0.0445,0.0725,0.0735,0.0780,0.0645,0.0570,0.0515,0.0700,0.0485,0.0795,0.0525,0.0775,0.0725,0.0545)$} \\
\end{tabular}
\end{ruledtabular}
\end{table}

\subsection{\label{app:branch-weights}Branch weights and clustering settings}

The reported experiments use two branch-weight rules, fixed before joint
training. For the global and data-agnostic initializations, the ancilla branches
are sampled uniformly,
\begin{equation}
\pi_\ell=\frac{1}{L},\qquad \ell=0,\ldots,L-1,
\label{eq:exp-uniform-weights}
\end{equation}
for every dataset, ancilla count, and optimization seed (with the single-circuit
case $a=0$ understood as $\pi_0=1$). For the cluster-initialized scheme,
spectral clustering~\cite{ng2001spectral,vonluxburg2007} partitions the
$N_{\mathrm{train}}$ training strings into
$L=2^a$ nonempty groups $\{C_\ell\}$ and the mixture weight is the empirical
cluster mass,
\begin{equation}
\pi_\ell=\frac{|C_\ell|}{N_{\mathrm{train}}}.
\label{eq:exp-cluster-weights}
\end{equation}
The partition is constructed from the affinity of Eq.~\eqref{eq:affinity}.
For each training set, let
\begin{equation*}
d_{\rm med}=\operatorname{median}\{d_H(\bm x,\bm y):d_H(\bm x,\bm y)>0\}.
\end{equation*}
We set $2\sigma_c^2=d_{\rm med}$, or equivalently
$\sigma_c=\sqrt{d_{\rm med}/2}$, so that
$A_{\bm x\bm y}=\exp[-d_H(\bm x,\bm y)/d_{\rm med}]$. Zero distances from
duplicate strings are excluded. This median rule is applied without tuning
and gives $(d_{\rm med},\sigma_c)=(6,1.732)$ for the blobs, $(8,2.000)$ for
Ising, $(134,8.185)$ for MNIST, and $(238,10.909)$ for D-Wave. The same
$\sigma_c$ is used for every ancilla count and run on a given training set.
We use the leading $L$ eigenvectors of the symmetrically normalized affinity,
normalize each row of the resulting spectral embedding, and apply unconstrained
$k$-means using the standard distance-weighted initialization for the cluster
centers. We repeat the clustering from $20$ independent initializations and
allow at most $150$ iterations per run. No
equal-size constraint is applied. This initialization-only affinity is separate
from the two MMD kernels used to train the circuit, and neither the affinity nor
the resulting assignments are updated during optimization. The spectral
clustering is rerun for each independent optimization run, so the realized
cluster masses can vary between runs; the branch index is only a cluster label
and carries no ordering from one run to another. These weights remain fixed
during optimization and evaluation: only the branch angles are trained.
Figure~\ref{fig:app-modesep-full} resolves every branch of the three mixtures
that Fig.~\ref{fig:modesep} samples.

Table~\ref{tab:all-branch-weights} gives the realized weight of every branch
in the four numerical benchmarks. Vectors are ordered by the binary branch
index $\ell=0,\ldots,2^a-1$ and are shown to four decimal places; their exact
values are $|C_\ell|/N_{\rm train}$. Runs A--J denote the ten independent
$n=16$ runs and A--E the five MNIST and D-Wave runs; a joined label means
that those runs produced the same rounded vector. Branch labels need not
identify the same cluster in different runs.
Figure~\ref{fig:modesep}(b) reorders the Ising branches by
mean magnetization and normalizes each density separately, so its curve
heights do not encode $\pi_\ell$.

\FloatBarrier

\emph{Equal-weight control.} The cluster-initialized scheme differs from the
two coincident schemes in its starting angles and in its branch weights,
which are the data-determined cluster masses rather than $1/L$. To separate
the two differences, we repeated cluster initialization with the assignment
step of the spectral clustering replaced by a capacity-balanced assignment
to the same $k$-means centers: every group then holds
$N_{\mathrm{train}}/L$ strings up to rounding, so the unchanged weight rule
$\pi_\ell=|C_\ell|/N_{\mathrm{train}}$ yields uniform weights by
construction while the branches still start moment-matched to their groups.
Training and evaluation are otherwise identical, with three training seeds
per configuration. Table~\ref{tab:eqweight-control} compares the two
cluster variants against the global scheme at the configurations whose
realized cluster masses deviate most from uniform
(Table~\ref{tab:all-branch-weights}). On the blobs, the Ising, and MNIST
the equal-weight variant matches the cluster-initialized scheme within the
seed scatter, so the advantage over the global scheme is carried by the
mode-separated starting angles rather than by the nonuniform weights. On
the spin glass, whose cluster masses spread over $0.073$--$0.169$, the
equal-size constraint returns part of the advantage
($2.98$ against $2.39\times10^{-3}$) while retaining most of it relative
to the global scheme ($3.75\times10^{-3}$); since the balanced constraint
also changes the partition itself, this difference is an upper bound on the
effect of the weights alone.

\begin{table}[!ht]
\caption{\label{tab:eqweight-control}Equal-weight control: peak test
MMD$^2$ (maximum over the six evaluation bandwidths; mean and, in
parentheses, standard deviation over the three training seeds in units of
the last digits) for the global scheme, the cluster-initialized scheme,
and its equal-weight variant, which balances the cluster sizes so that
$\pi_\ell=1/L$ while keeping the moment-matched starting angles.}
\begin{ruledtabular}
\begin{tabular}{llccc}
Dataset & $a$ & global & cluster & cluster, equal weights \\
\colrule
Binary blobs & $1$ & $1.41(9)\times10^{-2}$ & $1.433(5)\times10^{-2}$ & $1.42(2)\times10^{-2}$ \\
Binary blobs & $2$ & $3.4(7)\times10^{-3}$ & $2.1(8)\times10^{-3}$ & $2.149(9)\times10^{-3}$ \\
2D Ising & $2$ & $2.6(7)\times10^{-4}$ & $1.93(1)\times10^{-4}$ & $1.883(6)\times10^{-4}$ \\
2D Ising & $3$ & $2.4(7)\times10^{-4}$ & $1.942(9)\times10^{-4}$ & $1.842(5)\times10^{-4}$ \\
MNIST & $4$ & $1.54(3)\times10^{-3}$ & $1.487(3)\times10^{-3}$ & $1.48(3)\times10^{-3}$ \\
D-Wave & $3$ & $3.8(6)\times10^{-3}$ & $2.39(10)\times10^{-3}$ & $2.98(16)\times10^{-3}$ \\
\end{tabular}
\end{ruledtabular}
\end{table}

\FloatBarrier

\subsection{\label{app:train-details}Classical RBM baseline}

The classical reference is a Bernoulli--Bernoulli restricted Boltzmann machine
(RBM) with $n$ visible units and one layer of $h$ hidden units, giving
$nh+n+h$ trainable parameters. The RBM family also admits quantum extensions
with quantified expressivity relations to its classical form~\cite{demidik2025expressive};
here it serves purely as a classical reference. We train it on the same
training data as the IQP-QCBM using persistent contrastive divergence
(PCD)~\cite{tieleman2008pcd,hinton2012practical}.
All runs use a batch of $128$ persistent chains, weights initialized from
$\mathcal N(0,0.01^2)$, zero hidden biases, visible biases initialized to the
data log-odds, weight decay $10^{-4}$, and momentum $0.5$ for the first $5\%$
of updates and $0.9$ thereafter.

Ten percent of the training set is reserved for model selection. For each
hidden width, the learning rate $\eta$ and the number $k$ of Gibbs updates per
PCD step are chosen by the mean rank of the validation MMD$^2$ over the kernel
sweep. The chosen setting is then retrained on the full training set with five
independent seeds. Table~\ref{tab:rbm-settings} gives the complete search and
identifies the curve reported in Fig.~\ref{fig:train-results}(e)--(h). For $n=16$ we
evaluate the RBM distribution exactly. For the larger datasets, each evaluation
uses $2000$ independent Gibbs chains after $20{,}000$ burn-in updates.
This selection procedure deliberately gives the RBM a reporting-metric-aligned
hyperparameter search, while the IQP-QCBM results use prespecified settings and
no validation-based selection. The comparison is therefore conservative for
the IQP-QCBM with respect to model-selection effort.

\begin{table*}[!t]
\caption{\label{tab:rbm-settings}RBM search and reported settings. The width
$h$, learning rate $\eta$, and PCD depth $k$ are the only quantities varied;
all other settings are given in the text. The last column gives the RBM curve
shown in Fig.~\ref{fig:train-results}(e)--(h), including its parameter count. Each final
value is the mean over five independently trained models.}
\renewcommand{\arraystretch}{1.25}
\begin{ruledtabular}
\begin{tabular}{l c c c c}
Target & hidden-width grid $h$ & learning-rate grid $\eta$ & updates & reported $(h,\,\text{parameters},\,\eta,\,k)$ \\
\colrule
Binary blobs & $\{32,64,128\}$ & $\{3{\times}10^{-3},10^{-3},3{\times}10^{-4},10^{-4}\}$ & $100{,}000$ & $(64,\,1104,\,3{\times}10^{-4},\,10)$ \\
2D Ising & $\{32,64,128\}$ & $\{3{\times}10^{-3},10^{-3},3{\times}10^{-4},10^{-4}\}$ & $100{,}000$ & $(64,\,1104,\,3{\times}10^{-4},\,10)$ \\
MNIST & $\{224,448,896\}$ & $\{10^{-3},3{\times}10^{-4},10^{-4},10^{-5}\}$ & $30{,}000$ & $(896,\,704{,}144,\,10^{-5},\,10)$ \\
D-Wave & $\{57,115,230\}$ & $\{10^{-3},3{\times}10^{-4},10^{-4},10^{-5}\}$ & $30{,}000$ & $(230,\,112{,}034,\,10^{-5},\,10)$ \\
\end{tabular}
\end{ruledtabular}
\end{table*}

\FloatBarrier

\subsection{\label{app:theory-check}Protocol for the numerical tests of the theorems}

This appendix records the protocol and the per-configuration values behind
Sec.~\ref{sec:exp-theory}. Each center is rebuilt through the training code
path with the coincidence-breaking radius set to zero, so the point evaluated is
the analyzed center itself.

\emph{Exact loss.} For $n=16$ the loss is evaluated in closed form. The phase
$\sum_G\theta_G(-1)^{G\cdot\bm y}$ of the diagonal part is a Walsh polynomial,
hence the Walsh--Hadamard transform of the coefficient array carrying
$\theta_G$ on index $G$; transforming $e^{-i\,\mathrm{phase}}$ once more gives
the amplitudes, $p=|\psi|^2$ is summed over the ancilla register, and a third
transform of the system marginal returns every correlator $z_A$. The loss is
then $\mathcal{L}=\langle\sum_A P_\sigma(A)(z_A-t_A)^2\rangle_\sigma$ with
$P_\sigma$ the operator distribution of Sec.~\ref{sec:mmd}, which is the
expectation of the training estimator over its operator draw and its sampling;
its gradients follow by automatic differentiation. The system marginal agrees
with a state-vector simulation of the same circuit to between
$5\times10^{-19}$ and $2\times10^{-17}$ in every configuration. For the
cluster-initialized scheme the mixture is evaluated branch by branch with the
branch weights $\pi_\ell$ used in training (the cluster masses), so the
weighted loss of Appendix~\ref{app:assemble} is the one differentiated.

\emph{Curvature.} The compiled-angle direction that moves the angle of branch
$\ell$ alone is $v_{j,S}=(1/L)(-1)^{|\ell\cap S|}$, the inverse of the
Walsh map of Eq.~\eqref{eq:ciqp-angles}; the reported curvature is the second
derivative of the exact loss along $v$, taken by nested automatic
differentiation. Table~\ref{tab:theory-check} lists the values at the witness
coordinate, chosen in each configuration as the pair $(\ell,j)$
maximizing the leading positive term of Eq.~\eqref{eq:datadep-curv-main}.
Repeating the stronger comparison $c\geq c_{\mathrm{lead}}$ at all $16$
one-body coordinates of all $L$ branches gives $928$ pairs, of which $924$
satisfy it. The four exceptions all lie at
saturated coordinates, $1-(t_j^{(\ell)})^2\le0.16$, exactly the
coordinates Assumption~\ref{as:peak} excludes from the witness role: there
the retained sensitivity operator is small and the finite-size mismatch
remainder can outweigh it. The theorem requires a witness coordinate, and the
entries of Table~\ref{tab:theory-check} satisfy the stronger comparison in
every configuration. The sharpest test of
Assumption~\ref{as:peak} is the blobs at $a=3$: with eight clusters on eight
modes each group is a single blob mode, so the factorization of
Assumption~\ref{as:factor} is at its best while every marginal approaches
saturation, $1-(t_j^{(\ell)})^2$ ranging over $0.13$--$0.28$ with its
scale set by the bit-flip noise of Eq.~\eqref{eq:blobs}
[$4\eta(1-\eta)\approx0.19$ at $\eta=0.05$]. The nonsaturation constant is
small there rather than absent, and the stronger comparison holds at all $128$ coordinates
of this configuration, with a minimum ratio of $8.5$; the witness factor
would vanish only in the noiseless limit $\eta\to0$, where each group
concentrates on a single string, the case Assumption~\ref{as:peak}
excludes. The margin between the measured curvature and
the leading term is accounted for in closed form: summing Eq.~\eqref{eq:bpr-gid} over
all $A\ni j$ gives a sensitivity block larger than the retained
$A=\{j\}$ term by $\prod_{k\neq j}[1+p\,(t_k^{(\ell)})^2/(1-p)]$ with
$p$ the operator-sampling probability of Sec.~\ref{sec:mmd}, and that factor
reproduces every measured margin of Table~\ref{tab:theory-check} to within a
few percent. The two blocks of Eq.~\eqref{eq:bp-curv} are also evaluated
separately at the witness coordinate: they sum to the measured curvature to a
relative $10^{-15}$. Their signed ratio is reported in
Table~\ref{tab:theory-check}; at the global center its magnitude doubles when
$L$ doubles, as required by the respective $\pi_{\ell^\star}$ and
$\pi_{\ell^\star}^2$ scalings.

\begin{table}[!ht]
\caption{\label{tab:theory-check}Exact witness-coordinate curvature $c$.
Here $c_{\mathrm{ref}}$ is Theorem~\ref{thm:rand-bp}'s exact prediction for
the data-agnostic center and $c_{\mathrm{lead}}$ of
Eq.~\eqref{eq:datadep-curv-main} otherwise. The last column is the signed
data-mismatch/model-sensitivity ratio.}
\begin{ruledtabular}
\scriptsize
\setlength{\tabcolsep}{2.2pt}
\renewcommand{\arraystretch}{0.84}
\begin{tabular}{l l c c c c c c c}
Target & Center & $L$ & $\pi_{\ell^\star}$ & $1-t_j^2$ & $c$ & $c_{\mathrm{ref}}$ & ratio &
$c_{\mathrm{mis}}/c_{\mathrm{sens}}$ \\
\colrule
both & data-agnostic & $1$ & $1/1$ & $1.000$ & $6.88\times10^{-2}$ & $6.88\times10^{-2}$ & $1.0000$ & $0$ \\
both & data-agnostic & $2$ & $1/2$ & $1.000$ & $1.72\times10^{-2}$ & $1.72\times10^{-2}$ & $1.0000$ & $0$ \\
both & data-agnostic & $4$ & $1/4$ & $1.000$ & $4.30\times10^{-3}$ & $4.30\times10^{-3}$ & $1.0000$ & $0$ \\
both & data-agnostic & $8$ & $1/8$ & $1.000$ & $1.07\times10^{-3}$ & $1.07\times10^{-3}$ & $1.0000$ & $0$ \\
blobs & global & $1$ & $1.000$ & $0.954$ & $1.20\times10^{-1}$ & $6.56\times10^{-2}$ & $1.84$ & $-1.49\times10^{-2}$ \\
blobs & global & $2$ & $0.500$ & $0.954$ & $2.96\times10^{-2}$ & $1.64\times10^{-2}$ & $1.81$ & $-2.98\times10^{-2}$ \\
blobs & global & $4$ & $0.250$ & $0.954$ & $7.18\times10^{-3}$ & $4.10\times10^{-3}$ & $1.75$ & $-5.96\times10^{-2}$ \\
blobs & global & $8$ & $0.125$ & $0.954$ & $1.68\times10^{-3}$ & $1.02\times10^{-3}$ & $1.64$ & $-1.19\times10^{-1}$ \\
Ising & global & $1$ & $1.000$ & $1.000$ & $6.88\times10^{-2}$ & $6.88\times10^{-2}$ & $1.00$ & $-7.53\times10^{-6}$ \\
Ising & global & $2$ & $0.500$ & $1.000$ & $1.72\times10^{-2}$ & $1.72\times10^{-2}$ & $1.00$ & $-1.51\times10^{-5}$ \\
Ising & global & $4$ & $0.250$ & $1.000$ & $4.30\times10^{-3}$ & $4.30\times10^{-3}$ & $1.00$ & $-3.01\times10^{-5}$ \\
Ising & global & $8$ & $0.125$ & $1.000$ & $1.07\times10^{-3}$ & $1.07\times10^{-3}$ & $1.00$ & $-6.02\times10^{-5}$ \\
blobs & cluster & $2$ & $0.558$ & $0.960$ & $5.60\times10^{-2}$ & $2.05\times10^{-2}$ & $2.73$ & $+1.84\times10^{-2}$ \\
blobs & cluster & $4$ & $0.383$ & $0.999$ & $3.45\times10^{-2}$ & $1.01\times10^{-2}$ & $3.42$ & $-3.75\times10^{-3}$ \\
blobs & cluster & $8$ & $0.130$ & $0.271$ & $3.47\times10^{-3}$ & $3.15\times10^{-4}$ & $10.99$ & $-9.39\times10^{-2}$ \\
Ising & cluster & $2$ & $0.499$ & $0.730$ & $3.93\times10^{-2}$ & $1.25\times10^{-2}$ & $3.14$ & $+1.43\times10^{-1}$ \\
Ising & cluster & $4$ & $0.324$ & $0.741$ & $3.66\times10^{-2}$ & $5.36\times10^{-3}$ & $6.83$ & $+1.86\times10^{-1}$ \\
Ising & cluster & $8$ & $0.115$ & $0.997$ & $2.00\times10^{-3}$ & $9.03\times10^{-4}$ & $2.21$ & $-2.85\times10^{-2}$ \\
\end{tabular}
\end{ruledtabular}
\end{table}

\emph{Gradient suppression.} The sweep of Fig.~\ref{fig:theory-check}(c) uses
$\bm\theta(c)=\bm\theta^{\mathrm{center}}+c\,\bm\eta$ with $\bm\eta$ the run's
own coincidence-breaking draw, so $c=1$ is the initialization the runs use;
Fig.~\ref{fig:app-branch-spread} repeats the sweep at every ancilla count. The
estimated values use the training settings of Appendix~\ref{app:hyperparams}
with twelve estimator keys per point and the three training seeds. At
$\delta=0$ the exact ancilla-block gradient is zero to machine precision while
the estimator returns $2.0\times10^{-3}$ (global) and $5.3\times10^{-3}$
(data-agnostic) on the blobs at $a=1$, and $1.5\times10^{-4}$ and
$1.8\times10^{-4}$ on the Ising; the spreads at which the exact signal reaches
those floors are $\delta^\star=1.5\times10^{-2}$ and $4.3\times10^{-2}$ on the
blobs against $\approx1\times10^{-3}$ on the Ising, to be compared with the
$\delta=9.4\times10^{-2}$ of the runs.

\begin{figure*}[!ht]
\includegraphics[width=\textwidth]{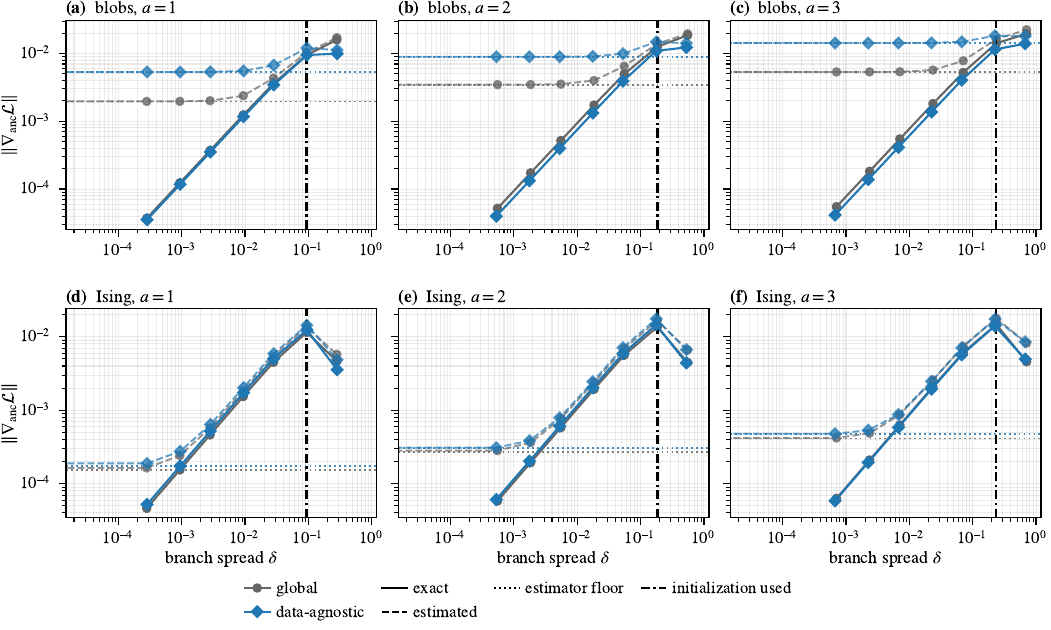}
\caption{\label{fig:app-branch-spread} Branch-spread sweep of
Fig.~\ref{fig:theory-check}(c) at every ancilla count: ancilla-coupling
gradient norm at a branch-coincident start against the branch spread $\delta$,
exactly (solid) and as the training estimator returns it (dashed), with the
estimator's $\delta=0$ floor (dotted) and the spread of the initialization the
coincident runs use (dash-dotted vertical line). The exact gradient vanishes at $\delta=0$ and is
proportional to $\delta$ in every panel. On the blobs the floor rises with the
ancilla count until, at $a=3$, the data-agnostic initialization lies below it;
on the Ising the floor stays one to two orders of magnitude under the signal.}
\end{figure*}

\FloatBarrier
\subsection{\label{app:zero2b}Ablation of the trained two-body angles}

To isolate the trained non-product structure, we set every two-body angle
of the headline cluster-initialized MNIST ($a=4$) and D-Wave ($a=3$)
models to zero while retaining their one-body angles and branch weights,
then repeat the test-MMD$^2$ evaluation with the same estimator seeds and
bandwidths. Each ablated branch is therefore a product distribution.

\begin{table}[!ht]
\caption{\label{tab:zero2b}Effect of zeroing all trained two-body angles in
the headline cluster-initialized models. Sweep-averaged test MMD$^2$ is
reported as mean $\pm$ one standard deviation over three seeds; ``Factor''
is the seed-wise increase.}
\begin{ruledtabular}
\begin{tabular}{lcccc}
Dataset & $a$ & Trained & Two-body zeroed & Factor \\
\colrule
MNIST  & $4$ & $(1.02\pm0.02)\times10^{-3}$
  & $(5.49\pm0.13)\times10^{-3}$ & $5.3$--$5.5$ \\
D-Wave & $3$ & $(2.03\pm0.09)\times10^{-3}$
  & $(14.7\pm0.8)\times10^{-3}$  & $7.0$--$7.4$ \\
\end{tabular}
\end{ruledtabular}
\end{table}

The ablation increases test MMD$^2$ by factors of $5.3$--$5.5$ on MNIST
and $7.0$--$7.4$ on D-Wave ($4.7$--$8.1$ across individual bandwidths),
confirming that the fitted branches use two-body structure. The blobs
require no corresponding ablation because their target is a mixture of
product branches by construction [Eq.~\eqref{eq:blobs}].

\subsection{\label{app:numerical-bandwidth}MMD bandwidth dependence}

Figure~\ref{fig:app-sigma-sweeps} resolves the sweep averages of
Fig.~\ref{fig:train-results}(e)--(h) by kernel bandwidth at each
benchmark's headline mixture size.

\begin{figure}[!ht]
\centering
\includegraphics[width=0.86\textwidth]{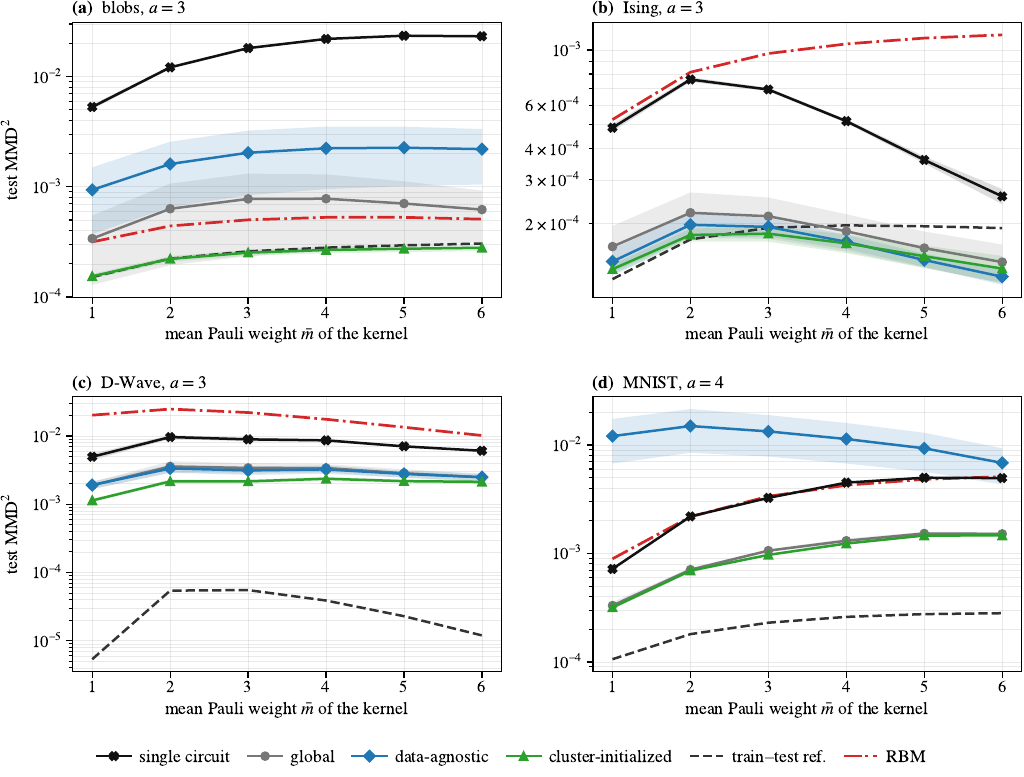}
\caption{\label{fig:app-sigma-sweeps}Kernel dependence of the test MMD$^2$
in Fig.~\ref{fig:train-results}(e)--(h), against the mean Pauli weight
$\bar m$. Headline sizes are $a=3$ (blobs, Ising, and D-Wave) and $a=4$
(MNIST). Solid IQP curves show the seed mean and one standard deviation
(ten seeds for $n=16$, five otherwise); dashed and dash-dotted curves denote
the train--test sampling reference and selected RBM, respectively.}
\end{figure}

\FloatBarrier

\end{document}

%% file: experiment_setup.tex
\subsection{\label{sec:setup}Benchmarks and training protocol}

The four benchmarks span $n=16$ to $784$ and range from well-separated
clusters to overlapping or long-range structure (Table~\ref{tab:datasets}).
Binary blobs and the 2D Ising model ($n=16$) permit exact evaluation, whereas
MNIST ($n=784$) and the D-Wave spin glass ($n=484$) test larger models. Each
uses an order-2 interaction graph, fully connected except for the
row/column graph of MNIST. For each $a>0$, the three schemes use the same
mixed-IQP architecture and number of trainable angles; they differ only in
the initial branch parameters and fixed mixture weights. For $a=0$, all three
reduce to the same ancilla-free IQP-QCBM.
Appendix~\ref{app:datasets} gives the dataset constructions,
and Appendix~\ref{app:branch-weights} reports the fixed weights and the
equal-weight control.

\begin{table}[t]
\caption{\label{tab:datasets} Benchmark datasets. The $16$- and $484$-qubit
datasets follow Ref.~\cite{recio2025}; MNIST uses $28\times28$
images~\cite{lecun1998}.}
\begin{ruledtabular}
\begin{tabular}{lcccl}
Dataset & $n$ & $N_{\text{train}}$ & $N_{\text{test}}$ & Structure \\
\colrule
Binary blobs        & $16$  & $5000$  & $10000$ & $8$ Hamming clusters \\
2D Ising            & $16$  & $5000$  & $50000$ & thermal, bit-flip symmetric \\
MNIST               & $784$ & $2000$ & $2000$ & binarized images \\
D-Wave              & $484$ & $2000$ & $2000$ & long-range spin glass \\
\end{tabular}
\end{ruledtabular}
\end{table}

\label{sec:dataset}\label{sec:train}%
We train with Adam~\cite{kingma2015adam} for $5000$ steps with learning rate
$\eta=10^{-3}$ on the mean of two MMD$^2$ losses ($\bar m=2,6$), evaluated
classically with our extension of \texttt{IQPopt}~\cite{iqpopt}. Results use
ten training seeds for the $n=16$ datasets and five for MNIST and D-Wave;
bands show one standard deviation. Data-agnostic MNIST alone uses
$20{,}000$ steps at $\eta=5\times10^{-4}$ because the shared rate leaves
several seeds in high-loss plateaus; the smaller rate removes this failure for
$a\le3$ but converges more slowly.
Coincident starts receive a small
symmetry-breaking perturbation, while cluster assignments and weights remain
fixed. Appendix~\ref{app:hyperparams} gives the full training settings and
Appendix~\ref{app:branch-weights} the clustering details.

\label{sec:eval}%
We report test MMD$^2$ against the test set, averaged over
$\bar m=1,\dots,6$. For reference, we also compute the MMD$^2$ directly
between the independent training and test sets. Because both sets are sampled
from the target distribution, this residual measures their finite-sample
difference. We call it the train--test sampling reference; it is an empirical
sampling discrepancy, not a lower bound on the model's test MMD$^2$.
Evaluation is exact for $n=16$ and uses the correlator estimator for the
larger datasets; values are compared only within each dataset.
Appendix~\ref{app:hyperparams} gives the estimator budgets.

\label{sec:baseline}%
Restricted Boltzmann machines trained over five seeds provide an external
reference. Their hyperparameters are selected on a validation split; because
their training and selection protocol differs from that of the IQP-QCBM, this
is not a controlled head-to-head comparison
(Appendix~\ref{app:train-details}).

%% file: main.bbl
\providecommand{\noopsort}[1]{}\providecommand{\singleletter}[1]{#1}%
\begin{thebibliography}{46}%
\makeatletter
\providecommand \@ifxundefined [1]{%
 \@ifx{#1\undefined}
}%
\providecommand \@ifnum [1]{%
 \ifnum #1\expandafter \@firstoftwo
 \else \expandafter \@secondoftwo
 \fi
}%
\providecommand \@ifx [1]{%
 \ifx #1\expandafter \@firstoftwo
 \else \expandafter \@secondoftwo
 \fi
}%
\providecommand \natexlab [1]{#1}%
\providecommand \enquote  [1]{``#1''}%
\providecommand \bibnamefont  [1]{#1}%
\providecommand \bibfnamefont [1]{#1}%
\providecommand \citenamefont [1]{#1}%
\providecommand \href@noop [0]{\@secondoftwo}%
\providecommand \href [0]{\begingroup \@sanitize@url \@href}%
\providecommand \@href[1]{\@@startlink{#1}\@@href}%
\providecommand \@@href[1]{\endgroup#1\@@endlink}%
\providecommand \@sanitize@url [0]{\catcode `\\12\catcode `\$12\catcode
  `\&12\catcode `\#12\catcode `\^12\catcode `\_12\catcode `\%12\relax}%
\providecommand \@@startlink[1]{}%
\providecommand \@@endlink[0]{}%
\providecommand \url  [0]{\begingroup\@sanitize@url \@url }%
\providecommand \@url [1]{\endgroup\@href {#1}{\urlprefix }}%
\providecommand \urlprefix  [0]{URL }%
\providecommand \Eprint [0]{\href }%
\providecommand \doibase [0]{https://doi.org/}%
\providecommand \selectlanguage [0]{\@gobble}%
\providecommand \bibinfo  [0]{\@secondoftwo}%
\providecommand \bibfield  [0]{\@secondoftwo}%
\providecommand \translation [1]{[#1]}%
\providecommand \BibitemOpen [0]{}%
\providecommand \bibitemStop [0]{}%
\providecommand \bibitemNoStop [0]{.\EOS\space}%
\providecommand \EOS [0]{\spacefactor3000\relax}%
\providecommand \BibitemShut  [1]{\csname bibitem#1\endcsname}%
\let\auto@bib@innerbib\@empty
\bibitem [{\citenamefont {Benedetti}\ \emph {et~al.}(2019)\citenamefont
  {Benedetti}, \citenamefont {Garcia-Pintos}, \citenamefont {Perdomo},
  \citenamefont {Leyton-Ortega}, \citenamefont {Nam},\ and\ \citenamefont
  {Perdomo-Ortiz}}]{benedetti2019}%
  \BibitemOpen
  \bibfield  {author} {\bibinfo {author} {\bibfnamefont {M.}~\bibnamefont
  {Benedetti}}, \bibinfo {author} {\bibfnamefont {D.}~\bibnamefont
  {Garcia-Pintos}}, \bibinfo {author} {\bibfnamefont {O.}~\bibnamefont
  {Perdomo}}, \bibinfo {author} {\bibfnamefont {V.}~\bibnamefont
  {Leyton-Ortega}}, \bibinfo {author} {\bibfnamefont {Y.}~\bibnamefont {Nam}},\
  and\ \bibinfo {author} {\bibfnamefont {A.}~\bibnamefont {Perdomo-Ortiz}},\
  }\bibfield  {title} {\bibinfo {title} {A generative modeling approach for
  benchmarking and training shallow quantum circuits},\ }\href
  {https://doi.org/10.1038/s41534-019-0157-8} {\bibfield  {journal} {\bibinfo
  {journal} {npj Quantum Information}\ }\textbf {\bibinfo {volume} {5}},\
  \bibinfo {pages} {45} (\bibinfo {year} {2019})}\BibitemShut {NoStop}%
\bibitem [{\citenamefont {Liu}\ and\ \citenamefont {Wang}(2018)}]{liu2018}%
  \BibitemOpen
  \bibfield  {author} {\bibinfo {author} {\bibfnamefont {J.-G.}\ \bibnamefont
  {Liu}}\ and\ \bibinfo {author} {\bibfnamefont {L.}~\bibnamefont {Wang}},\
  }\bibfield  {title} {\bibinfo {title} {Differentiable learning of quantum
  circuit born machines},\ }\href {https://doi.org/10.1103/PhysRevA.98.062324}
  {\bibfield  {journal} {\bibinfo  {journal} {Phys. Rev. A}\ }\textbf {\bibinfo
  {volume} {98}},\ \bibinfo {pages} {062324} (\bibinfo {year}
  {2018})}\BibitemShut {NoStop}%
\bibitem [{\citenamefont {Coyle}\ \emph {et~al.}(2020)\citenamefont {Coyle},
  \citenamefont {Mills}, \citenamefont {Danos},\ and\ \citenamefont
  {Kashefi}}]{coyle2020}%
  \BibitemOpen
  \bibfield  {author} {\bibinfo {author} {\bibfnamefont {B.}~\bibnamefont
  {Coyle}}, \bibinfo {author} {\bibfnamefont {D.}~\bibnamefont {Mills}},
  \bibinfo {author} {\bibfnamefont {V.}~\bibnamefont {Danos}},\ and\ \bibinfo
  {author} {\bibfnamefont {E.}~\bibnamefont {Kashefi}},\ }\bibfield  {title}
  {\bibinfo {title} {The {Born} supremacy: quantum advantage and training of an
  {Ising} {Born} machine},\ }\href {https://doi.org/10.1038/s41534-020-00288-9}
  {\bibfield  {journal} {\bibinfo  {journal} {npj Quantum Information}\
  }\textbf {\bibinfo {volume} {6}},\ \bibinfo {pages} {60} (\bibinfo {year}
  {2020})}\BibitemShut {NoStop}%
\bibitem [{\citenamefont {Bremner}\ \emph {et~al.}(2011)\citenamefont
  {Bremner}, \citenamefont {Jozsa},\ and\ \citenamefont
  {Shepherd}}]{bremner2011}%
  \BibitemOpen
  \bibfield  {author} {\bibinfo {author} {\bibfnamefont {M.~J.}\ \bibnamefont
  {Bremner}}, \bibinfo {author} {\bibfnamefont {R.}~\bibnamefont {Jozsa}},\
  and\ \bibinfo {author} {\bibfnamefont {D.~J.}\ \bibnamefont {Shepherd}},\
  }\bibfield  {title} {\bibinfo {title} {Classical simulation of commuting
  quantum computations implies collapse of the polynomial hierarchy},\ }\href
  {https://doi.org/10.1098/rspa.2010.0301} {\bibfield  {journal} {\bibinfo
  {journal} {Proceedings of the Royal Society A: Mathematical, Physical and
  Engineering Sciences}\ }\textbf {\bibinfo {volume} {467}},\ \bibinfo {pages}
  {459} (\bibinfo {year} {2011})},\ \Eprint {https://arxiv.org/abs/1005.1407}
  {arXiv:1005.1407 [quant-ph]} \BibitemShut {NoStop}%
\bibitem [{\citenamefont {Du}\ \emph {et~al.}(2022)\citenamefont {Du},
  \citenamefont {Tu}, \citenamefont {Wu}, \citenamefont {Yuan},\ and\
  \citenamefont {Tao}}]{du2022power}%
  \BibitemOpen
  \bibfield  {author} {\bibinfo {author} {\bibfnamefont {Y.}~\bibnamefont
  {Du}}, \bibinfo {author} {\bibfnamefont {Z.}~\bibnamefont {Tu}}, \bibinfo
  {author} {\bibfnamefont {B.}~\bibnamefont {Wu}}, \bibinfo {author}
  {\bibfnamefont {X.}~\bibnamefont {Yuan}},\ and\ \bibinfo {author}
  {\bibfnamefont {D.}~\bibnamefont {Tao}},\ }\href
  {https://arxiv.org/abs/2205.04730} {\bibinfo {title} {Power of quantum
  generative learning}} (\bibinfo {year} {2022}),\ \Eprint
  {https://arxiv.org/abs/2205.04730} {arXiv:2205.04730 [quant-ph]} \BibitemShut
  {NoStop}%
\bibitem [{\citenamefont {McClean}\ \emph {et~al.}(2018)\citenamefont
  {McClean}, \citenamefont {Boixo}, \citenamefont {Smelyanskiy}, \citenamefont
  {Babbush},\ and\ \citenamefont {Neven}}]{mcclean2018}%
  \BibitemOpen
  \bibfield  {author} {\bibinfo {author} {\bibfnamefont {J.~R.}\ \bibnamefont
  {McClean}}, \bibinfo {author} {\bibfnamefont {S.}~\bibnamefont {Boixo}},
  \bibinfo {author} {\bibfnamefont {V.~N.}\ \bibnamefont {Smelyanskiy}},
  \bibinfo {author} {\bibfnamefont {R.}~\bibnamefont {Babbush}},\ and\ \bibinfo
  {author} {\bibfnamefont {H.}~\bibnamefont {Neven}},\ }\bibfield  {title}
  {\bibinfo {title} {Barren plateaus in quantum neural network training
  landscapes},\ }\href {https://doi.org/10.1038/s41467-018-07090-4} {\bibfield
  {journal} {\bibinfo  {journal} {Nature Communications}\ }\textbf {\bibinfo
  {volume} {9}},\ \bibinfo {pages} {4812} (\bibinfo {year} {2018})}\BibitemShut
  {NoStop}%
\bibitem [{\citenamefont {Cerezo}\ \emph {et~al.}(2021)\citenamefont {Cerezo},
  \citenamefont {Sone}, \citenamefont {Volkoff}, \citenamefont {Cincio},\ and\
  \citenamefont {Coles}}]{cerezo2021}%
  \BibitemOpen
  \bibfield  {author} {\bibinfo {author} {\bibfnamefont {M.}~\bibnamefont
  {Cerezo}}, \bibinfo {author} {\bibfnamefont {A.}~\bibnamefont {Sone}},
  \bibinfo {author} {\bibfnamefont {T.}~\bibnamefont {Volkoff}}, \bibinfo
  {author} {\bibfnamefont {L.}~\bibnamefont {Cincio}},\ and\ \bibinfo {author}
  {\bibfnamefont {P.~J.}\ \bibnamefont {Coles}},\ }\bibfield  {title} {\bibinfo
  {title} {Cost function dependent barren plateaus in shallow parametrized
  quantum circuits},\ }\href {https://doi.org/10.1038/s41467-021-21728-w}
  {\bibfield  {journal} {\bibinfo  {journal} {Nature Communications}\ }\textbf
  {\bibinfo {volume} {12}},\ \bibinfo {pages} {1791} (\bibinfo {year}
  {2021})}\BibitemShut {NoStop}%
\bibitem [{\citenamefont {Van~den Nest}(2011)}]{vandennest2011}%
  \BibitemOpen
  \bibfield  {author} {\bibinfo {author} {\bibfnamefont {M.}~\bibnamefont
  {Van~den Nest}},\ }\bibfield  {title} {\bibinfo {title} {Simulating quantum
  computers with probabilistic methods},\ }\href
  {https://doi.org/10.26421/QIC11.9-10-5} {\bibfield  {journal} {\bibinfo
  {journal} {Quantum Information and Computation}\ }\textbf {\bibinfo {volume}
  {11}},\ \bibinfo {pages} {784} (\bibinfo {year} {2011})}\BibitemShut
  {NoStop}%
\bibitem [{\citenamefont {Recio-Armengol}\ \emph {et~al.}(2026)\citenamefont
  {Recio-Armengol}, \citenamefont {Ahmed},\ and\ \citenamefont
  {Bowles}}]{recio2025}%
  \BibitemOpen
  \bibfield  {author} {\bibinfo {author} {\bibfnamefont {E.}~\bibnamefont
  {Recio-Armengol}}, \bibinfo {author} {\bibfnamefont {S.}~\bibnamefont
  {Ahmed}},\ and\ \bibinfo {author} {\bibfnamefont {J.}~\bibnamefont
  {Bowles}},\ }\href {https://arxiv.org/abs/2503.02934} {\bibinfo {title}
  {Train on classical, deploy on quantum: scaling generative quantum machine
  learning to a thousand qubits}} (\bibinfo {year} {2026}),\ \Eprint
  {https://arxiv.org/abs/2503.02934} {arXiv:2503.02934 [quant-ph]} \BibitemShut
  {NoStop}%
\bibitem [{\citenamefont {Bremner}\ \emph {et~al.}(2016)\citenamefont
  {Bremner}, \citenamefont {Montanaro},\ and\ \citenamefont
  {Shepherd}}]{bremner2016}%
  \BibitemOpen
  \bibfield  {author} {\bibinfo {author} {\bibfnamefont {M.~J.}\ \bibnamefont
  {Bremner}}, \bibinfo {author} {\bibfnamefont {A.}~\bibnamefont {Montanaro}},\
  and\ \bibinfo {author} {\bibfnamefont {D.~J.}\ \bibnamefont {Shepherd}},\
  }\bibfield  {title} {\bibinfo {title} {Average-case complexity versus
  approximate simulation of commuting quantum computations},\ }\href
  {https://doi.org/10.1103/PhysRevLett.117.080501} {\bibfield  {journal}
  {\bibinfo  {journal} {Phys. Rev. Lett.}\ }\textbf {\bibinfo {volume} {117}},\
  \bibinfo {pages} {080501} (\bibinfo {year} {2016})}\BibitemShut {NoStop}%
\bibitem [{\citenamefont {Bremner}\ \emph {et~al.}(2025)\citenamefont
  {Bremner}, \citenamefont {Cheng},\ and\ \citenamefont
  {Ji}}]{bremner2025stabilizer}%
  \BibitemOpen
  \bibfield  {author} {\bibinfo {author} {\bibfnamefont {M.~J.}\ \bibnamefont
  {Bremner}}, \bibinfo {author} {\bibfnamefont {B.}~\bibnamefont {Cheng}},\
  and\ \bibinfo {author} {\bibfnamefont {Z.}~\bibnamefont {Ji}},\ }\bibfield
  {title} {\bibinfo {title} {Instantaneous quantum polynomial-time sampling and
  verifiable quantum advantage: Stabilizer scheme and classical security},\
  }\href {https://doi.org/10.1103/PRXQuantum.6.020315} {\bibfield  {journal}
  {\bibinfo  {journal} {PRX Quantum}\ }\textbf {\bibinfo {volume} {6}},\
  \bibinfo {pages} {020315} (\bibinfo {year} {2025})}\BibitemShut {NoStop}%
\bibitem [{\citenamefont {Rudolph}\ \emph {et~al.}(2024)\citenamefont
  {Rudolph}, \citenamefont {Lerch}, \citenamefont {Thanasilp}, \citenamefont
  {Kiss}, \citenamefont {Shaya}, \citenamefont {Vallecorsa}, \citenamefont
  {Grossi},\ and\ \citenamefont {Holmes}}]{rudolph2024barriers}%
  \BibitemOpen
  \bibfield  {author} {\bibinfo {author} {\bibfnamefont {M.~S.}\ \bibnamefont
  {Rudolph}}, \bibinfo {author} {\bibfnamefont {S.}~\bibnamefont {Lerch}},
  \bibinfo {author} {\bibfnamefont {S.}~\bibnamefont {Thanasilp}}, \bibinfo
  {author} {\bibfnamefont {O.}~\bibnamefont {Kiss}}, \bibinfo {author}
  {\bibfnamefont {O.}~\bibnamefont {Shaya}}, \bibinfo {author} {\bibfnamefont
  {S.}~\bibnamefont {Vallecorsa}}, \bibinfo {author} {\bibfnamefont
  {M.}~\bibnamefont {Grossi}},\ and\ \bibinfo {author} {\bibfnamefont
  {Z.}~\bibnamefont {Holmes}},\ }\bibfield  {title} {\bibinfo {title}
  {Trainability barriers and opportunities in quantum generative modeling},\
  }\href {https://doi.org/10.1038/s41534-024-00902-0} {\bibfield  {journal}
  {\bibinfo  {journal} {npj Quantum Information}\ }\textbf {\bibinfo {volume}
  {10}},\ \bibinfo {pages} {116} (\bibinfo {year} {2024})}\BibitemShut
  {NoStop}%
\bibitem [{\citenamefont {Lerch}\ \emph {et~al.}(2026)\citenamefont {Lerch},
  \citenamefont {Bowles}, \citenamefont {Puig}, \citenamefont {Armengol},
  \citenamefont {Holmes},\ and\ \citenamefont {Thanasilp}}]{lerch2026}%
  \BibitemOpen
  \bibfield  {author} {\bibinfo {author} {\bibfnamefont {S.}~\bibnamefont
  {Lerch}}, \bibinfo {author} {\bibfnamefont {J.}~\bibnamefont {Bowles}},
  \bibinfo {author} {\bibfnamefont {R.}~\bibnamefont {Puig}}, \bibinfo {author}
  {\bibfnamefont {E.}~\bibnamefont {Armengol}}, \bibinfo {author}
  {\bibfnamefont {Z.}~\bibnamefont {Holmes}},\ and\ \bibinfo {author}
  {\bibfnamefont {S.}~\bibnamefont {Thanasilp}},\ }\href
  {https://arxiv.org/abs/2603.14576} {\bibinfo {title} {Iqp born machines under
  data-dependent and agnostic initialization strategies}} (\bibinfo {year}
  {2026}),\ \Eprint {https://arxiv.org/abs/2603.14576} {arXiv:2603.14576
  [quant-ph]} \BibitemShut {NoStop}%
\bibitem [{\citenamefont {Shen}\ \emph {et~al.}(2026)\citenamefont {Shen},
  \citenamefont {Pielawa}, \citenamefont {Dunjko},\ and\ \citenamefont
  {Wang}}]{characterizing_iqp}%
  \BibitemOpen
  \bibfield  {author} {\bibinfo {author} {\bibfnamefont {K.}~\bibnamefont
  {Shen}}, \bibinfo {author} {\bibfnamefont {S.}~\bibnamefont {Pielawa}},
  \bibinfo {author} {\bibfnamefont {V.}~\bibnamefont {Dunjko}},\ and\ \bibinfo
  {author} {\bibfnamefont {H.}~\bibnamefont {Wang}},\ }\href
  {https://arxiv.org/abs/2602.11042} {\bibinfo {title} {Characterizing
  trainability of instantaneous quantum polynomial circuit born machines}}
  (\bibinfo {year} {2026}),\ \Eprint {https://arxiv.org/abs/2602.11042}
  {arXiv:2602.11042 [quant-ph]} \BibitemShut {NoStop}%
\bibitem [{\citenamefont {De~Luca}(2026)}]{deluca2026}%
  \BibitemOpen
  \bibfield  {author} {\bibinfo {author} {\bibfnamefont {G.}~\bibnamefont
  {De~Luca}},\ }\href {https://arxiv.org/abs/2606.10179} {\bibinfo {title}
  {Trainability of {IQP} quantum circuit born machines under gaussian
  initialization}} (\bibinfo {year} {2026}),\ \Eprint
  {https://arxiv.org/abs/2606.10179} {arXiv:2606.10179 [quant-ph]} \BibitemShut
  {NoStop}%
\bibitem [{\citenamefont {Kurkin}\ \emph
  {et~al.}(2025{\natexlab{a}})\citenamefont {Kurkin}, \citenamefont {Shen},
  \citenamefont {Pielawa}, \citenamefont {Wang},\ and\ \citenamefont
  {Dunjko}}]{kurkin2025universality}%
  \BibitemOpen
  \bibfield  {author} {\bibinfo {author} {\bibfnamefont {A.}~\bibnamefont
  {Kurkin}}, \bibinfo {author} {\bibfnamefont {K.}~\bibnamefont {Shen}},
  \bibinfo {author} {\bibfnamefont {S.}~\bibnamefont {Pielawa}}, \bibinfo
  {author} {\bibfnamefont {H.}~\bibnamefont {Wang}},\ and\ \bibinfo {author}
  {\bibfnamefont {V.}~\bibnamefont {Dunjko}},\ }\href
  {https://arxiv.org/abs/2510.08476} {\bibinfo {title} {Universality and
  kernel-adaptive training for classically trained, quantum-deployed generative
  models}} (\bibinfo {year} {2025}{\natexlab{a}}),\ \Eprint
  {https://arxiv.org/abs/2510.08476} {arXiv:2510.08476 [quant-ph]} \BibitemShut
  {NoStop}%
\bibitem [{\citenamefont {Kurkin}\ \emph
  {et~al.}(2025{\natexlab{b}})\citenamefont {Kurkin}, \citenamefont {Shen},
  \citenamefont {Pielawa}, \citenamefont {Wang},\ and\ \citenamefont
  {Dunjko}}]{kurkin2025note}%
  \BibitemOpen
  \bibfield  {author} {\bibinfo {author} {\bibfnamefont {A.}~\bibnamefont
  {Kurkin}}, \bibinfo {author} {\bibfnamefont {K.}~\bibnamefont {Shen}},
  \bibinfo {author} {\bibfnamefont {S.}~\bibnamefont {Pielawa}}, \bibinfo
  {author} {\bibfnamefont {H.}~\bibnamefont {Wang}},\ and\ \bibinfo {author}
  {\bibfnamefont {V.}~\bibnamefont {Dunjko}},\ }\href
  {https://arxiv.org/abs/2504.05997} {\bibinfo {title} {Note on the
  universality of parameterized iqp circuits with hidden units for generating
  probability distributions}} (\bibinfo {year} {2025}{\natexlab{b}}),\ \Eprint
  {https://arxiv.org/abs/2504.05997} {arXiv:2504.05997 [quant-ph]} \BibitemShut
  {NoStop}%
\bibitem [{\citenamefont {Zhong}\ \emph {et~al.}(2024)\citenamefont {Zhong},
  \citenamefont {Gao}, \citenamefont {Yelin},\ and\ \citenamefont
  {Najafi}}]{zhong2024mbl}%
  \BibitemOpen
  \bibfield  {author} {\bibinfo {author} {\bibfnamefont {W.}~\bibnamefont
  {Zhong}}, \bibinfo {author} {\bibfnamefont {X.}~\bibnamefont {Gao}}, \bibinfo
  {author} {\bibfnamefont {S.~F.}\ \bibnamefont {Yelin}},\ and\ \bibinfo
  {author} {\bibfnamefont {K.}~\bibnamefont {Najafi}},\ }\bibfield  {title}
  {\bibinfo {title} {Many-body localized hidden generative models},\ }\href
  {https://doi.org/10.1103/PhysRevResearch.6.043041} {\bibfield  {journal}
  {\bibinfo  {journal} {Phys. Rev. Research}\ }\textbf {\bibinfo {volume}
  {6}},\ \bibinfo {pages} {043041} (\bibinfo {year} {2024})}\BibitemShut
  {NoStop}%
\bibitem [{\citenamefont {Wiebe}\ and\ \citenamefont
  {Wossnig}(2019)}]{wiebe2019hidden}%
  \BibitemOpen
  \bibfield  {author} {\bibinfo {author} {\bibfnamefont {N.}~\bibnamefont
  {Wiebe}}\ and\ \bibinfo {author} {\bibfnamefont {L.}~\bibnamefont
  {Wossnig}},\ }\href {https://arxiv.org/abs/1905.09902} {\bibinfo {title}
  {Generative training of quantum boltzmann machines with hidden units}}
  (\bibinfo {year} {2019}),\ \Eprint {https://arxiv.org/abs/1905.09902}
  {arXiv:1905.09902 [quant-ph]} \BibitemShut {NoStop}%
\bibitem [{\citenamefont {Slim}\ \emph {et~al.}(2026)\citenamefont {Slim},
  \citenamefont {Monaco}, \citenamefont {Rehm}, \citenamefont {Krücker},\ and\
  \citenamefont {Borras}}]{slim2026}%
  \BibitemOpen
  \bibfield  {author} {\bibinfo {author} {\bibfnamefont {J.}~\bibnamefont
  {Slim}}, \bibinfo {author} {\bibfnamefont {S.}~\bibnamefont {Monaco}},
  \bibinfo {author} {\bibfnamefont {F.}~\bibnamefont {Rehm}}, \bibinfo {author}
  {\bibfnamefont {D.}~\bibnamefont {Krücker}},\ and\ \bibinfo {author}
  {\bibfnamefont {K.}~\bibnamefont {Borras}},\ }\href
  {https://arxiv.org/abs/2605.27735} {\bibinfo {title} {An iqp born machine for
  calorimeter image generation at 64 qubits with compiled-iqp deployment}}
  (\bibinfo {year} {2026}),\ \Eprint {https://arxiv.org/abs/2605.27735}
  {arXiv:2605.27735 [quant-ph]} \BibitemShut {NoStop}%
\bibitem [{\citenamefont {Cheng}\ \emph {et~al.}(2018)\citenamefont {Cheng},
  \citenamefont {Chen},\ and\ \citenamefont {Wang}}]{cheng2018}%
  \BibitemOpen
  \bibfield  {author} {\bibinfo {author} {\bibfnamefont {S.}~\bibnamefont
  {Cheng}}, \bibinfo {author} {\bibfnamefont {J.}~\bibnamefont {Chen}},\ and\
  \bibinfo {author} {\bibfnamefont {L.}~\bibnamefont {Wang}},\ }\bibfield
  {title} {\bibinfo {title} {Information perspective to probabilistic modeling:
  Boltzmann machines versus born machines},\ }\href
  {https://doi.org/10.3390/e20080583} {\bibfield  {journal} {\bibinfo
  {journal} {Entropy}\ }\textbf {\bibinfo {volume} {20}},\ \bibinfo {pages}
  {583} (\bibinfo {year} {2018})}\BibitemShut {NoStop}%
\bibitem [{\citenamefont {Han}\ \emph {et~al.}(2018)\citenamefont {Han},
  \citenamefont {Wang}, \citenamefont {Fan}, \citenamefont {Wang},\ and\
  \citenamefont {Zhang}}]{han2018}%
  \BibitemOpen
  \bibfield  {author} {\bibinfo {author} {\bibfnamefont {Z.-Y.}\ \bibnamefont
  {Han}}, \bibinfo {author} {\bibfnamefont {J.}~\bibnamefont {Wang}}, \bibinfo
  {author} {\bibfnamefont {H.}~\bibnamefont {Fan}}, \bibinfo {author}
  {\bibfnamefont {L.}~\bibnamefont {Wang}},\ and\ \bibinfo {author}
  {\bibfnamefont {P.}~\bibnamefont {Zhang}},\ }\bibfield  {title} {\bibinfo
  {title} {Unsupervised generative modeling using matrix product states},\
  }\bibfield  {journal} {\bibinfo  {journal} {Physical Review X}\ }\textbf
  {\bibinfo {volume} {8}},\ \href {https://doi.org/10.1103/physrevx.8.031012}
  {10.1103/physrevx.8.031012} (\bibinfo {year} {2018})\BibitemShut {NoStop}%
\bibitem [{\citenamefont {Shepherd}\ and\ \citenamefont
  {Bremner}(2009)}]{shepherd2009}%
  \BibitemOpen
  \bibfield  {author} {\bibinfo {author} {\bibfnamefont {D.}~\bibnamefont
  {Shepherd}}\ and\ \bibinfo {author} {\bibfnamefont {M.~J.}\ \bibnamefont
  {Bremner}},\ }\bibfield  {title} {\bibinfo {title} {Temporally unstructured
  quantum computation},\ }\href {https://doi.org/10.1098/rspa.2008.0443}
  {\bibfield  {journal} {\bibinfo  {journal} {Proceedings of the Royal Society
  A: Mathematical, Physical and Engineering Sciences}\ }\textbf {\bibinfo
  {volume} {465}},\ \bibinfo {pages} {1413–1439} (\bibinfo {year}
  {2009})}\BibitemShut {NoStop}%
\bibitem [{\citenamefont {Gretton}\ \emph {et~al.}(2012)\citenamefont
  {Gretton}, \citenamefont {Borgwardt}, \citenamefont {Rasch}, \citenamefont
  {Sch{{\"o}}lkopf},\ and\ \citenamefont {Smola}}]{gretton2012}%
  \BibitemOpen
  \bibfield  {author} {\bibinfo {author} {\bibfnamefont {A.}~\bibnamefont
  {Gretton}}, \bibinfo {author} {\bibfnamefont {K.~M.}\ \bibnamefont
  {Borgwardt}}, \bibinfo {author} {\bibfnamefont {M.~J.}\ \bibnamefont
  {Rasch}}, \bibinfo {author} {\bibfnamefont {B.}~\bibnamefont
  {Sch{{\"o}}lkopf}},\ and\ \bibinfo {author} {\bibfnamefont {A.}~\bibnamefont
  {Smola}},\ }\bibfield  {title} {\bibinfo {title} {A kernel two-sample test},\
  }\href {http://jmlr.org/papers/v13/gretton12a.html} {\bibfield  {journal}
  {\bibinfo  {journal} {Journal of Machine Learning Research}\ }\textbf
  {\bibinfo {volume} {13}},\ \bibinfo {pages} {723} (\bibinfo {year}
  {2012})}\BibitemShut {NoStop}%
\bibitem [{\citenamefont {Muandet}\ \emph {et~al.}(2017)\citenamefont
  {Muandet}, \citenamefont {Fukumizu}, \citenamefont {Sriperumbudur},\ and\
  \citenamefont {Sch\"{o}lkopf}}]{muandet2017}%
  \BibitemOpen
  \bibfield  {author} {\bibinfo {author} {\bibfnamefont {K.}~\bibnamefont
  {Muandet}}, \bibinfo {author} {\bibfnamefont {K.}~\bibnamefont {Fukumizu}},
  \bibinfo {author} {\bibfnamefont {B.}~\bibnamefont {Sriperumbudur}},\ and\
  \bibinfo {author} {\bibfnamefont {B.}~\bibnamefont {Sch\"{o}lkopf}},\
  }\bibfield  {title} {\bibinfo {title} {Kernel mean embedding of
  distributions: A review and beyond},\ }\href
  {https://doi.org/10.1561/2200000060} {\bibfield  {journal} {\bibinfo
  {journal} {Found. Trends Mach. Learn.}\ }\textbf {\bibinfo {volume} {10}},\
  \bibinfo {pages} {1} (\bibinfo {year} {2017})}\BibitemShut {NoStop}%
\bibitem [{\citenamefont {Sriperumbudur}\ \emph {et~al.}(2010)\citenamefont
  {Sriperumbudur}, \citenamefont {Gretton}, \citenamefont {Fukumizu},
  \citenamefont {Sch\"{o}lkopf},\ and\ \citenamefont
  {Lanckriet}}]{sriperumbudur2010}%
  \BibitemOpen
  \bibfield  {author} {\bibinfo {author} {\bibfnamefont {B.~K.}\ \bibnamefont
  {Sriperumbudur}}, \bibinfo {author} {\bibfnamefont {A.}~\bibnamefont
  {Gretton}}, \bibinfo {author} {\bibfnamefont {K.}~\bibnamefont {Fukumizu}},
  \bibinfo {author} {\bibfnamefont {B.}~\bibnamefont {Sch\"{o}lkopf}},\ and\
  \bibinfo {author} {\bibfnamefont {G.~R.}\ \bibnamefont {Lanckriet}},\
  }\bibfield  {title} {\bibinfo {title} {Hilbert space embeddings and metrics
  on probability measures},\ }\href@noop {} {\bibfield  {journal} {\bibinfo
  {journal} {J. Mach. Learn. Res.}\ }\textbf {\bibinfo {volume} {11}},\
  \bibinfo {pages} {1517} (\bibinfo {year} {2010})}\BibitemShut {NoStop}%
\bibitem [{\citenamefont {Li}\ \emph {et~al.}(2017)\citenamefont {Li},
  \citenamefont {Chang}, \citenamefont {Cheng}, \citenamefont {Yang},\ and\
  \citenamefont {Póczos}}]{li2017mmdgan}%
  \BibitemOpen
  \bibfield  {author} {\bibinfo {author} {\bibfnamefont {C.-L.}\ \bibnamefont
  {Li}}, \bibinfo {author} {\bibfnamefont {W.-C.}\ \bibnamefont {Chang}},
  \bibinfo {author} {\bibfnamefont {Y.}~\bibnamefont {Cheng}}, \bibinfo
  {author} {\bibfnamefont {Y.}~\bibnamefont {Yang}},\ and\ \bibinfo {author}
  {\bibfnamefont {B.}~\bibnamefont {Póczos}},\ }\bibfield  {title} {\bibinfo
  {title} {{MMD GAN}: Towards deeper understanding of moment matching
  network},\ }in\ \href
  {https://papers.nips.cc/paper_files/paper/2017/hash/dfd7468ac613286cdbb40872c8ef3b06-Abstract.html}
  {\emph {\bibinfo {booktitle} {Advances in Neural Information Processing
  Systems 30}}}\ (\bibinfo  {publisher} {Curran Associates, Inc.},\ \bibinfo
  {year} {2017})\ pp.\ \bibinfo {pages} {2203--2213}\BibitemShut {NoStop}%
\bibitem [{\citenamefont {Li}\ \emph {et~al.}(2015)\citenamefont {Li},
  \citenamefont {Swersky},\ and\ \citenamefont {Zemel}}]{li2015gmmn}%
  \BibitemOpen
  \bibfield  {author} {\bibinfo {author} {\bibfnamefont {Y.}~\bibnamefont
  {Li}}, \bibinfo {author} {\bibfnamefont {K.}~\bibnamefont {Swersky}},\ and\
  \bibinfo {author} {\bibfnamefont {R.}~\bibnamefont {Zemel}},\ }\bibfield
  {title} {\bibinfo {title} {Generative moment matching networks},\ }in\
  \href@noop {} {\emph {\bibinfo {booktitle} {Proceedings of the 32nd
  International Conference on Machine Learning - Volume 37}}},\ \bibinfo
  {series and number} {ICML'15}\ (\bibinfo  {publisher} {JMLR.org},\ \bibinfo
  {year} {2015})\ p.\ \bibinfo {pages} {1718–1727}\BibitemShut {NoStop}%
\bibitem [{\citenamefont {Grant}\ \emph {et~al.}(2019)\citenamefont {Grant},
  \citenamefont {Wossnig}, \citenamefont {Ostaszewski},\ and\ \citenamefont
  {Benedetti}}]{grant2019}%
  \BibitemOpen
  \bibfield  {author} {\bibinfo {author} {\bibfnamefont {E.}~\bibnamefont
  {Grant}}, \bibinfo {author} {\bibfnamefont {L.}~\bibnamefont {Wossnig}},
  \bibinfo {author} {\bibfnamefont {M.}~\bibnamefont {Ostaszewski}},\ and\
  \bibinfo {author} {\bibfnamefont {M.}~\bibnamefont {Benedetti}},\ }\bibfield
  {title} {\bibinfo {title} {An initialization strategy for addressing barren
  plateaus in parametrized quantum circuits},\ }\href
  {https://doi.org/10.22331/q-2019-12-09-214} {\bibfield  {journal} {\bibinfo
  {journal} {{Quantum}}\ }\textbf {\bibinfo {volume} {3}},\ \bibinfo {pages}
  {214} (\bibinfo {year} {2019})}\BibitemShut {NoStop}%
\bibitem [{\citenamefont {Zhang}\ \emph {et~al.}(2022)\citenamefont {Zhang},
  \citenamefont {Liu}, \citenamefont {Hsieh},\ and\ \citenamefont
  {Tao}}]{zhang2022}%
  \BibitemOpen
  \bibfield  {author} {\bibinfo {author} {\bibfnamefont {K.}~\bibnamefont
  {Zhang}}, \bibinfo {author} {\bibfnamefont {L.}~\bibnamefont {Liu}}, \bibinfo
  {author} {\bibfnamefont {M.-H.}\ \bibnamefont {Hsieh}},\ and\ \bibinfo
  {author} {\bibfnamefont {D.}~\bibnamefont {Tao}},\ }\bibfield  {title}
  {\bibinfo {title} {Escaping from the barren plateau via gaussian
  initializations in deep variational quantum circuits},\ }in\ \href@noop {}
  {\emph {\bibinfo {booktitle} {Proceedings of the 36th International
  Conference on Neural Information Processing Systems}}},\ \bibinfo {series and
  number} {NIPS '22}\ (\bibinfo  {publisher} {Curran Associates Inc.},\
  \bibinfo {address} {Red Hook, NY, USA},\ \bibinfo {year} {2022})\BibitemShut
  {NoStop}%
\bibitem [{\citenamefont {Verdon}\ \emph {et~al.}(2019)\citenamefont {Verdon},
  \citenamefont {Broughton}, \citenamefont {McClean}, \citenamefont {Sung},
  \citenamefont {Babbush}, \citenamefont {Jiang}, \citenamefont {Neven},\ and\
  \citenamefont {Mohseni}}]{verdon2019learning}%
  \BibitemOpen
  \bibfield  {author} {\bibinfo {author} {\bibfnamefont {G.}~\bibnamefont
  {Verdon}}, \bibinfo {author} {\bibfnamefont {M.}~\bibnamefont {Broughton}},
  \bibinfo {author} {\bibfnamefont {J.~R.}\ \bibnamefont {McClean}}, \bibinfo
  {author} {\bibfnamefont {K.~J.}\ \bibnamefont {Sung}}, \bibinfo {author}
  {\bibfnamefont {R.}~\bibnamefont {Babbush}}, \bibinfo {author} {\bibfnamefont
  {Z.}~\bibnamefont {Jiang}}, \bibinfo {author} {\bibfnamefont
  {H.}~\bibnamefont {Neven}},\ and\ \bibinfo {author} {\bibfnamefont
  {M.}~\bibnamefont {Mohseni}},\ }\href {https://arxiv.org/abs/1907.05415}
  {\bibinfo {title} {Learning to learn with quantum neural networks via
  classical neural networks}} (\bibinfo {year} {2019}),\ \Eprint
  {https://arxiv.org/abs/1907.05415} {arXiv:1907.05415 [quant-ph]} \BibitemShut
  {NoStop}%
\bibitem [{\citenamefont {Ng}\ \emph {et~al.}(2001)\citenamefont {Ng},
  \citenamefont {Jordan},\ and\ \citenamefont {Weiss}}]{ng2001spectral}%
  \BibitemOpen
  \bibfield  {author} {\bibinfo {author} {\bibfnamefont {A.~Y.}\ \bibnamefont
  {Ng}}, \bibinfo {author} {\bibfnamefont {M.~I.}\ \bibnamefont {Jordan}},\
  and\ \bibinfo {author} {\bibfnamefont {Y.}~\bibnamefont {Weiss}},\ }\bibfield
   {title} {\bibinfo {title} {On spectral clustering: Analysis and an
  algorithm},\ }in\ \href@noop {} {\emph {\bibinfo {booktitle} {Advances in
  Neural Information Processing Systems}}},\ Vol.~\bibinfo {volume} {14}\
  (\bibinfo {year} {2001})\ pp.\ \bibinfo {pages} {849--856}\BibitemShut
  {NoStop}%
\bibitem [{\citenamefont {von Luxburg}(2007)}]{vonluxburg2007}%
  \BibitemOpen
  \bibfield  {author} {\bibinfo {author} {\bibfnamefont {U.}~\bibnamefont {von
  Luxburg}},\ }\bibfield  {title} {\bibinfo {title} {A tutorial on spectral
  clustering},\ }\href {https://doi.org/10.1007/s11222-007-9033-z} {\bibfield
  {journal} {\bibinfo  {journal} {Statistics and Computing}\ }\textbf {\bibinfo
  {volume} {17}},\ \bibinfo {pages} {395} (\bibinfo {year} {2007})}\BibitemShut
  {NoStop}%
\bibitem [{\citenamefont {Arrasmith}\ \emph {et~al.}(2022)\citenamefont
  {Arrasmith}, \citenamefont {Holmes}, \citenamefont {Cerezo},\ and\
  \citenamefont {Coles}}]{arrasmith2022}%
  \BibitemOpen
  \bibfield  {author} {\bibinfo {author} {\bibfnamefont {A.}~\bibnamefont
  {Arrasmith}}, \bibinfo {author} {\bibfnamefont {Z.}~\bibnamefont {Holmes}},
  \bibinfo {author} {\bibfnamefont {M.}~\bibnamefont {Cerezo}},\ and\ \bibinfo
  {author} {\bibfnamefont {P.~J.}\ \bibnamefont {Coles}},\ }\bibfield  {title}
  {\bibinfo {title} {Equivalence of quantum barren plateaus to cost
  concentration and narrow gorges},\ }\href
  {https://doi.org/10.1088/2058-9565/ac7d06} {\bibfield  {journal} {\bibinfo
  {journal} {Quantum Science and Technology}\ }\textbf {\bibinfo {volume}
  {7}},\ \bibinfo {pages} {045015} (\bibinfo {year} {2022})}\BibitemShut
  {NoStop}%
\bibitem [{\citenamefont {Ghosh}\ \emph {et~al.}(2018)\citenamefont {Ghosh},
  \citenamefont {Kulharia}, \citenamefont {Namboodiri}, \citenamefont {Torr},\
  and\ \citenamefont {Dokania}}]{Ghosh_2018_CVPR}%
  \BibitemOpen
  \bibfield  {author} {\bibinfo {author} {\bibfnamefont {A.}~\bibnamefont
  {Ghosh}}, \bibinfo {author} {\bibfnamefont {V.}~\bibnamefont {Kulharia}},
  \bibinfo {author} {\bibfnamefont {V.~P.}\ \bibnamefont {Namboodiri}},
  \bibinfo {author} {\bibfnamefont {P.~H.~S.}\ \bibnamefont {Torr}},\ and\
  \bibinfo {author} {\bibfnamefont {P.~K.}\ \bibnamefont {Dokania}},\
  }\bibfield  {title} {\bibinfo {title} {Multi-agent diverse generative
  adversarial networks},\ }in\ \href {https://doi.org/10.1109/CVPR.2018.00888}
  {\emph {\bibinfo {booktitle} {Proceedings of the IEEE Conference on Computer
  Vision and Pattern Recognition (CVPR)}}}\ (\bibinfo {year} {2018})\ pp.\
  \bibinfo {pages} {8513--8521}\BibitemShut {NoStop}%
\bibitem [{\citenamefont {Lecun}\ \emph {et~al.}(1998)\citenamefont {Lecun},
  \citenamefont {Bottou}, \citenamefont {Bengio},\ and\ \citenamefont
  {Haffner}}]{lecun1998}%
  \BibitemOpen
  \bibfield  {author} {\bibinfo {author} {\bibfnamefont {Y.}~\bibnamefont
  {Lecun}}, \bibinfo {author} {\bibfnamefont {L.}~\bibnamefont {Bottou}},
  \bibinfo {author} {\bibfnamefont {Y.}~\bibnamefont {Bengio}},\ and\ \bibinfo
  {author} {\bibfnamefont {P.}~\bibnamefont {Haffner}},\ }\bibfield  {title}
  {\bibinfo {title} {Gradient-based learning applied to document recognition},\
  }\href {https://doi.org/10.1109/5.726791} {\bibfield  {journal} {\bibinfo
  {journal} {Proceedings of the IEEE}\ }\textbf {\bibinfo {volume} {86}},\
  \bibinfo {pages} {2278} (\bibinfo {year} {1998})}\BibitemShut {NoStop}%
\bibitem [{\citenamefont {Kingma}\ and\ \citenamefont
  {Ba}(2015)}]{kingma2015adam}%
  \BibitemOpen
  \bibfield  {author} {\bibinfo {author} {\bibfnamefont {D.~P.}\ \bibnamefont
  {Kingma}}\ and\ \bibinfo {author} {\bibfnamefont {J.}~\bibnamefont {Ba}},\
  }\bibfield  {title} {\bibinfo {title} {Adam: A method for stochastic
  optimization},\ }in\ \href@noop {} {\emph {\bibinfo {booktitle} {Proceedings
  of the 3rd International Conference on Learning Representations ({ICLR})}}}\
  (\bibinfo {year} {2015})\BibitemShut {NoStop}%
\bibitem [{\citenamefont {Armengol}\ and\ \citenamefont
  {Bowles}(2025)}]{iqpopt}%
  \BibitemOpen
  \bibfield  {author} {\bibinfo {author} {\bibfnamefont {E.}~\bibnamefont
  {Armengol}}\ and\ \bibinfo {author} {\bibfnamefont {J.}~\bibnamefont
  {Bowles}},\ }\href {https://arxiv.org/abs/2501.04776} {\bibinfo {title}
  {{IQPopt}: Fast optimization of instantaneous quantum polynomial circuits in
  {JAX}}} (\bibinfo {year} {2025}),\ \Eprint {https://arxiv.org/abs/2501.04776}
  {arXiv:2501.04776 [quant-ph]} \BibitemShut {NoStop}%
\bibitem [{\citenamefont {K{\"u}bler}\ \emph {et~al.}(2021)\citenamefont
  {K{\"u}bler}, \citenamefont {Buchholz},\ and\ \citenamefont
  {Sch{\"o}lkopf}}]{kubler2021inductive}%
  \BibitemOpen
  \bibfield  {author} {\bibinfo {author} {\bibfnamefont {J.~M.}\ \bibnamefont
  {K{\"u}bler}}, \bibinfo {author} {\bibfnamefont {S.}~\bibnamefont
  {Buchholz}},\ and\ \bibinfo {author} {\bibfnamefont {B.}~\bibnamefont
  {Sch{\"o}lkopf}},\ }\bibfield  {title} {\bibinfo {title} {The inductive bias
  of quantum kernels},\ }in\ \href {https://arxiv.org/abs/2106.03747} {\emph
  {\bibinfo {booktitle} {Advances in Neural Information Processing Systems
  34}}}\ (\bibinfo  {publisher} {Curran Associates, Inc.},\ \bibinfo {year}
  {2021})\ \Eprint {https://arxiv.org/abs/2106.03747} {arXiv:2106.03747
  [quant-ph]} \BibitemShut {NoStop}%
\bibitem [{\citenamefont {Shazeer}\ \emph {et~al.}(2017)\citenamefont
  {Shazeer}, \citenamefont {Mirhoseini}, \citenamefont {Maziarz}, \citenamefont
  {Davis}, \citenamefont {Le}, \citenamefont {Hinton},\ and\ \citenamefont
  {Dean}}]{shazeer2017outrageously}%
  \BibitemOpen
  \bibfield  {author} {\bibinfo {author} {\bibfnamefont {N.}~\bibnamefont
  {Shazeer}}, \bibinfo {author} {\bibfnamefont {A.}~\bibnamefont {Mirhoseini}},
  \bibinfo {author} {\bibfnamefont {K.}~\bibnamefont {Maziarz}}, \bibinfo
  {author} {\bibfnamefont {A.}~\bibnamefont {Davis}}, \bibinfo {author}
  {\bibfnamefont {Q.}~\bibnamefont {Le}}, \bibinfo {author} {\bibfnamefont
  {G.}~\bibnamefont {Hinton}},\ and\ \bibinfo {author} {\bibfnamefont
  {J.}~\bibnamefont {Dean}},\ }\bibfield  {title} {\bibinfo {title}
  {Outrageously large neural networks: The sparsely-gated mixture-of-experts
  layer},\ }in\ \href {https://arxiv.org/abs/1701.06538} {\emph {\bibinfo
  {booktitle} {International Conference on Learning Representations}}}\
  (\bibinfo {year} {2017})\ \Eprint {https://arxiv.org/abs/1701.06538}
  {arXiv:1701.06538 [cs.LG]} \BibitemShut {NoStop}%
\bibitem [{\citenamefont {Huang}\ \emph {et~al.}(2026)\citenamefont {Huang},
  \citenamefont {Maxwell}, \citenamefont {Belis}, \citenamefont {Peters},
  \citenamefont {Pye}, \citenamefont {Jahangiri},\ and\ \citenamefont
  {Bowles}}]{huang2026spectral}%
  \BibitemOpen
  \bibfield  {author} {\bibinfo {author} {\bibfnamefont {A.}~\bibnamefont
  {Huang}}, \bibinfo {author} {\bibfnamefont {W.}~\bibnamefont {Maxwell}},
  \bibinfo {author} {\bibfnamefont {V.}~\bibnamefont {Belis}}, \bibinfo
  {author} {\bibfnamefont {E.}~\bibnamefont {Peters}}, \bibinfo {author}
  {\bibfnamefont {J.}~\bibnamefont {Pye}}, \bibinfo {author} {\bibfnamefont
  {S.}~\bibnamefont {Jahangiri}},\ and\ \bibinfo {author} {\bibfnamefont
  {J.}~\bibnamefont {Bowles}},\ }\href {https://arxiv.org/abs/2607.06675}
  {\bibinfo {title} {Spectral born machines: Classically trainable quantum
  generative models for discrete data}} (\bibinfo {year} {2026}),\ \Eprint
  {https://arxiv.org/abs/2607.06675} {arXiv:2607.06675 [quant-ph]} \BibitemShut
  {NoStop}%
\bibitem [{\citenamefont {Maslov}\ \emph {et~al.}(2024)\citenamefont {Maslov},
  \citenamefont {Bravyi}, \citenamefont {Tripier}, \citenamefont {Maksymov},\
  and\ \citenamefont {Latone}}]{maslov2024fast}%
  \BibitemOpen
  \bibfield  {author} {\bibinfo {author} {\bibfnamefont {D.}~\bibnamefont
  {Maslov}}, \bibinfo {author} {\bibfnamefont {S.}~\bibnamefont {Bravyi}},
  \bibinfo {author} {\bibfnamefont {F.}~\bibnamefont {Tripier}}, \bibinfo
  {author} {\bibfnamefont {A.}~\bibnamefont {Maksymov}},\ and\ \bibinfo
  {author} {\bibfnamefont {J.}~\bibnamefont {Latone}},\ }\href
  {https://arxiv.org/abs/2402.03211} {\bibinfo {title} {Fast classical
  simulation of {Harvard/QuEra IQP} circuits}} (\bibinfo {year} {2024}),\
  \Eprint {https://arxiv.org/abs/2402.03211} {arXiv:2402.03211 [quant-ph]}
  \BibitemShut {NoStop}%
\bibitem [{\citenamefont {Scriva}\ \emph {et~al.}(2023)\citenamefont {Scriva},
  \citenamefont {Costa}, \citenamefont {McNaughton},\ and\ \citenamefont
  {Pilati}}]{scriva2023}%
  \BibitemOpen
  \bibfield  {author} {\bibinfo {author} {\bibfnamefont {G.}~\bibnamefont
  {Scriva}}, \bibinfo {author} {\bibfnamefont {E.}~\bibnamefont {Costa}},
  \bibinfo {author} {\bibfnamefont {B.}~\bibnamefont {McNaughton}},\ and\
  \bibinfo {author} {\bibfnamefont {S.}~\bibnamefont {Pilati}},\ }\bibfield
  {title} {\bibinfo {title} {Accelerating equilibrium spin-glass simulations
  using quantum annealers via generative deep learning},\ }\href
  {https://doi.org/10.21468/SciPostPhys.15.1.018} {\bibfield  {journal}
  {\bibinfo  {journal} {SciPost Phys.}\ }\textbf {\bibinfo {volume} {15}},\
  \bibinfo {pages} {018} (\bibinfo {year} {2023})}\BibitemShut {NoStop}%
\bibitem [{\citenamefont {Demidik}\ \emph {et~al.}(2025)\citenamefont
  {Demidik}, \citenamefont {T{\"u}ys{\"u}z}, \citenamefont {Piatkowski},
  \citenamefont {Grossi},\ and\ \citenamefont
  {Jansen}}]{demidik2025expressive}%
  \BibitemOpen
  \bibfield  {author} {\bibinfo {author} {\bibfnamefont {M.}~\bibnamefont
  {Demidik}}, \bibinfo {author} {\bibfnamefont {C.}~\bibnamefont
  {T{\"u}ys{\"u}z}}, \bibinfo {author} {\bibfnamefont {N.}~\bibnamefont
  {Piatkowski}}, \bibinfo {author} {\bibfnamefont {M.}~\bibnamefont {Grossi}},\
  and\ \bibinfo {author} {\bibfnamefont {K.}~\bibnamefont {Jansen}},\
  }\bibfield  {title} {\bibinfo {title} {Expressive equivalence of classical
  and quantum restricted boltzmann machines},\ }\href
  {https://doi.org/10.1038/s42005-025-02353-1} {\bibfield  {journal} {\bibinfo
  {journal} {Communications Physics}\ }\textbf {\bibinfo {volume} {8}},\
  \bibinfo {pages} {413} (\bibinfo {year} {2025})}\BibitemShut {NoStop}%
\bibitem [{\citenamefont {Tieleman}(2008)}]{tieleman2008pcd}%
  \BibitemOpen
  \bibfield  {author} {\bibinfo {author} {\bibfnamefont {T.}~\bibnamefont
  {Tieleman}},\ }\bibfield  {title} {\bibinfo {title} {Training restricted
  {B}oltzmann machines using approximations to the likelihood gradient},\ }in\
  \href {https://doi.org/10.1145/1390156.1390290} {\emph {\bibinfo {booktitle}
  {Proceedings of the 25th International Conference on Machine Learning}}}\
  (\bibinfo {year} {2008})\ pp.\ \bibinfo {pages} {1064--1071}\BibitemShut
  {NoStop}%
\bibitem [{\citenamefont {Hinton}(2012)}]{hinton2012practical}%
  \BibitemOpen
  \bibfield  {author} {\bibinfo {author} {\bibfnamefont {G.~E.}\ \bibnamefont
  {Hinton}},\ }\bibinfo {title} {A practical guide to training restricted
  boltzmann machines},\ in\ \href
  {https://doi.org/10.1007/978-3-642-35289-8_32} {\emph {\bibinfo {booktitle}
  {Neural Networks: Tricks of the Trade: Second Edition}}},\ \bibinfo {editor}
  {edited by\ \bibinfo {editor} {\bibfnamefont {G.}~\bibnamefont {Montavon}},
  \bibinfo {editor} {\bibfnamefont {G.~B.}\ \bibnamefont {Orr}},\ and\ \bibinfo
  {editor} {\bibfnamefont {K.-R.}\ \bibnamefont {M{\"u}ller}}}\ (\bibinfo
  {publisher} {Springer Berlin Heidelberg},\ \bibinfo {address} {Berlin,
  Heidelberg},\ \bibinfo {year} {2012})\ pp.\ \bibinfo {pages}
  {599--619}\BibitemShut {NoStop}%
\end{thebibliography}%
